\documentclass[11pt]{article}
\usepackage[utf8]{inputenc}
\usepackage[T1]{fontenc}  

\usepackage{libertine}
\usepackage[scaled=0.96]{zi4} 

\usepackage{amsmath, amssymb, amsthm, thm-restate}
\usepackage[libertine]{newtxmath} 

\usepackage[DIV=12,letterpaper]{typearea} 

\usepackage{microtype}

\usepackage[linesnumbered,ruled,vlined]{algorithm2e}
\usepackage[noend]{algpseudocode}
\usepackage[labelfont=bf]{caption}
\usepackage[nottoc,notlot,notlof]{tocbibind}

\usepackage{xcolor}

\usepackage[ocgcolorlinks]{hyperref} 

\colorlet{DarkRed}{red!50!black}
\colorlet{DarkGreen}{green!50!black}
\colorlet{DarkBlue}{blue!50!black}

\hypersetup{
	linkcolor = DarkRed,
	citecolor = DarkGreen,
	urlcolor = DarkBlue,
	bookmarksnumbered = true,
	linktocpage = true
}

\allowdisplaybreaks[1]

\newtheorem{theorem}{Theorem}
\newtheorem{lemma}{Lemma}[section]
\newtheorem*{lemma*}{Lemma}
\newtheorem{corollary}[lemma]{Corollary}
\newtheorem*{corollary*}{Corollary}
\newtheorem{claim}[lemma]{Claim}

\theoremstyle{definition}

\newtheorem*{theorem*}{Theorem}
\newtheorem{definition}[lemma]{Definition}

\newtheorem*{rem*}{Remark}

\newcommand{\abs}[1]{\left| #1 \right|}
\newcommand{\norm}[1]{\left\| #1 \right\|}

\newcommand{\mA}{\boldsymbol{\mathit{A}}}
\newcommand{\mB}{\boldsymbol{\mathit{B}}}
\newcommand{\mC}{\boldsymbol{\mathit{C}}}
\newcommand{\mD}{\boldsymbol{\mathit{D}}}
\newcommand{\mI}{\boldsymbol{\mathit{I}}}
\newcommand{\mJ}{\boldsymbol{\mathit{J}}}
\newcommand{\mL}{\boldsymbol{\mathit{L}}}
\newcommand{\mM}{\boldsymbol{\mathit{M}}}
\newcommand{\mP}{\boldsymbol{\mathit{P}}}
\newcommand{\mQ}{\boldsymbol{\mathit{Q}}}
\newcommand{\mR}{\boldsymbol{\mathit{R}}}
\newcommand{\mS}{\boldsymbol{\mathit{S}}}
\newcommand{\mU}{\boldsymbol{\mathit{U}}}
\newcommand{\mV}{\boldsymbol{\mathit{V}}}
\newcommand{\mW}{\boldsymbol{\mathit{W}}}
\newcommand{\mX}{\boldsymbol{\mathit{X}}}
\newcommand{\mY}{\boldsymbol{\mathit{Y}}}
\newcommand{\mZ}{\boldsymbol{\mathit{Z}}}

\newcommand{\hmX}{\boldsymbol{\widehat{\mathit{X}}}}
\newcommand{\hmY}{\boldsymbol{\widehat{\mathit{Y}}}}
\newcommand{\tmM}{\boldsymbol{\widetilde{\mathit{M}}}}

\newcommand{\E}{\mathbb{E}}
\newcommand{\R}{\mathbb{R}}
\newcommand{\assign}{\leftarrow}
\newcommand{\Tr}{\mathrm{Tr}}

\newcommand{\tomato}{\top}

\newcommand{\res}{\mathrm{res}}
\newcommand{\lev}{\mathrm{lev}}

\newcommand{\clocal}{C_{\mathrm{local}}}
\newcommand{\mSC}{\mathbf{SC}}

\newcommand{\congest}{\textnormal{{\textsf{CONGEST}}}\xspace}
\newcommand{\FindSteady}{\textnormal{\textsc{FindSteady}}\xspace}
\newcommand{\Recover}{\textnormal{\textsc{Recover}}\xspace}
\newcommand{\DiffApx}{\textnormal{\textsc{DiffApx}}\xspace}
\newcommand{\ColumnApx}{\textnormal{\textsc{ColumnApx}}\xspace}
\newcommand{\LevApx}{\textnormal{\textsc{LevApx}}\xspace}
\newcommand{\Sparsify}{\textnormal{\textsc{SpectralSparsifyKX}}\xspace}
\newcommand{\RandWalkSchur}{\textnormal{\textsc{RandWalkSchur}}\xspace}
\newcommand{\DDSubset}{\textnormal{\textsc{DDSubset}}\xspace}
\newcommand{\Jacobi}{\textnormal{\textsc{Jacobi}}\xspace}
\newcommand{\RandomWalk}{\textnormal{\textsc{RandomWalk}}\xspace}
\newcommand{\RandomWalkSchur}{\textnormal{\textsc{RandomWalkSchur}}\xspace}
\newcommand{\Split}{\textnormal{\textsc{Split}}\xspace}
\newcommand{\Unsplit}{\textnormal{\textsc{Unsplit}}\xspace}
\newcommand{\ApproxSC}{\textnormal{\textsc{ApproxSC}}\xspace}
\newcommand{\solve}{\textnormal{\textsc{Solve}}\xspace}
\newcommand{\Solve}{\textnormal{\textsc{Solve}}\xspace}
\newcommand{\BuildChain}{\textnormal{\textsc{BuildChain}}\xspace}
\newcommand{\Eliminate}{\textnormal{\textsc{Eliminate}}\xspace}

\newcommand{\Ultrasparsify}{\textnormal{\textsc{UltraSparsify}}\xspace}

\newcommand{\PseudoinverseMulti}{\textnormal{\textsc{PseudoinverseMulti}}\xspace}
\newcommand{\MaxFlow}{\textnormal{\textsc{MaxFlow}}\xspace}
\newcommand{\Augmentation}{\textnormal{\textsc{Augmentation}}\xspace}
\newcommand{\Fixing}{\textnormal{\textsc{Fixing}}\xspace}
\newcommand{\Boosting}{\textnormal{\textsc{Boosting}}\xspace}
\newcommand{\FlowRounding}{\textnormal{\textsc{FlowRounding}}\xspace}
\newcommand{\MinCostFlow}{\textnormal{\textsc{MinCostFlow}}\xspace}
\newcommand{\Initialization}{\textnormal{\textsc{Initialization}}\xspace}
\newcommand{\Perturbation}{\textnormal{\textsc{Perturbation}}\xspace}
\newcommand{\Progress}{\textnormal{\textsc{Progress}}\xspace}
\newcommand{\Repairing}{\textnormal{\textsc{Repairing}}\xspace}
\newcommand{\ShortestPaths}{\textnormal{\textsc{ShortestPaths}}\xspace}
\newcommand{\SpecialVertices}{\textnormal{\textsc{SpecialVertices}}\xspace}
\newcommand{\SplitGraph}{\textnormal{\textsc{SplitGraph}}\xspace}

\newcommand{\poly}{\mathrm{poly}}
\newcommand{\load}{\mathrm{cong}}

\newcommand{\eps}{\epsilon}
\newcommand{\T}{\mathcal{T}}

\renewcommand{\d}{\delta}
\newcommand{\defeq}{\stackrel{\mathrm{\scriptscriptstyle def}}{=}}
\newcommand{\bs}{\backslash}
\newcommand{\pe}{\preceq}

\renewcommand{\O}{\widetilde{O}}

\author{
Sebastian Forster\thanks{University of Salzburg, Austria.}
\and
Gramoz Goranci\thanks{University of Toronto, Canada.}
\and
Yang P. Liu\thanks{Stanford University, USA.}
\and
Richard Peng\thanks{Georgia Institute of Technology, USA.}
\and
Xiaorui Sun\thanks{University of Illinois at Chicago, USA.}
\and
Mingquan Ye\thanks{University of Illinois at Chicago, USA.}
}

\title{
Minor Sparsifiers and the Distributed Laplacian Paradigm\thanks{Accepted to the 62nd Annual Symposium on Foundations of Computer Science (FOCS 2021)}
}

\date{}

\hypersetup{
	pdftitle = {Minor Sparsifiers and the Distributed Laplacian Paradigm},
	pdfauthor = {Sebastian Forster, Gramoz Goranci, Yang P. Liu, Richard Peng, Xiaorui Sun, Mingquan Ye}
}

\begin{document}
\maketitle
\thispagestyle{empty}
\begin{abstract}
We study distributed algorithms built around minor-based vertex sparsifiers, and give the first algorithm in the $\textsf{CONGEST}$ model for solving linear systems in graph Laplacian matrices to high accuracy. Our Laplacian solver has a round complexity of $O(n^{o(1)}(\sqrt{n}+D))$, and thus almost matches the lower bound of $\widetilde{\Omega}(\sqrt{n}+D)$, where $n$ is the number of nodes in the network and $D$ is its diameter.

We show that our distributed solver yields new sublinear round algorithms for several cornerstone problems in combinatorial optimization. This is achieved by leveraging the powerful algorithmic framework of Interior Point Methods (IPMs) and the Laplacian paradigm in the context of distributed graph algorithms, which entails numerically solving optimization problems on graphs via a series of Laplacian systems. Problems that benefit from our distributed algorithmic paradigm include exact mincost flow, negative weight shortest paths, maxflow, and bipartite matching on sparse directed graphs. For the maxflow problem, this is the first exact distributed algorithm that applies to directed graphs, while the previous work by [Ghaffari et al.~SICOMP'18] considered the approximate setting and works only for undirected graphs. For the mincost flow and the negative weight shortest path problems, our results constitute the first exact distributed algorithms running in a sublinear number of rounds. Given that the hybrid between IPMs and the Laplacian paradigm has proven useful for tackling numerous optimization problems in the centralized setting, we believe that our distributed solver will find future applications.

At the heart of our distributed Laplacian solver is the notion of spectral subspace sparsifiers of [Li, Schild FOCS'18]. We present a nontrivial distributed implementation of their construction by (i) giving a parallel variant of their algorithm that avoids the sampling of random spanning trees and uses approximate leverage scores instead, and (ii) showing that the algorithm still produces a high-quality subspace spectral sparsifier by carefully setting up and analyzing matrix martingales. Combining this vertex reduction recursively with both tree and elimination-based preconditioners leads to our algorithm for solving Laplacian systems. The construction of the elimination-based preconditioners is based on computing short random walks, and we introduce a new technique for reducing the congestion incurred by the simulation of these walks on weighted graphs.
\end{abstract}

\newpage
\tableofcontents
\newpage

\section{Introduction}
\label{sec:Intro}

The steady growth of data makes it
increasingly important to control and reduce the communication of algorithms.
The \congest model~\cite{Peleg00} is a widely studied model for
low communication algorithms on large graphs and sparse matrices. In this model, each vertex/variable occupies a separate machine,
and communicates in synchronous rounds by sending messages of length $O(\log n)$ to its neighbors given by the edges of the underlying graph.
This bandwidth restriction implies a polynomial lower bound in the round complexity for many fundamental graph problems~\cite{PR00,E06,SHKKNPPW12}.
While early work on efficient algorithms in this model has focused on the minimum spanning tree problem~\cite{GHS83,GKP98,KP98}, extensive work over the past few years has led to efficient algorithms for several more fundamental graph problems, such as approximate and exact single-source shortest paths~\cite{N14,HKN16,E20,GL18,FN18,CM20}, approximate and exact all-pairs shortest paths~\cite{HW12,HNS17,ARKP18,EN18,LPP19,AR19,BN19,AR20}, approximate and exact minimum cut~\cite{GK13,NS14,DHNS19,GNT20,DEMN20}, approximate maximum flow~\cite{GKKLP15}, bipartite maximum matching~\cite{AKO18}, triangle counting~\cite{IG17,CPZ19,CS19}, and single-source reachability~\cite{GU15,LJS19}.

A major development in sequential and parallel graph algorithms is the development of hybrid algorithms that combine numerical and combinatorial building blocks. This line of work was initiated by the seminal work of Spielman and Teng \cite{ST14}, which showed that a Laplacian linear system on a graph can be solved in nearly linear time. Here, the Laplacian of a weighted undirected graph $G = (V, E)$ is defined as $\mL(G) = \mD(G) - \mA(G)$, where $\mD(G)$ is the diagonal weighted degree matrix, and $\mA(G)$ is the weighted adjacency matrix. Equivalently, if $\vec{w} \in \R^m_{>0}$ are the edge weights,
\[ \mL\left( G \right)_{uv} =
\begin{cases}
\sum_{(u,z) \in E} \vec{w}_{uz} & \text{if $u = v$},\\
-\vec{w}_{uv} & \text{otherwise}.
\end{cases}. \]
Since then, there has been extensive work
towards giving more efficient and simpler Laplacian system solvers sequentially
\cite{KMP10,KMP11,KOSZ13,CKMPPRX14,KS16}, as well as parallel versions
\cite{PS14,KLPSS16}. These results have in turn been used to give the state-of-the-art
runtimes for a variety of graph problems, including exact maximum flows, bipartite
matchings, and mincost flows \cite{M16,LS20,LS20:arxiv,CMSV17,AMV20,matching},
approximate maximum flows \cite{KLOS14,Sherman13,P16}, and approximate parallel
shortest paths \cite{Li20,AndoniSZ20}. Ideas from the latter works have found
application in the distributed setting, giving nearly optimal algorithms for
approximate maxflows \cite{GKKLP15} and approximate single-source shortest paths \cite{BKKL17} in
the \congest~model.

Our main result is an algorithm for solving graph Laplacian linear systems
in the \congest~model in $O(n^{o(1)}(\sqrt{n}+D))$ rounds (Theorem \ref{thm:Main}), where $n$ is the number of nodes in the underling graph and $D$ is its diameter. This nearly matches a lower bound of $\widetilde{\Omega}(\sqrt{n}+D)$, which we show for completeness in Appendix~\ref{sec:LowerBound} by reduction to \cite{SHKKNPPW12}.
\begin{restatable}{theorem}{Main}
\label{thm:Main}
There is an algorithm in the \congest~model that on a weighted graph $G = (V, E, \vec{w})$
with $n$ vertices and diameter $D$, vector $\vec{b}$ on the vertices of $G$, and error $\epsilon < 0.1$, produces in $O(n^{o(1)}(n^{1/2} + D)\log(1/\eps))$ rounds a vector $\vec{x}$  distributed over the vertices such that
\[ \left\|\vec{x} - \mL(G)^{\dag} \vec{b}\right\|_{\mL(G)} \le \epsilon \cdot \left\|\vec{b}\right\|_{\mL(G)}. \]
\end{restatable}
\begin{restatable}{theorem}{LowerBound}
\label{thm:LowerBound}
In the \congest~model of computation, solving Laplacian systems to accuracy $ \epsilon \leq \tfrac{1}{2} $ requires at least $\widetilde{\Omega}(n^{1/2} + D)$ rounds of communication.
\end{restatable}
We give several applications towards designing hybrid algorithms for graph problems in the \congest~model. Specifically, by combining our Laplacian solver with interior point methods \cite{M16,CMSV17}, we obtain the first algorithms for exact computation of maximum flows, bipartite matchings, and negative-weight shortest paths that run in a sublinear number of rounds in the \congest~model on sparse graphs (Section~\ref{sec:Implications}).

At a high level, we build our \congest~model algorithm by first building a parallel/PRAM algorithm for solving Laplacian systems that only works with minors\footnote{In fact, our algorithm deals with $\rho$-minors~(Definition~\ref{def:Minor}), which can be thought of as minors with congestion $\rho$, where $\rho \geq 1$ is a parameter. However, for the sake of simplicity, we refer to them as minors throughout the informal discussions of our techniques in the introduction and overview.} of the
original graph, and show that one round of communication necessary in our algorithm (such as matrix vector multiplication)
between neighbors on a minor can be simulated in the original graph in $\tilde{O}(\sqrt{n}+D)$ rounds.
Previous methods for computing low stretch spanning trees and approximate maximum flows \cite{GKKLP15} use a similar notion
of considering a graph on clusters of nodes in the original graph, however -- to the best of our knowledge -- we are the first to work
explicitly with the notion of minors. We are optimistic that our approach based on minor vertex sparsifiers may
provide a general framework for designing \congest~model algorithms with
near optimal complexities.

The main backbone of our algorithm for solving Laplacian systems that works
with minors only is the parallel Laplacian solver of \cite{KLPSS16}. This solver relies on sparse spectral approximations of the Schur complements of an $n \times n$ matrix, which can be thought of as a smaller matrix that preserves the solutions of linear systems on a subset of coordinates in $[n]$. At a high level, the algorithm eliminates onto (sparse) Schur complements of the original graph while adding edges, leading to graphs that are not minors of the original graph.
To resolve this, a major contribution of this paper
is an efficient parallel algorithm to
construct a spectral sparsifier for a Schur complement which is a minor of the
original graph. While the existence of such a minor spectral sparsifier was known
\cite{LS18}, the algorithm required sampling a random spanning tree, and hence could
not be implemented in parallel. We instead show that a large batch of edges may be
independently sampled at the same time using leverage scores (Definition \ref{def:lev}), providing an arguably
simpler and more direct analysis than \cite{LS18}.

\subsection{Applications to Flow Problems}
We briefly discuss how our Laplacian solver can be applied to achieve results on maximum flow, bipartite matching, mincost flows, and negative weight shortest paths, and compare to previous complexities. We achieve our bounds by combining our Laplacian system solver in Theorem \ref{thm:Main} with recent interior point methods of \cite{M16,CMSV17}.

For unit capacity graphs, the runtimes we achieve in Theorems \ref{thm:max-flow},
\ref{thm:min-cost-flow}, and \ref{thm:shortest-paths} for the maximum flow problem,
mincost flow, and negative weight shortest path problems are
\[ O(m^{3/7+o(1)}(n^{1/2}D^{1/4} + D)). \]
For sparse unweighted graphs with $m = O(n)$, and polynomially small diameter $D =
n^{2/7-\Omega(1)}$, the algorithms in Theorems \ref{thm:max-flow},
\ref{thm:min-cost-flow}, and \ref{thm:shortest-paths} run in a \emph{sublinear} number
of rounds, i.e. $n^{1-\Omega(1)}$ rounds. To our knowledge, these are the first exact
sublinear round algorithms for unit maximum flows, bipartite matchings, and negative
weight shortest paths for any regime of diameter $D$. Our distributed maximum flow algorithm extends to directed graphs, while the previous work by Ghaffari et al.~\cite{GKKLP15} considered the approximate setting and works only for undirected graphs. In fact, for the maximum flow
problem, our results are -- to the best of our knowledge --
the fastest known in the low-diameter regime; see Section \ref{sec:Related} for further discussion.

At a high level, our runtime comes from two pieces. The results of \cite{M16,CMSV17} show that in $\O(m^{3/7})$ rounds of an interior point method, in each round which involves solving a Laplacian system on the underlying graph with edge weights/resistances, we can reduce the amount of residual flow to $\O(m^{3/7})$. The residual flow can routed combinatorially with $\O(m^{3/7})$ rounds of an augmenting paths or shortest paths computation. Therefore, the total number of rounds required to implement the interior point method is $O(m^{3/7+o(1)}(n^{1/2}+D))$ using Theorem \ref{thm:Main}, and the shortest path computations can be done in $\O(m^{3/7}(n^{1/2}D^{1/4} + D))$ rounds using the results of \cite{CM20}. Combining these gives the result.

\subsection{Related Work}
\label{sec:Related}

\paragraph{Distributed Graph Algorithms}

Previous works in distributed algorithms most related to our result and the corresponding techniques are
the algorithms for simulating \textbf{random walks} and generating \textbf{random spanning
trees}~\cite{DNPT13,GB20}. On unweighted, undirected graphs with diameter $D$,
the algorithms by Das Sarma, Nanongkai, Pandurangan, and Tetali~\cite{DNPT13}
generate an $\ell$-step random walk in $\O(\sqrt{\ell D} + D)$ rounds,
and a random spanning tree in $\O(\sqrt{mD})$ rounds, respectively.
There are well known connections between sampling a large number of random
walks and Laplacian solving~\cite{DST17,DGT17}. However, it is not clear how to utilize these methods in the context of our algorithms, since many of the intermediate graph structures we deal with involve dealing with weighted random walks, which in turn leads to congestion issues when trying to simulate these walks in the distributed setting. We discuss how to overcome such obstacles in Section~\ref{sec:Overview}. 

There has also been work in the distributed setting relating to \textbf{spectral graph properties}.
This in particular includes distributed sparsification~\cite{KX16}, PageRank \cite{SMPU15}, Laplacian solvers in well-mixing settings \cite{GB20}, and expander decomposition~\cite{CPZ19,CS19,CS20}.

\textbf{Continuous optimization methods} have been used to give the state-of-the-art distributed algorithms for approximate max-flow~\cite{GKKLP15} and approximate transshipment~\cite{BKKL17}.
Note however that these approximation algorithms are tailored to undirected graphs and their running time depends polynomially on $ 1/\epsilon $ (for a desired accuracy of $ \epsilon $).
Our max-flow routine also works on directed graphs and only depends polylogarithmically on $ 1/\epsilon $, which allows for computing a high-accuracy solution and rounding it to an exact one.
Furthermore, these prior approaches for $\ell_\infty$ and $\ell_1$-norm minimization, respectively, do not carry over to $\ell_2$-norm minimization (as would be needed for solving Laplacian systems) as it is not known how to efficiently sample from a collection of trees when using an $\ell_2$ variant of tree-based graph approximations to build oblivious routing schemes. 

In addition to these works, there are many papers related to the three problems we solve by applying our distributed Laplacian solver.
There have been numerous results on exact and approximate \textbf{shortest path} computation in the past decade~\cite{HW12,N14,HKN16,EN19-SICOMP,E20,HNS17,GL18,ARKP18,EN18,FN18,LPP19,EN19-SPAA,AR19,BN19,AR20,CM20}.
For the single-source shortest paths (SSSP) problem all of these works assume non-negative or positive edge weights.
It is well-known that the SSSP problem in presence of negative edge weights can be solved in $O(n)$ rounds by a variant of the Bellman-Ford algorithm. To the best of our knowledge no algorithm that improves upon this bound has been formulated (or implied) in the \congest~model so far.

For \textbf{distributed computations of maximum flows},
Ghaffari, Karrenbauer, Kuhn, Lenzen, and Patt-Shamir \cite{GKKLP15}
designed an algorithm that returns an $ (1 + \epsilon)$-approximation
in $O( (\sqrt{n} + D) n^{o(1)} \epsilon^{-3} )$ rounds.
In terms of exact algorithms, we are not aware of any paper claiming a sublinear number of rounds in the \congest~model (cf.~\cite{GKKLP15} for a detailed discussion of maximum flow for other distributed models).
To the best of our knowledge, we need to compare ourselves with the following two approaches:
\begin{itemize}
\item The problem can trivially be solved in $ O (m + D) $ rounds by collecting the whole graph topology in a single node and then solving the problem with internal computation.
\item The Ford-Fulkerson algorithm~\cite{FF56} takes $ |f^*| $ iterations (where $
|f^*| $ is the value of a maximum flow) and the running time in each iteration is
dominated by the time needed to perform an $s$-$t$ reachability computation (on a
directed graph).
The latter problem can be solved in $ \tilde O (\sqrt{n} D^{1/4} + D) $~\cite{GU15} or $ \tilde O
(\sqrt{n} + n^{1/3+o(1)} D^{2/3}) $ rounds, respectively, which yields total running time of $
\tilde O (|f^*| (\sqrt{n} D^{1/4} + D)) $ or $
\tilde O (|f^*| (\sqrt{n} + n^{1/3+o(1)} D^{2/3})) $ rounds, respectively.
In unit-capacity (``unweighted'') graphs, where $ |f^*| \leq n $, this gives a total running time of $ \tilde O (n^{3/2} D^{1/4} + nD) $ or $ \tilde O (n^{3/2} + n^{4/3+o(1)} D^{2/3}) $, respectively.
\end{itemize}

Due to a well-known reduction to maximum flow, the \textbf{bipartite maximum matching} problem
is intimately connected to the maximum flow problem.
In the \congest~model, the fastest known algorithm for computing a bipartite maximum
matching (of an unweighted graph) takes $ O (n \log n) $ rounds~\cite{AKO18} -- more
precisely the algorithm takes $ O (s^* \log s^*) $ rounds, where $ s^* $ is the size
of a maximum matching.
Obtaining a subquadratic maximum matching algorithm for networks of arbitrary topology
is a major open problem~\cite{AK20}.
In addition, there are numerous works on computing approximate matchings, which are
usually based on computing a maximal matching, using the framework of Hopcroft and
Karp~\cite{HK73}, or rounding a fractional matching (cf.~\cite{AK20} for an overview
on approximate matching algorithms in the \congest~model).

\paragraph{Laplacian Solvers}
Our algorithm combines both tree-based ultrasparsification
algorithms \cite{ST14,KMP10,KMP11,CKMPPRX14}
and elimination-based algorithms that utilize Schur
complements \cite{KLPSS16,KS16,K17:thesis}.
Both types of algorithms were originally developed for the sequential model.
The issue of round complexity was previously addressed in parallel
Laplacian solving~\cite{BGKMPT14,PS14}.

We believe that a variant of~\cite{BGKMPT14} tailored to the \congest model gives a round complexity of around $n^{3/4} + Dn^{1/4}$ as opposed to the bound in Theorem \ref{thm:Main} to because the depth of the parallel algorithm of \cite{BGKMPT14} is more than polylogarithmic.
The polylogarithmic depth parallel algorithm from~\cite{PS14} is more difficult
to convert to the \congest~setting because it explicitly adds edges to the graph,
which causes increased congestion.

The outer layer recursion of our algorithm is akin to the recursive construction
of solvers and preconditioners present in Laplacian
solving~\cite{P13:thesis,KLPSS16}, approximate max-flow~\cite{P16},
and matrix sampling~\cite{CP15,CLMMPS15,CMM17}.

Parallel Laplacian solvers and spectral algorithms have also motivated
the study of (nearly) log space variants of these
algorithms~\cite{MurtaghRSV17,MurtaghRSV19,AKMPSV19:arxiv}.
It's an intriguing question to formally connect these low space
algorithms with distributed algorithms, both of which stem from
works on low iteration count algorithms.

\paragraph{Vertex Sparsification}
Critical to our result is the construction of minor based Schur complements
by Li and Schild~\cite{LS18}.
Minor based sparsification has been studied for distances~\cite{CGH16,KrauthgamerNZ14},
and implicitly for cuts via hierarchical routing schemes~\cite{Racke02,RackeST14}.
A more systematic treatment of uses of such sparsifiers,
in dynamic graph algorithms, can be found in~\cite{Goranci19:thesis}.
Some of the cut preserving vertex sparsifiers~\cite{M09,LM10,CLLM10,EGKRTT10,KrauthgamerR13},
as well as their recent variations in small cut settings~\cite{CDLKLPSV20} produce either minors or probability distribution over minors. 

\section{Preliminaries}
\label{sec:prelimes}

We start by describing general notation we use throughout the paper.
\paragraph{General notation.}
Given a symmetric matrix $\mM$, we let $\|\mM\|_2 = \max_{\|x\|_2 = 1} |x^\top\mM x|$ denote the maximum absolute value of any eigenvalue. For a vector $v$ and matrix $\mM$, we define $\|v\|_{\mM} \defeq
\sqrt{v^\top \mM v}$. For positive real numbers $a, b$ we say that $a \approx_\eps b$ if $\exp(-\eps)a \le b \le \exp(\eps)a$.
We say that a matrix $\mM \in \R^{n \times n}$ is positive semidefinite if $x^\top\mM x \ge 0$ for all $x \in \R^n$. For matrices $\mA$ and $\mB$, we write $\mA \pe \mB$ if $\mB - \mA$ is positive semidefinite. For positive semidefinite matrices $\mA, \mB$ we say that $\mA \approx_\eps \mB$ if $\exp(-\eps)\mA \pe \mB \pe \exp(\eps)\mA$.

\paragraph{Schur complements and Cholesky factorization.}
Our algorithms are based on Schur complements and sparsified Cholesky factorization. At a high level, the Schur complement of an $n \times n$ matrix provides a matrix which is equivalent under linear system solves on a subset of coordinates in $[n]$.
\begin{definition}[Schur complement]
\label{def:sc}
For an $n \times n$ symmetric matrix $\mM$ and subset of \emph{terminals} $\T \subseteq [n]$, let $S = [n] \bs \T.$ Permute the rows/columns of $\mM$ to write
\[
\mM
=
\left[
\begin{array}{cc}
	\mM_{[S,S]} & \mM_{[S,\T]} \\
	\mM_{[\T,S]} & \mM_{[\T,\T]}.
\end{array}
\right]
\]
Then the \emph{Schur complement} of $\mM$ onto $\T$ is denoted
$\mSC(\mM, \T) \defeq \mM_{[\T,\T]} - \mM_{[\T,S]}\mM_{[S,S]}^{-1}\mM_{[S,\T]}$.

For a graph $G$ and subset $\T \subseteq V(G)$, for simplicity we write $\mSC(G, \T) \defeq \mSC(\mL_G, \T)$. It is well-known that $\mSC(G, \T)$ is also a Laplacian.
\end{definition}
\begin{lemma}[Cholesky factorization]
\label{lemma:cholesky}
Given a matrix $\mM \in \R^{n \times n}$, a subset $\T \subseteq [n]$, and $S = [n] \bs \T$, we have
	\[ \mM^{-1} = \left[
	\begin{array}{cc}
	\mI & - \mM_{[S, S]}^{-1}\mM_{[S, \T]}\\
	0 & \mI
	\end{array}
	\right]
	\left[
	\begin{array}{cc}
	\mM_{[S, S]}^{-1} & 0\\
	0 & \mSC(\mM, \T)^{-1}
	\end{array}
	\right]
	\left[
	\begin{array}{cc}
	\mI & 0 \\
	- \mM_{[\T, S]}\mM_{[S, S]}^{-1}& \mI
	\end{array}
	\right].
	\]
\end{lemma}

The Cholesky factorization directly implies that the Schur complement represents the inverse of the Laplacian on a subset of the coordinates.
\begin{lemma}[e.g. Fact 5.4 in \cite{DKPRS17}]
\label{lemma:Inverse}
Let $\mI$ be the identity matrix, and let $\mJ$ be the all $1$ matrix. For any graph
$G$, and subset $\T \subseteq V(G)$ we have that
\[ \mSC(G,\T)^\dagger = (\mI - |\T|^{-1}\mJ)(\mL_G^\dagger)_{[\T,\T]}(\mI - |\T|^{-1}\mJ). \] In addition, we have that
\[ \mSC(G,\T)(\mL_G^\dagger)_{[\T,\T]}\mSC(G,\T) = \mSC(G,\T). \]
\end{lemma}
An equivalent view is that the quadratic form of the Schur complement gives the minimum energy
extension of a vector on the terminals to the original vertex set,
in the quadratic form of the original Laplacian \cite{G96:thesis,MP13}.
\begin{lemma}
\label{lemma:SCEnergy}
(Lemma B.2. of~\cite{MP13},
matrix version in Appendix A.5.5 of~\cite{BV04:book})
For a graph $G$ and a $\T \subseteq V(G)$,
the Schur complement of the Laplacian of $G$ onto $\T$,
$\mSC(G,\T)$ satisfies for all vectors $\vec{x}_{[\T]}$:
\[
\norm{\vec{x}_{\left[\T\right]}}_{\mSC\left( G, \T \right)}
=
\min_{\vec{x}_{\left[V \setminus \T \right]} \in \R^{V \setminus \T}}
\norm{
\left[
\begin{array}{c}
	\vec{x}_{\left[V \setminus \T \right]}\\
	\vec{x}_{\left[\T\right]}
\end{array}
\right]
}_{\mL( G)}.
\]
\end{lemma}

\paragraph{Matrix Analysis Tools}
Our algorithm for computing Schur complement sparsifiers which are minors requires
computing and sampling via leverage scores.

\begin{definition}[Effective resistance and leverage scores]
\label{def:lev}
For a graph $G$ with resistances $r_e$, define $\res_G(e) \defeq b_e^\top \mL_G^\dagger b_e$ and $\lev_G(e) \defeq \res_G(e)/r_e$.
\end{definition}
Note that $0 \le \lev_G(e) \le 1$ and $\sum_{e \in E(G)} \lev_G(e) = n-1$ for connected graphs $G$.

Let $G \bs e$ and $G/e$ denote the graphs resulting respectively from deleting and contracting edge $e$. Note that these correspond to setting the resistance of edge $e$ to positive infinity or $0$, respectively. The Woodbury matrix formula allows us to understand changes in the quadratic form when resistances of of the edges change. 
\begin{lemma}[Woodbury matrix formula]
\label{lemma:woodbury}
For matrices $\mA, \mU, \mC, \mV$ of compatible sizes we have
\[ (\mA+\mU\mC\mV)^\dagger = \mA^\dagger - \mA^\dagger \mU(\mC^{-1}+\mV\mA^\dagger \mU)^{-1}\mV\mA^\dagger. \]
\end{lemma}
We use the following to understand the matrix martingales that arise in the analysis minor-based Schur complements.
\begin{lemma}[Freedman's inequality for matrix martingales \cite{Tropp11}]
\label{lemma:freedman}
Consider a matrix martingale $(\mY^{(k)})_{k \ge 0}$ whose values are symmetric matrices with dimension $d$ and let $(\mX^{(k)})_{k \ge 1}$ be the difference sequence $\mX^{(k)} \defeq \mY^{(k)} - \mY^{(k-1)}$. Assume that the difference sequence is uniformly bounded in that $\|\mX^{(k)}\|_2 \le R$ almost surely for $k \ge 1$. Define the predictable quadratic variation random matrix
\[ \mW^{(k)} \defeq \sum_{j=1}^k \E[(\mX^{(j)})^2 | \mX^{(j-1)}]. \]
Then for all $\eps \ge 0$ and $\sigma^2 > 0$ we have that
\[ \Pr\left[\exists k > 0 : \|\mY^{(k)} - \mY^{(0)}\|_2 \ge \eps \text{ and } \|\mW^{(k)}\|_2 \le \sigma^2\right] \le 2d \cdot \exp\left(\frac{-\eps^2/3}{\sigma^2 + R\eps/3}\right). \]
\end{lemma}
The induced $2$-norm of a symmetric matrix is bounded by its maximum row sum.
\begin{lemma}
\label{lemma:sdd}
For a symmetric matrix $\mM \in \R^{n \times n}$, we have that
\[ \|\mM\|_2 \le \max_{i \in [n]} \sum_{j \in [n]} |\mM_{ij}|. \]
\end{lemma}
\begin{proof}
For all vectors $\vec{x}$, note by the AM-GM inequality that
\[
\vec{x}^\top \mM \vec{x}
=
\sum_{1 \le i,j \le n} \vec{x}_i \vec{x}_j \mM_{ij}
\le
\sum_{1 \le i,j \le n} \vec{x}_i^2 \abs{\mM_{ij}}
\le
\max_{i \in [n]} \sum_{j \in [n]} \abs{\mM_{ij}} \sum_{i \in [n]} \vec{x}_i^2
\le
\max_{i \in [n]} \sum_{j \in [n]} \abs{\mM_{ij}} \norm{x}_2^2.
\]
\end{proof}

\paragraph{\congest~model}

In the \congest~model \cite{Peleg00}, we are given a communication network $ \overline{G} = (\overline{V}, \overline{E}) $ with $ \overline{n} $ nodes modelling processors that have unique $ O (\log n) $-bit IDs, $ \overline{m} $ edges modelling bidirectional communication links between the processors, and diameter $ D $.
Initially, each node knows its own ID and the IDs of its neighbors as well as the value of $ n $.
Computation in this model is carried out in rounds synchronized by a global clock.
In each round, every node sends to each of its neighbors an arbitrary $ O (\log n) $-bit message, receives the messages of its neighbors, performs arbitrary internal computation, and stores arbitrary information for the next round.
The main goal in this paper is to design algorithms for graph problems with a small number of rounds.
For problems on directed graphs, the direction of each edge is known by both of its endpoints, but the corresponding communication link is still bidirectional.
For problems on weighted graphs (involving, e.g., costs or capacities), the weight of each edge is known by both of its endpoints, but the corresponding communication link still allows for direct transmission of each message within a single round.
In particular the diameter $ D $ always refers to the underlying undirected, unweighted communication network.
Our running time bounds hold under the assumption that all weights are polynomial in $ n $, which is a standard assumption in the \congest~model literature.

\section{Overview}
\label{sec:Overview}

Here we will give the main ideas behind our algorithm that efficiently solves Laplacians in a distributed setting (Theorem \ref{thm:Main}). We start by discussing elimination-based parallel Laplacian algorithms which remove a constant fraction of vertices to reduce the size of the graph. This naturally leads to requiring sparsifiers of the Schur complement that are minors of the original graph, whose existence is shown by \cite{LS18}. Our key contribution is a nontrivial distributed implementation of their construction by giving a parallel variant of their algorithm that avoids the sampling of random spanning trees and samples by approximate leverage scores. We analyze this algorithm using matrix martingales. Finally, to achieve our main result we combine the algorithm with tree-based ultrasparsifiers, an alternate vertex reduction scheme that is not parallel but significantly reduces the size of the graph.

\paragraph{Parallel Laplacian Solvers via Elimination}

The starting point for our algorithm is based on the $\poly(\log{n})$ round Laplacian system
solvers in the PRAM model, namely the sparsified Cholesky algorithm
from~\cite{KLPSS16}.
This algorithm repeatedly finds a constant fraction of the vertices
on which the block minor is ``almost independent'' and hence easy to solve.
The inverse of this block then gives the result of eliminating these
vertices, which is the Schur complement on the rest of the vertices,
which we view as the terminal vertices $\T$.

More explicitly, this can be seen in the context of the Cholesky factorization
in Lemma \ref{lemma:cholesky}, where we let $\mM = \mL(G)$ be the Laplacian.
We find an ``almost independent'' set of vertices $S$
so that computing $\mM_{[S,S]}^{-1}$ to high accuracy is simple using a preconditioned gradient descent method.
Therefore, the remaining difficulty in computing $\mM^{-1}$ is simply from computing and inverting the Schur complement: $\mSC(\mM, \T)^{-1}$.
To do this, we first approximately compute the Schur complement $\mSC(\mM, \T)$, which is again a Laplacian, and then recursively apply
a Cholesky factorization to it again.

However, this resulting Schur complement may be dense, even if the original graph is sparse.
For example, eliminating the center of a star results in a complete graph
on the peripheral leaf vertices.
To make this more efficient, sparsified elimination
algorithms~\cite{KLPSS16,CGPSSW18,DPPR20} seek to directly construct
a sparse approximation of this Schur complement.
This can be done in a variety of ways, but algorithmically one of the simplest
interpretations is through the sampling of random walks.
Indeed, matrix concentration bounds imply that the following procedure
suffices for generating a good approximation of $\mSC(G, \T)$ with high probability:

\begin{algorithm}[ht]
  \caption{Approximate Schur Complement using Random Walks}
    Set $H \gets \emptyset$ \\
    \For{ each edge $e = uv$ in $G$}
    {
    		Repeat the following two steps $O(\epsilon^{-2}\log{n})$ times: \\
    		Random walk both endpoints $u$ and $v$ until they are in $\T$,
	to $t_u$ and $t_v$ respectively. \\
	    Add an edge to the approximate Schur complement $H$ between
	$t_u$ and $t_v$, with weight as function of the original weight,
	and the number of steps the walk took.
    }
    \Return $H$
\end{algorithm}

By picking $\T$ so that $V \bs \T$ is almost independent, that is, each vertex
not in $\T$ has a constant fraction of its weight going to $\T$, it can be ensured
that the lengths of the walks don't exceed $O(\log{n})$ with high probability.
As a result, PRAM algorithms are able to construct low error Schur complements
by sampling about $O(\epsilon^{-2}\log{n})$ walks of length $O(\log{n})$ per edge.
As the number of vertices in the Schur complement decreases by a constant factor
per step, this process yields a parallel solver with another $O(\log{n})$ factor
overhead in parallel depth.

Another contribution we make is introducing a new technique that reduces the congestion of these random walks by augmenting the terminal set $\mathcal{T}$. More concretely, recall that in the \congest~model, each edge can only pass $O(\log n)$ bits
per round. Random walks in weighted graphs on the other hand may severely congest some edges:
consider for example, a star with one very heavily weighted edge, and rest lightly
weighted. All the walks starting from the lightly weighted edges' end points will likely utilize
the heavily weighted edges, leading to a congestion of $\Omega(n)$ in the worst case. To resolve this we use a procedure to \emph{estimate} the congestion of an edge accumulated by such random walks. We use these estimates to add edges with high estimated congestion to $\mathcal{T}$ to ensure that remaining edges have low congestion.

However, a single elimination round only removes a constant fraction of the vertices,
but performing $\Omega(\log n)$ elimination rounds would result in a significant blowup in the congestion (as each elimination round accumulates $\O(1)$ congestion).
Hence, we only perform $\Theta((\log \log n)^2)$ rounds of elimination between
sparsification steps.
A formal statement of this elimination scheme is shown in Lemma \ref{lem:Elimination}.

\paragraph{Minor Sparsifiers and its Distributed Construction}

After that, the core component of our algorithm is that we must bring the Schur complement back to
being a minor of the original graph, by constructing a spectral sparsifier of the Schur complement which is a minor of the original graph (Theorem \ref{thm:MinorSC}). That is, the Schur complement results from contracting connected subsets of vertices in the original graph and reweighting edges. Minors are particularly useful for distributed algorithms because we can simulate
 one round of communication between neighbors on a minor, such as multiplying
by the incidence matrix, in $\O(\sqrt{n}+D)$ rounds (Lemma
\ref{lemma:Communication}).
They interact particularly well with the parallel Laplacian solving algorithm
which is a short sequence of matrix-vector multiplies on submatrices.
Existence and efficient sequential constructions of these objects were first
shown by Li and Schild~\cite{LS18}.

A key contribution of our work is to give a simplified parallel variant of the algorithm of Li-Schild~\cite{LS18} which leads to an efficient distributed implementation.
The algorithm of \cite{LS18} works by contracting or deleting edges $e$
with probability given by its leverage score.
The main difference is that instead of sampling edges using a random spanning tree,
we identify a large subset of edges that can be sampled independent of each other,
without affecting each edges's sampling probability too much.
This is done via localization~\cite{SRS18}, which provides an overall bound
on the total influence of edges' effective resistances.

The algorithm then comprises of three main steps, and is analyzed via matrix martingales.
\begin{enumerate}
	\item Calculating edges' influence on the Schur complement (Lemma \ref{lem:DiffApx}).
	\item Computing the mutual influence of edges's resistances,
	and picking a large set that has small mutual influence,
	which we term the steady set (Definition \ref{def:alpha_delta_steady}).
	\item Among these steady edges, randomly contract/delete them with probability
	given by an approximation of their leverage scores.
\end{enumerate}
Note that all steps actually require solving Laplacian systems in the original graph, which seems circular.
We address this using the now well understood recursive approach of \cite{P16}, which we discuss below together with the overall algorithm.

\paragraph{Overall Recursive Scheme}
Given a graph $G$, the goal of the algorithm is to return a \emph{chain} of approximate Schur complements of $G$, each with $0.99$ as many vertices as the previous. This chain has length $O(\log n)$, and after built, can be applied in $O(\log n)$ steps and $\O(\sqrt{n}+D)$ rounds to solve a Laplacian system to high accuracy in the \congest~model. The construction of the chain is as follows -- pick $d = \Theta((\log \log n)^2)$ say, and run $d$ rounds of the sparsified Cholesky elimination scheme (Lemma \ref{lem:Elimination}) to reduce the graph size to $0.99^d|V(G)|$. Now, use the minor Schur complement algorithm (Theorem \ref{thm:MinorSC}) to build a minor of $G$ which is a Schur complement sparsifier with respect to the remaining $0.99^d|V(G)|$ remaining vertices.

To compute the Schur complement sparsifier, we employ a separate recursion, because the Schur complement sparsifier construction requires Laplacian system solves to compute leverage scores (and other similar measures). To do this, we ultrasparsify the graph $G$, thus reducing the size by a factor of $k$ (Lemma \ref{lem:UltraSparsify}), and build a Schur complement chain on the ultrasparsifier. Now, we can use this solver on the ultrasparsifier to precondition a solver on $G$ with $\O(\sqrt{k})$ steps of preconditioned conjugate gradient to compute the desired leverage scores. We want to emphasize the final Schur complement chain we output for $G$ \emph{does not} involve the ultrasparsifier, and hence can still be applied in parallel.

One final technical detail is that due to needing to solve submatrices of the Laplacian (Lemma \ref{lem:DiffApx}) we require tracking graphs that embed with \emph{low congestion} in the original graph, a slight generalization of minors (Definition \ref{def:Minor}). We ensure that the congestion stays as $n^{o(1)}$ throughout the algorithm, so it does not affect the final round complexity.

\section{Full Algorithm and Analysis}
\label{sec:ClownFiesta}

The goal of this section is to formalize the notions and graph reduction algorithms
described in Section \ref{sec:Overview}, and provide a bound for the overall
performance.
\subsection{Distributed Communication on Minors of Overlay Networks}
\label{subsec:Communication}

As described, we will work with graphs that are minors
of the original graph, which doubles as the communication network.
However, some of our linear systems reductions duplicate edges,
leading to minors with slightly larger congestion.
So we will need to incorporate such congestion parameters into
our definition of minors.
The following definition is a direct extension
of the distributed $N$-node cluster graph from \cite{GKKLP15},
with congestion incorporated, and the connection with graph minors
stated more explicitly.

\begin{definition}
\label{def:Minor}
  Given a parameter $\rho \geq 1$, a graph $G$ is a $\rho$-minor of $H$ if we have the following mappings:
  \begin{enumerate}
    \item For each vertex of $G$, $u \in V(G)$:
    \begin{enumerate}
      \item A subset of vertices of $H$,
      which we term a supervertex, $S^{G \rightarrow H} (u) \subseteq V(H)$,
      with a root vertex $V_{map}^{G \rightarrow H} (u) \in S^{G \rightarrow H}(u)$.
      \item A connected subgraph of $H$ on $S^{G \rightarrow H}(u)$,
      which for simplicity we will keep as a tree, $T^{G \rightarrow H} (u)$.
    Note that this requires $S^{G \rightarrow H}(u)$ being connected in $H$.
    \end{enumerate}
    \item A mapping of the edges of $G$ onto edges of $H$, or self-loops on
    vertices of $H$,
    such that for any $u^{G} v^{G} = e^{G} \in E(G)$,
    the mapped edge $E_{map}^{G \rightarrow H}(e^{G}) = e^{H} = u^{H} v^{H}$
    satisfies $u^{H} \in S^{G \rightarrow H} ( u^{G})$
    and $v^{H} \in S^{G \rightarrow H} ( v^{G} )$.
  \end{enumerate}
and additionally:
  \begin{enumerate}
    \item Each vertex of $H$ is contained in at most $\rho$ supervertices
    $V_{map}^{G \rightarrow H}(v^{G})$ for some $v^{G}$.
    \item Each edge of $H$ appears as the image of the edge map
    $E_{map}^{G \rightarrow H}(\cdot)$, or in one of the trees
    connecting supervertices, $T^{G \rightarrow H}(v^{G})$ for some $v^{G}$,
    at most $\rho$ times.
  \end{enumerate}
When $\rho = 1$, then $G$ is simply a minor of $H$.

  Finally, we say a $\rho$-minor mapping is stored distributedly,
  or that $G$ is $\rho$-minor distributed over $H$ if it's stored by
  having all the images of the maps recording their sources.
  That is, each $v^H \in V(H)$ records
  \begin{enumerate}
    \item All $v^{G}$ for which $v^{H} \in V_{map}^{G \rightarrow H}(v^{G})$,
    \item For each edge $e^{H}$ incident to $v^{H}$
       (including self loops that may not exist in original $H$):
      \begin{enumerate}
         \item All vertices for which $e^{H}$ is in the corresponding tree
         \[
           \left\{ v^{G} \mid e^{H} \in T^{G \rightarrow H} \left( v^{G} \right)
           \right\}
         \]
         \item All edges $e^{G}$ that map to it.
      \end{enumerate}
  \end{enumerate}
\end{definition}

We will denote the original graph,
which doubles as the overlay network, using $\overline{G}$.

Note that the vertex mappings, or even the neighborhoods of $G$,
cannot be stored at one vertex in $\overline{G}$.
This is because both of these sets may have size up to $\Omega(n)$,
and passing that information to a single low degree
vertex would incur too much communication.

We store vectors on $G$ by putting the values
at the root vertices of each of its corresponding supervertices.
This notion of rooting can be made more explicit:
we can compute directions for all edges in the
spanning tree $T^{G \rightarrow \overline{G}}(v^{G})$ that point to
the corresponding root vertex $V_{map}^{G \rightarrow \overline{G}}(v^{G})$.

\begin{lemma}
\label{lemma:RootTrees}
Given a graph $G$ that's $\rho$-minor distributed over a communication network
 $\overline{G} = (\overline{V}, \overline{E})$ with $\overline{n}$ vertices,
 $\overline{m}$ edges, and diameter $D$,
  we can compute in $O(\rho\sqrt{\overline n} \log{\overline n} +
  D)$
  rounds of communication on $\overline{G}$, an orientation for each $v^{G}$
  and each edge $e \in T^{G \rightarrow \overline{G}}(v^{G})$ such that
  each vertex other than the root has exactly one edge pointing away from it,
  and following these edges leads us to the root.
\end{lemma}

We will make extensive usage of the following lemma,
which we prove in Appendix~\ref{sec:MinorProofs}, about simulating communications
on $G$ using rounds of \congest communications
in a graph that it $\rho$-minor distributes into.

\begin{lemma}
  \label{lemma:Communication}
  Let $G = (V, E)$ be a graph with $n$ vertices and $m$ edges
  that $\rho$-minor distributes into a communication network
  $\overline{G} = (\overline{V}, \overline{E})$ with $\overline{n}$ vertices,
  $\overline{m}$ edges, and diameter $D$.
  In the \congest model, the following operations can be performed
  with high probability using $O(t \rho\sqrt{\overline n} \log{\overline n} +
  D)$
  rounds of communication on $\overline{G}$:
  \begin{enumerate}
    \item Each $V_{map}^{G \rightarrow \overline{G}}(v^{G})$ sends
    $O(t \log{n})$ bits of information to
    all vertices in $S^{G \rightarrow \overline{G}}(v^{G})$.
    \item Simultaneously aggregate the sum/minimum
    of $O(t \log{n})$ bits, from all vertices in
    $S^{G \rightarrow \overline{G}}(v^{G})$
    to
    $V_{map}^{G \rightarrow \overline{G}}(v^{G})$
    for all $v^{G} \in V(G)$.
  \end{enumerate}
\end{lemma}

One direct use of this communication result is that it allows us to efficiently
compute matrix-vector products.
\begin{corollary}
\label{corollary:MatVec}
Given a matrix $\mA$ with nonzeroes supported on the edges of a graph $G$
that's $\rho$-minor distributed over a communication network
 $\overline{G} = (\overline{V}, \overline{E})$ with $\overline{n}$ vertices,
 $\overline{m}$ edges, and diameter $D$,
with values stored with endpoints of the corresponding edge,
and a vector $\vec{x} \in \R^{|V(G)|}$  stored distributedly on the vertices
$V_{map}^{G \rightarrow \overline{G}}(v^{G})$, we can compute the vector
$\mA \vec{x}$, also stored at $V_{map}^{G \rightarrow \overline{G}}(v^{G})$ using
$O(t \rho\sqrt{\overline n} \log{\overline n} + D)$ of communication
in the \congest model, with high probability.
\end{corollary}

\begin{proof}
We first invoke Lemma~\ref{lemma:Communication} to pass $\vec{x}_{V^{G}}$
to all of $S^{G \rightarrow \overline{G}}(V^{G})$.
Then in $O(\rho)$ round of distributed communication, we can pass these entries
(multiplied by the weights of $\mA$) to the corresponding row index.
That is, if
$E_{map}^{G \rightarrow \overline{G}}(u^{G} v^{G}) = u^{H} v^{H}$,
we pass
\[
\mA_{u^{G} v^{G}} \vec{x}_{v^{G}}
\]
from $v^{H}$ to $u^{H}$.
Running Lemma~\ref{lemma:Communication} again to sum together the passed values
over each super vertex then brings the values to the root vertex.
\end{proof}

Another implication is that that the Koutis-Xu distributed sparsification algorithm
can also be simulated on $\overline{G}$ in the \congest model, with a round overhead
of $\O(\sqrt{\overline{n}} + D)$~\cite{KX16}.
\begin{corollary}
\label{corollary:Sparsify}
There is an algorithm, $\Sparsify$,
that for a graph $G$ that $\rho$-minor distributes into $\overline{G}$,
and some error $0 < \epsilon < 0.1$,
$\Sparsify(G, \overline{G}, \epsilon)$
with high probability
returns in
\[
O\left(\left(\rho \sqrt{\overline n} \log \overline n + D\right)\log^{8}
\overline n/\eps^2 \right)
\]
rounds, a graph $\widetilde{G}$, distributed as a $\rho$-minor in $\overline{G}$
such that:
\begin{enumerate}
	\item $\widetilde{G}$ is a (reweighted) subgraph of $G$,
	\item $\mL(G) \approx_{\epsilon} \mL(\widetilde{G})$,
	\item $\widetilde{G}$ has $ \O ( |V(G)| / \epsilon^2) $ edges.
\end{enumerate}
\end{corollary}

\begin{proof}
The algorithm by~\cite{KX16} is based on repeated spanner computations on subgraphs (which are obtained by removing edges from previous spanner computations and uniform sampling of edges).
The spanner algorithm of~\cite{BS07}, internally used in~\cite{KX16}, iteratively grows clusters -- organized as spanning trees rooted at center nodes -- and adds edges to the spanner.
In each iteration some of the existing clusters first are sampled at random, which is done by the respective center node who then forwards the information whether the cluster is sampled to all nodes in its cluster.
Then each node decides whether it joins a cluster and if so which one and also decides which of its neighboring edges it adds to the spanner.
These decisions are made by comparing the weights of its incident edges.
Thus, all the operations performed by nodes in the algorithm of~\cite{BS07}, and thus the algorithm of~\cite{KX16} fit the description of operations supported by Lemma~\ref{lemma:Communication}.
\end{proof}

We will call this sparsification routine regularly, often as
preprocessing.
This is partly because subgraphs $1$-minor distributes into itself trivially.

The minor property also compose naturally:
a minor of a minor of $G$ is also a minor of $G$.
This holds with $\rho$-minors too, up to multiplications of the congestion parameters.
We prove the following general composition result in Appendix~\ref{sec:MinorProofs}.

\begin{lemma}
\label{lemma:MinorCompose}
Given graphs $G_1$, and $G_2$ via
a $\rho_2$-minor distribution of $G_2$ into $\overline{G}$,
and a $\rho_1$-minor distribution of $G_1$ into $G_2$ stored
on the root vertices of the supervertices of $G_2$, and images
of $G_2$'s edges in $\overline{G}$,
we can, with high probability,
compute using $\O( \rho_{1} \rho_2 \cdot (\sqrt{\overline{n}} + D))$
rounds of communication in the \congest model a $\rho_1 \cdot \rho_2$-minor
distribution of $G_1$ into $G$.
\end{lemma}

We will always work with congestions in the $n^{o(1)}$ range:
this essentially means we can perform distributed algorithms on $G$,
while paying an overhead of about $\sqrt{\overline{n}} + D$ in round
complexity to simulate on the original graph.

The composition of minors from Lemma~\ref{lemma:MinorCompose} implies,
among others, that a subset of edges can be quickly contracted.
\begin{corollary}
\label{corollary:ActualMinor}
Given $G$ that's $\rho$-minor distributed on $\overline{G}$,
along with a subset of edges $F \subseteq E(G)$
then we can obtain a $\rho$-minor distribution of $G / F$
($G$ with $F$ contracted),
into $\overline{G}$ in $\O( \rho ( \sqrt{\overline{n}} + D))$ rounds,
under the \congest model of computation.
\end{corollary}

The proof of this requires running $O(\log{n})$ rounds of parallel
contraction on the edges of $F$.
We defer it to Appendix~\ref{sec:MinorProofs} as well.

\subsection{Laplacian Building Blocks}

Some of our algorithm require working with submatrices of Laplacians, which may not be Laplacians anymore, but are still SDD.
Fortunately, we can reduce solving an SDD matrix on a graph to solving a Laplacian on
a $2$-minor.
The proof is straightforward and can be found in \cite{G96:thesis} or \cite{ST14}.
\begin{lemma}[Gremban \cite{G96:thesis}]
  \label{lem:SDDtoL}
  Given an $n$-by-$n$ SDD matrix $\mM$ that $\rho$-minor distributes
  into $\overline{G}$,
  we can construct a graph $H$ on $2n$ vertices,
  along with a $2\rho$-minor distribution of $H$ into $\overline{G}$
  with vertices $i$ and $n + i$'s roots mapping to the same root vertex,
  so that for any vector $\vec{b} \in \R^{n}$
  and any vectors $\vec{x} \in \R^{2n}$ such that
  \[
  \norm{\vec{x} - \mL\left( H \right)^{\dag}
  \left[\begin{array}{c} \vec{b} \\ \vec{b} \end{array} \right]
	}_{\mL\left( H \right)}
	\leq
	\epsilon \norm{\left[\begin{array}{c} \vec{b} \\ \vec{b} \end{array}
	\right]}_{\mL\left( H \right)^{\dag}},
	\]
	we have
	\[
		\norm{\frac{\vec{x}_{1:n} - \vec{x}_{n+1:2n}}{2}
			 - \mM^{\dag}\vec{b}}_{\mM}
		 \leq
		\epsilon \norm{\vec{b}}_{\mM^{\dag}}.
	\]
\end{lemma}

In this section we outline the main pieces needed to prove Theorem \ref{thm:Main}.
As described in Section \ref{sec:Overview}, we require three main graph reduction procedures: ultrasparsification (Lemma \ref{lem:UltraSparsify}), sparsified Cholesky (Lemma \ref{lem:Elimination}), and minor-based Schur complements (Theorem \ref{thm:MinorSC}).

The ultrasparsification procedure, based on \cite{ST14,KMP10}, allows us to significantly reduce the size of the graph and maintain congestion, but incurs a large approximation error.
\begin{restatable}{lemma}{UltraSparsify}
\label{lem:UltraSparsify}
	There is a routine $\Ultrasparsify (G, k)$ in the \congest\ model that
	given a graph $G$ with $n$ vertices and $m$ edges,
	that $\rho$-minor distributes into the communication network $\overline{G}$,
  which has $\overline n$ vertices, $\overline m$ edges, and diameter $D$,
	along with a parameter $k$,
	produces in $O(n^{o(1)}(\rho\sqrt{\overline{n}}+ D))$
	rounds a graph $H$ such that:
	\begin{enumerate}
		\item $H$ is a subgraph of $G$,
		\item $H$ has at most $n - 1 + m 2^{O(\sqrt{\log n \log \log n})}/ k$ edges.
		\item $\mL(G) \preceq \mL(H) \preceq k \mL(G)$.
	\end{enumerate}
	Furthermore, the algorithm also gives $\widehat{G}, \mZ_1, \mZ_2, C$ such that
	\begin{enumerate}
	\item  $\widehat{G}$ $1$-minor distributes into $H$ such that
	$\widehat{G} = \mSC(H, C)$ with $|C| = m 2^{O(\sqrt{\log n \log \log n})}/ k$.
	\item There are operators 	$\mZ_1$ and $\mZ_2$ evaluable with
	$O(\rho \sqrt{\overline n} \log \overline n+ D)$ rounds of \congest
	communication on $\overline{G}$ such that:
	\[
	\mL\left( H \right)^{\dag}
	= \mZ_1^{\tomato}
	\left[
	\begin{array}{cc}
	\mZ_{2} & 0\\
	0 & \mL\left( \widehat{G} \right)^{\dag}
	\end{array}
	\right]
	\mZ_1
	\]
	\end{enumerate}
\end{restatable}

The elimination procedure, based on \cite{KLPSS16}, incurs small approximation error, but significantly increases the congestion.
\begin{restatable}{lemma}{Elimination}
\label{lem:Elimination}
	There is a routine $\Eliminate (G, d, \epsilon)$ in the \congest~model
	that given a graph $G$ that $\rho$-minor distributes
	into a communication network $\overline{G}$,
	along with step count $d$ and error $\epsilon$,
	produces in
	\[
	O( (\eps^{-6}\log^{14} n)^d (\rho\sqrt{\overline{n}} \log \overline{n} + D))
	\]
	rounds a subset $\T$ and access to operators $\mZ_{1}$ and $\mZ_{2}$ such that
	\begin{enumerate}
		\item $|\T| \leq (\frac{49}{50})^{d} |V(G)|$.
		\item The cost of applying $\mZ_{1}$, $\mZ_{1}^{\tomato}$ and $\mZ_{2}$ to
		vectors is $O((\eps^{-6}\log^{14} \overline{n})^d (\rho\sqrt{\overline{n}} \log
		\overline{n} + D))$ rounds of communication on $\overline{G}$.
		\item $\mL(G)^{\dag}$ is $(1 \pm \epsilon)^d$-approximated
		by a composed
		operator built from $\mZ_{1}$, $\mZ_{2}$, and the inverse of the
		of the Schur complement of $\mL(G)$ onto $C$, $\mSC(\mL(G), C)$:
		\[
		\left( 1 - \epsilon \right)^d \mL\left( G \right)^{\dag}
		\preceq
		\mZ_{1}^{\tomato}
		\left[
		\begin{array}{cc}
		\mZ_{2} & 0\\
		0 & \mSC\left( \mL\left( G \right), C \right)^{\dag}
		\end{array}
		\right] \mZ_{1}
		\preceq
		\left( 1 + \epsilon \right)^d \mL\left( G \right)^{\dag}
		\]
	\end{enumerate}
\end{restatable}

Note that we cannot directly set $d = \Omega(\log n)$ to finish
with $|\T|$ a constant: $(\log{n})^{\log{n}}$ may be even larger than $n$.
So we need to bring the structure back to a minor of the original
communication network.
For this we use the construction of spectral vertex sparsifiers
that are minors~\cite{LS18}, modified to not use random spanning trees.

\begin{restatable}{theorem}{MinorSC}
  \label{thm:MinorSC}
  There is a routine $\ApproxSC(G, \T, \epsilon)$
  in the \congest~model
  that given a graph $G$ with $n$ vertices and $m$ edges that $\rho$-minor distributes into the communication network $\overline{G}$,
  a subset of vertices $\T \subseteq V(G)$, an error parameter $\epsilon < 0.1$,
  and access to a (distributed) Laplacian solver $\solve$, it 
  returns a graph $H$, represented as a distributed $\rho$-minor of $\overline{G}$ such
  that:
  \begin{enumerate}
    \item $\T \subseteq V(H)$,
    \item $H$ has $O(|\T|\eps^{-2}\log^2 n)$ edges
    (and hence at most that many vertices as well).
    \item The Schur complements of $G$ and $H$ well approximate each other, i.e., 
    \[ \mSC\left(G, \T \right) \approx_\eps \mSC\left(H, \T \right). \]
  \end{enumerate}
  The cost of this computation consists of:
  \begin{enumerate}
    \item $O(\eps^{-3}\log^{10} n)$ calls to $\solve$ with accuracy $1/\poly(n)$
    on graphs that $2\rho$-distribute into $\overline{G}$.
    \item An overhead of $O(\rho (\overline{n}^{1/2} + D)\eps^{-3}\log^{11}
    \overline{n})$ rounds.
  \end{enumerate}
\end{restatable}

\subsection{Schur complement chains and a proof of Theorem \ref{thm:Main}}
In this section, we formally show how to combine Lemma \ref{lem:Elimination}, Lemma \ref{lem:UltraSparsify}, and Theorem \ref{thm:MinorSC}
to efficiently construct a Schur complement chain and prove Theorem \ref{thm:Main}.
\begin{definition}
For a graph $G$ of $n$ vertices, $\{(G_i, \mZ_{i, 1}, \mZ_{i, 2}, \T_i)\}_{i = 1}^t$ is a
$(\gamma, \epsilon)$-Schur-complement solver chain of $G$ if the following conditions hold.
\begin{enumerate}
	\item $G_1 = G$.
	\item \label{cond:OverallApprox}
			\[
		\left( 1 - \epsilon \right) \mL\left( G_{i} \right)^{\dag}
		\preceq
		\mZ_{i, 1}^{\tomato}
		\left[
		\begin{array}{cc}
		\mZ_{i, 2} & 0\\
		0 & \mSC\left( \mL\left( G_{i} \right), \T_i \right)^{\dag}
		\end{array}
		\right] \mZ_{i, 1}
		\preceq
		\left( 1 + \epsilon \right) \mL\left( G_{i} \right)^{\dag}
		\]
	\item \label{cond:SCApprox}
	$\T_i \subset V(G_{i+1}) \subset V(G_i)$ and $\mSC(G_{i}, \T_i) \approx_{\epsilon}
	\mSC(G_{i+1}, \T_i)$
	\item $|V(G_{i})| \geq \gamma \cdot |V(G_{i+1})|$ if $i < t$, and
	$|V(G_t)| \leq \gamma$.
\end{enumerate}
\end{definition}

\begin{algorithm}[ht]
\caption{\PseudoinverseMulti ($\{(G_i, \mZ_{i,1}, \mZ_{i,2}, \mathcal{T}_i)\mid j\le i\le t\}$, $\vec{b}$)}
\label{alg:PInvMult}
$\vec{u}\leftarrow\mZ_{j,1}\vec{b}$\;
\eIf{$j=t$}{$\vec{v}\leftarrow\mSC(G_t,\mathcal{T}_{t})^{\dag}\vec{u}_{[\mathcal{T}_t]}$\;
\Return $\mZ_{t,1}^{T}\left[\begin{array}{cc}
     \mZ_{t,2}\vec{u}_{[V(G_t)\setminus\mathcal{T}_t]}  \\
     \vec{v} 
\end{array}\right]$\;
}{
$\vec{u}_1\leftarrow\vec{u}_{[\mathcal{T}_j]}-\frac{\vec{u}_{[\mathcal{T}_j]}^T\vec{1}}{\left\|\vec{1}\right\|_{2}^2}\cdot\vec{1}$\;
$\vec{v}\leftarrow \PseudoinverseMulti \left(\{(G_i, \mZ_{i,1}, \mZ_{i,2}, \mathcal{T}_i)\mid j+1\le i\le t\}, \left[\begin{array}{cc}
     \vec{0}  \\
     \vec{u}_1 
\end{array}\right]\right)$\;
$\vec{v}_1\leftarrow\vec{v}_{[\mathcal{T}_j]}-\frac{\vec{v}_{[\mathcal{T}_j]}^T\vec{1}}{\left\|\vec{1}\right\|_{2}^2}\cdot\vec{1}$\;
\Return $\mZ_{j,1}^{T}\left[\begin{array}{cc}
     \mZ_{j,2}\vec{u}_{[V(G_j)\setminus\mathcal{T}_j]}  \\
     \vec{v}_1 
\end{array}\right]$\;
}
\end{algorithm}

\begin{lemma}\label{lem:applying_pseudoinverse}
	Let $\overline{G}$ be a communication network with $\overline n$ vertices and $\overline m$ edges.
	Let \\ $\{(G_i, \mZ_{i, 1}, \mZ_{i, 2}, \T_i)\}_{i = 1}^t$ be a $(\gamma, \epsilon)$-Schur-complement solver chain of graph $G$ for some $\gamma \geq 2$ and $\eps \le \frac{1}{C\log n}$ for large constant $C$, satisfying the following conditions:
\begin{enumerate}
	\item $G_i$ $\rho$-minor distributes into $\overline G$.
	\item Linear operators $\mZ_{i, 1}$ and $\mZ_{i, 2}$ can be evaluated in $O({\overline n}^{o(1)} (\overline n^{1/2} + D))$ rounds.
\end{enumerate}

Then for a given vector $\vec{b}$, Algorithm \PseudoinverseMulti computes a vector $\vec{x}$ in \\ $O(\rho {\overline n}^{o(1)} (\overline n^{1/2} + D))$  rounds such that
\[
\norm{\vec{x} - \mL\left( G \right)^{\dag} \vec{b}}_{\mL\left( G \right)}
\leq
2\epsilon\log n \cdot \norm{\vec{b}}_{\mL\left( G \right)^{\dag}}.
\]
\end{lemma}

The correctness proof rely heavily on the following conversion from
operator guarantees to error guarantees.

\begin{lemma}
	\label{lem:ApproxToError}
	(Lemma 1.6.7 of~\cite{P13:thesis})
	If $\mA$ and $\mB$ are two symmetric PSD matrices such that $\mA \approx_{\delta}\mB^{\dag}$
	for some $0 < \delta < 1$, then for any vector $\vec{b}$, we have
	\[
	\norm{\mA\vec{b} - \mB^{\dag} \vec{b}}_{\mB}
	\leq\delta\norm{\vec{b}}_{\mB^{\dag}}.
	\]
\end{lemma}

\begin{proof}[Proof of Lemma~\ref{lem:applying_pseudoinverse}]
We prove this lemma by induction on $j$. For the base case $j=t$, we have that 
\begin{equation}
\label{base_case_t_alg}
\begin{aligned}
    &\mZ_{t,1}^{T}\left[\begin{array}{cc}
         \mZ_{t,2}\vec{u}_{[V(G_t)\setminus\mathcal{T}_t]}  \\
         \vec{v} 
    \end{array}\right]=\mZ_{t,1}^{T}\left[\begin{array}{cc}
         \mZ_{t,2}\vec{u}_{[V(G_t)\setminus\mathcal{T}_t]}  \\
         \mSC(G_t,\mathcal{T}_t)^{\dag}\vec{u}_{[\mathcal{T}_t]}
    \end{array}\right]\\
    =&\mZ_{t,1}^{T}\left[\begin{array}{cc}
         \mZ_{t,2}&0  \\
         0&\mSC(G_t,\mathcal{T}_t)^{\dag} 
    \end{array}\right]\vec{u}=\mZ_{t,1}^{T}\left[\begin{array}{cc}
         \mZ_{t,2}&0  \\
         0&\mSC(G_t,\mathcal{T}_t)^{\dag} 
    \end{array}\right]\mZ_{t,1}\vec{b},
\end{aligned}
\end{equation}
in which 
\begin{align}
\label{base_case_t_apprx}
    (1-\epsilon)\mL(G_t)^{\dag}\preceq\mZ_{t,1}^{T}\left[\begin{array}{cc}
         \mZ_{t,2}&0  \\
         0&\mSC(G_t,\mathcal{T}_t)^{\dag} 
    \end{array}\right]\mZ_{t,1}\preceq(1+\epsilon)\mL(G_t)^{\dag}.
\end{align}
By Lemma \ref{lem:ApproxToError} and combining (\ref{base_case_t_alg}) and (\ref{base_case_t_apprx}), we have 
\begin{align*}
    \left\|\mZ_{t,1}^{T}\left[\begin{array}{cc}
         \mZ_{t,2}\vec{u}_{[V(G_t)\setminus\mathcal{T}_t]}  \\
         \vec{v} 
    \end{array}\right]-\mL(G_t)^{\dag}\vec{b}\right\|_{\mL(G_t)}\le\epsilon\left\|\vec{b}\right\|_{\mL(G_t)^{\dag}}. 
\end{align*}

We suppose that it holds for case $j+1$, i.e., 
\begin{align*}
    &\left\|\PseudoinverseMulti (\{(G_i, \mZ_{i,1}, \mZ_{i,2}, \mathcal{T}_i)\mid j+1\le i\le t\},\vec{b})-\mL(G_{j+1})^{\dag}\vec{b}\right\|_{\mL(G_{j+1})}\\
    \le&2(t-j)\epsilon\left\|\vec{b}\right\|_{\mL(G_{j+1})^{\dag}},
\end{align*}
then we have that 
\begin{align}
\label{induction_v_LG_j+1}
\left\|\vec{v}-\mL(G_{j+1})^{\dag}\left[\begin{array}{cc}
     \vec{0}  \\
     \vec{u}_1
\end{array}\right]\right\|_{\mL(G_{j+1})}\le 2(t-j)\epsilon\left\|\left[\begin{array}{c}
     \vec{0}  \\
     \vec{u}_1 
\end{array}\right]\right\|_{\mL(G_{j+1})^{\dag}}.
\end{align}
Lemma \ref{lemma:SCEnergy} gives 
\begin{align}
\label{v_Tj_LG_j+1_dag}
    \left\|\vec{v}_{[\mathcal{T}_j]}-\left(\mL(G_{j+1}^{\dag})\left[\begin{array}{c}
         \vec{0}  \\
         \vec{u}_1 
    \end{array}\right]\right)_{[\mathcal{T}_j]}\right\|_{\mSC(G_{j+1}, \mathcal{T}_j)}\le\left\|\vec{v}-\mL(G_{j+1})^{\dag}\left[\begin{array}{cc}
     \vec{0}  \\
     \vec{u}_1
\end{array}\right]\right\|_{\mL(G_{j+1})}.
\end{align}
Combining (\ref{induction_v_LG_j+1}) and (\ref{v_Tj_LG_j+1_dag}), we have  
\begin{align}
\label{v_Tj_le_induction_solution}
   \left\|\vec{v}_{[\mathcal{T}_j]}-\left(\mL(G_{j+1}^{\dag})\left[\begin{array}{c}
         \vec{0}  \\
         \vec{u}_1 
    \end{array}\right]\right)_{[\mathcal{T}_j]}\right\|_{\mSC(G_{j+1}, \mathcal{T}_j)} \le 2(t-j)\epsilon\left\|\left[\begin{array}{c}
     \vec{0}  \\
     \vec{u}_1 
\end{array}\right]\right\|_{\mL(G_{j+1})^{\dag}}.
\end{align}

Now we prove the case $j$. By triangle inequality, we have that 
\begin{align}
&\label{induction_j_left}\left\|\mZ_{j,1}^{T}\left[\begin{array}{cc}
     \mZ_{j,2}\vec{u}_{[V(G_j)\setminus\mathcal{T}_j]}  \\
     \vec{v}_1 
\end{array}\right]-\mL(G_j)^{\dag}\vec{b}\right\|_{\mL(G_j)}\\
\le&\label{induction_j_right_part1}\left\|\mZ_{j,1}^{T}\left[\begin{array}{cc}
     \mZ_{j,2}\vec{u}_{[V(G_j)\setminus\mathcal{T}_j]}  \\
     \vec{v}_1 
\end{array}\right]-\mZ_{j,1}^{T}\left[\begin{array}{cc}
         \mZ_{j,2}&0  \\
         0&\mSC(G_j,\mathcal{T}_j)^{\dag} 
    \end{array}\right]\mZ_{j,1}\vec{b}\right\|_{\mL(G_j)}\\
    +&\label{induction_j_right_part2}\left\|\mZ_{j,1}^{T}\left[\begin{array}{cc}
         \mZ_{j,2}&0  \\
         0&\mSC(G_j,\mathcal{T}_j)^{\dag} 
    \end{array}\right]\mZ_{j,1}\vec{b}-\mL(G_j)^{\dag}\vec{b}\right\|_{\mL(G_j)}. 
\end{align}
Obviously, for (\ref{induction_j_right_part2}), Lemma \ref{lem:applying_pseudoinverse} gives 
\begin{align}
\label{bound_8}
\left\|\mZ_{j,1}^{T}\left[\begin{array}{cc}
         \mZ_{j,2}&0  \\
         0&\mSC(G_j,\mathcal{T}_j)^{\dag} 
    \end{array}\right]\mZ_{j,1}\vec{b}-\mL(G_j)^{\dag}\vec{b}\right\|_{\mL(G_j)}\le\epsilon\left\|\vec{b}\right\|_{\mL(G_j)^{\dag}}.
\end{align}
Now our task is to bound (\ref{induction_j_right_part1}),
\begin{equation}
\label{bound_induction_j_right_part2}
\begin{aligned}
    &\left\|\mZ_{j,1}^{T}\left[\begin{array}{cc}
     \mZ_{j,2}\vec{u}_{[V(G_j)\setminus\mathcal{T}_j]}  \\
     \vec{v}_1 
\end{array}\right]-\mZ_{j,1}^{T}\left[\begin{array}{cc}
         \mZ_{j,2}&0  \\
         0&\mSC(G_j,\mathcal{T}_j)^{\dag} 
    \end{array}\right]\mZ_{j,1}\vec{b}\right\|_{\mL(G_j)}\\
    =&\left\|\mZ_{j,1}^{T}\left[\begin{array}{cc}
     \mZ_{j,2}\vec{u}_{[V(G_j)\setminus\mathcal{T}_j]}  \\
     \vec{v}_1 
\end{array}\right]-\mZ_{j,1}^{T}\left[\begin{array}{cc}
         \mZ_{j,2}&0  \\
         0&\mSC(G_j,\mathcal{T}_j)^{\dag} 
    \end{array}\right]\vec{u}\right\|_{\mL(G_j)}\\
    =&\left\|\mZ_{j,1}^{T}\left[\begin{array}{cc}
     \mZ_{j,2}\vec{u}_{[V(G_j)\setminus\mathcal{T}_j]}  \\
     \vec{v}_1 
\end{array}\right]-\mZ_{j,1}^{T}\left[\begin{array}{c}
     \mZ_{j,2}\vec{u}_{[V(G_j)\setminus\mathcal{T}_j]} \\
     \mSC(G_j, \mathcal{T}_j)^{\dag}\vec{u}_{[\mathcal{T}_j]}
\end{array}\right]\right\|_{\mL(G_j)}\\
=&\left\|\mZ_{j,1}^{T}\left[\begin{array}{c}
     \vec{0}_{[V(G_j)\setminus\mathcal{T}_j]} \\
     \vec{v}_1-\mSC(G_j,\mathcal{T}_j)^{\dag}\vec{u}_{[\mathcal{T}_j]}
\end{array}\right]\right\|_{\mL(G_j)}=\left\|\left[\begin{array}{c}
     \vec{0}_{[V(G_j)\setminus\mathcal{T}_j]}  \\
     \vec{v}_1-\mSC(G_j,\mathcal{T}_j)^{\dag}\vec{u}_{[\mathcal{T}_j]} 
\end{array}\right]\right\|_{\mZ_{j,1}\mL(G_j)\mZ_{j,1}^T}.
\end{aligned}
\end{equation}
By triangle inequality, (\ref{bound_induction_j_right_part2}) gives 
\begin{equation}
\label{412_0_v1-SCGjTjuTj}
\begin{aligned}
&\left\|\left[\begin{array}{c}
     \vec{0}_{[V(G_j)\setminus\mathcal{T}_j]}  \\
     \vec{v}_1-\mSC(G_j,\mathcal{T}_j)^{\dag}\vec{u}_{[\mathcal{T}_j]} 
\end{array}\right]\right\|_{\mZ_{j,1}\mL(G_j)\mZ_{j,1}^T}\\
\le&\left\|\left[\begin{array}{c}
     \vec{0}_{[V(G_j)\setminus\mathcal{T}_j]}  \\
     \vec{v}_1-\mSC(G_{j+1},\mathcal{T}_j)^{\dag}\vec{u}_{[\mathcal{T}_j]} 
\end{array}\right]\right\|_{\mZ_{j,1}\mL(G_j)\mZ_{j,1}^T}\\
+&\left\|\left[\begin{array}{c}
     \vec{0}_{[V(G_j)\setminus\mathcal{T}_j]}  \\
     \mSC(G_{j+1}, \mathcal{T}_j)^{\dag}\vec{u}_{[\mathcal{T}_j]}-\mSC(G_j,\mathcal{T}_j)^{\dag}\vec{u}_{[\mathcal{T}_j]}
\end{array}\right]\right\|_{\mZ_{j,1}\mL(G_j)\mZ_{j,1}^T},
\end{aligned}
\end{equation}
in which 
\begin{equation}
\label{0_v1_Zj1}
\begin{aligned}
    &\left\|\left[\begin{array}{c}
     \vec{0}_{[V(G_j)\setminus\mathcal{T}_j]}  \\
     \vec{v}_1-\mSC(G_{j+1},\mathcal{T}_j)^{\dag}\vec{u}_{[\mathcal{T}_j]} 
\end{array}\right]\right\|_{\mZ_{j,1}\mL(G_j)\mZ_{j,1}^T}\\
=&\left\|\left[\begin{array}{c}
     \vec{0}_{[V(G_j)\setminus\mathcal{T}_j]}  \\
     \vec{v}_1-\mP(\mL(G_{j+1})^{\dag}))_{[\mathcal{T}_j, \mathcal{T}_j]}\mP\vec{u}_{[\mathcal{T}_j]} 
\end{array}\right]\right\|_{\mZ_{j,1}\mL(G_j)\mZ_{j,1}^T}\\
=&\left\|\left[\begin{array}{c}
     \vec{0}_{[V(G_j)\setminus\mathcal{T}_j]}  \\
     \vec{v}_1-\mP(\mL(G_{j+1})^{\dag}))_{[\mathcal{T}_j, \mathcal{T}_j]}\vec{u}_1 
\end{array}\right]\right\|_{\mZ_{j,1}\mL(G_j)\mZ_{j,1}^T}\\
=&\left\|\left[\begin{array}{c}
     \vec{0}_{[V(G_j)\setminus\mathcal{T}_j]}  \\
    \mP\vec{v}_{[\mathcal{T}_j]}-\mP\left(\mL(G_{j+1})^{\dag}\left[\begin{array}{c}
         \vec{0}  \\
         \vec{u}_1 
    \end{array}\right]\right)_{[\mathcal{T}_j]}
\end{array}\right]\right\|_{\mZ_{j,1}\mL(G_j)\mZ_{j,1}^T}\\
=&\left\|\vec{v}_{[\mathcal{T}_j]}-\left(\mL(G_{j+1})^{\dag}\left[\begin{array}{c}
         \vec{0}  \\
         \vec{u}_1 
    \end{array}\right]\right)_{[\mathcal{T}_j]}\right\|_{\mP(\mZ_{j,1}\mL(G_j)\mZ_{j,1}^T)_{[\mathcal{T}_j, \mathcal{T}_j]}\mP},
\end{aligned}
\end{equation}
where $\mP$ is the projection matrix of the space spanned by $\mSC(G_{j+1}, \T_j)$. Furthermore, the vectors $\vec{u}_1$ and~$\vec{v}_1$ are the projections of the vectors $\vec{u}_{[\T_j]}$ and $\vec{v}_{[\T_j]}$ onto $\mP$ respectively. 

By the given condition,  we have that 
\begin{align}
\label{412_LGj_inv}
\mL(G_j)^{\dag}\approx_{\epsilon}\mZ_{j,1}^{T}\left[\begin{array}{cc}
         \mZ_{j,2}&0  \\
         0&\mSC(G_j,\mathcal{T}_j)^{\dag} 
    \end{array}\right]\mZ_{j,1}.
\end{align}    
Multiplying the both sides of the LHS and RHS of (\ref{412_LGj_inv}) by $\mL(G_j)$ gives 
\begin{align}
\label{412_LGj}
\mL(G_j)\approx_{\epsilon}\mL(G_j)\mZ_{j,1}^{T}\left[\begin{array}{cc}
         \mZ_{j,2}&0  \\
         0&\mSC(G_j,\mathcal{T}_j)^{\dag} 
    \end{array}\right]\mZ_{j,1}\mL(G_j).
\end{align}
Multiplying the left (resp. right) side of the LHS and RHS of (\ref{412_LGj}) by $\mZ_{j,1}$ (resp. $\mZ_{j,1}^T$) gives 
\begin{align}
\label{412_Zj1LGjZj1T}
\mZ_{j,1}\mL(G_j)\mZ_{j,1}^T\approx_{\epsilon}\mZ_{j,1}\mL(G_j)\mZ_{j,1}^{T}\left[\begin{array}{cc}
         \mZ_{j,2}&0  \\
         0&\mSC(G_j,\mathcal{T}_j)^{\dag} 
    \end{array}\right]\mZ_{j,1}\mL(G_j)\mZ_{j,1}^T,
\end{align}
which implies that 
\begin{align}
\left[\begin{array}{cc}
         \mZ_{j,2}&0  \\
         0&\mSC(G_j,\mathcal{T}_j)^{\dag} 
    \end{array}\right]\approx_{\epsilon}(\mZ_{j,1}\mL(G_j)\mZ_{j,1}^T)^{\dag}
\end{align}
and
\begin{align}
\label{412_Zj1LGjZj1T_epsilon}
\mZ_{j,1}\mL(G_j)\mZ_{j,1}^T\approx_{\epsilon}\left[\begin{array}{cc}
         \mZ_{j,2}^{\dag}&0  \\
         0&\mSC(G_j,\mathcal{T}_j) 
    \end{array}\right].
\end{align}
Moreover, (\ref{412_Zj1LGjZj1T_epsilon}) gives 
\begin{align}
\label{412_Zj1LGjZj1T_Tj}
(\mZ_{j,1}\mL(G_j)\mZ_{j,1}^T)_{[\mathcal{T}_j,\mathcal{T}_j]}\approx_{\epsilon}\mSC(G_j, \mathcal{T}_j).
\end{align}
Multiplying the both sides of the LHS and RHS of (\ref{412_Zj1LGjZj1T_Tj}) by $\mP$ gives 
\begin{align}
\label{412_PZj1LGjZj1T}
\mP(\mZ_{j,1}\mL(G_j)\mZ_{j,1}^T)_{[\mathcal{T}_j,\mathcal{T}_j]}\mP\approx_{\epsilon}\mSC(G_j, \mathcal{T}_j).
\end{align}
In addition, $\mSC(G_{j+1}, \T_j)\approx_{\epsilon}\mSC(G_j, \T_j)$. Combining with (\ref{412_PZj1LGjZj1T}), we have  
\begin{align}
\label{412_PZj1LGjZj1T_2epsilon}
\mP(\mZ_{j,1}\mL(G_j)\mZ_{j,1}^T)_{[\mathcal{T}_j,\mathcal{T}_j]}\mP\approx_{2\epsilon}\mSC(G_{j+1}, \mathcal{T}_j).
\end{align}

Getting back to (\ref{0_v1_Zj1}), by (\ref{412_PZj1LGjZj1T_2epsilon}) we have that 
\begin{equation}
\label{bound_10_part1}
\begin{aligned}
&\left\|\vec{v}_{[\mathcal{T}_j]}-\left(\mL(G_{j+1})^{\dag}\left[\begin{array}{c}
         \vec{0}  \\
         \vec{u}_1 
    \end{array}\right]\right)_{[\mathcal{T}_j]}\right\|_{\mP(\mZ_{j,1}\mL(G_j)\mZ_{j,1}^T)_{[\mathcal{T}_j, \mathcal{T}_j]}\mP}\\
\le&e^{\epsilon}\left\|\vec{v}_{[\mathcal{T}_j]}-\left(\mL(G_{j+1})^{\dag}\left[\begin{array}{c}
         \vec{0}  \\
         \vec{u}_1 
    \end{array}\right]\right)_{[\mathcal{T}_j]}\right\|_{\mSC(G_{j+1}, \mathcal{T}_j)}.
\end{aligned}
\end{equation}
Combining (\ref{bound_10_part1}) with (\ref{v_Tj_le_induction_solution}) gets 
\begin{align}
\label{412_norm_vTj-LGj+10u1}
\left\|\vec{v}_{[\mathcal{T}_j]}-\left(\mL(G_{j+1})^{\dag}\left[\begin{array}{c}
         \vec{0}  \\
         \vec{u}_1 
    \end{array}\right]\right)_{[\mathcal{T}_j]}\right\|_{\mP(\mZ_{j,1}\mL(G_j)\mZ_{j,1}^T)_{[\mathcal{T}_j, \mathcal{T}_j]}\mP}\le 2(t-j)\epsilon e^{\epsilon}\left\|\left[\begin{array}{c}
     \vec{0}  \\
     \vec{u}_1 
\end{array}\right]\right\|_{\mL(G_{j+1})^{\dag}}.
\end{align}
Now consider bounding $\left\|\left[\begin{array}{c}
     \vec{0} \\
     \vec{u}_1
\end{array}\right]\right\|_{\mL(G_{j+1})^{\dag}}$,
\begin{align}
\label{412_0u1_LGj+1}
\begin{aligned}
    \left\|\left[\begin{array}{c}
     \vec{0} \\
     \vec{u}_1
\end{array}\right]\right\|_{\mL(G_{j+1})^{\dag}}&=\left\|\vec{u}_1\right\|_{(\mL(G_{j+1})^{\dag})_{[\mathcal{T}_j, \mathcal{T}_j]}}=\left\|\mP\vec{u}_{[\mathcal{T}_j]}\right\|_{(\mL(G_{j+1})^{\dag})_{[\mathcal{T}_j, \mathcal{T}_j]}}\\
&=\left\|\vec{u}_{[\mathcal{T}_j]}\right\|_{\mP(\mL(G_{j+1})^{\dag})_{[\mathcal{T}_j, \mathcal{T}_j]}\mP}=\left\|\vec{u}_{[\mathcal{T}_j]}\right\|_{\mSC(G_{j+1}, \mathcal{T}_j)^{\dag}}\\
&\le e^{\epsilon/2}\left\|\vec{u}_{[\mathcal{T}_j]}\right\|_{\mSC(G_{j}, \mathcal{T}_j)^{\dag}}=e^{\epsilon/2}\left\|\vec{u}\right\|_{\mA},
\end{aligned}
\end{align}
where $\mA=\left[\begin{array}{cc}
     0&0  \\
     0&\mSC(G_j, \mathcal{T}_j)^{\dag}
\end{array}\right]\preceq\left[\begin{array}{cc}
     \mZ_{j,2}&0  \\
     0&\mSC(G_j, \mathcal{T}_j)^{\dag}
\end{array}\right]$, 
then 
\begin{equation}
\label{bound_10_part2}
\begin{aligned}
    \left\|\left[\begin{array}{c}
     \vec{0} \\
     \vec{u}_1
\end{array}\right]\right\|_{\mL(G_{j+1})^{\dag}}\le e^{\epsilon/2}\left\|\mZ_{j,1}\vec{b}\right\|_{\mA}=e^{\epsilon/2}\left\|\vec{b}\right\|_{\mZ_{j,1}^T\mA\mZ_{j,1}}\le e^{\epsilon}\left\|\vec{b}\right\|_{\mL(G_j)^{\dag}}.
\end{aligned}
\end{equation}
Combining (\ref{412_norm_vTj-LGj+10u1}), (\ref{bound_10_part2}) with (\ref{0_v1_Zj1}), we have that 
\begin{equation}
\label{bound_9_part1}
\left\|\left[\begin{array}{c}
     \vec{0}_{[V(G_j)\setminus\mathcal{T}_j]}  \\
     \vec{v}_1-\mSC(G_{j+1},\mathcal{T}_j)^{\dag}\vec{u}_{[\mathcal{T}_j]} 
\end{array}\right]\right\|_{\mZ_{j,1}\mL(G_j)\mZ_{j,1}^T}\le 2(t-j)\epsilon e^{2\epsilon}\left\|\vec{b}\right\|_{\mL(G_j)^{\dag}}.
\end{equation}

Another item in (\ref{412_0_v1-SCGjTjuTj}) is 
\begin{equation}
\label{bound_9_part2}
\begin{aligned}
&\left\|\left[\begin{array}{c}
     \vec{0}_{[V(G_j)\setminus\mathcal{T}_j]}  \\
     \mSC(G_{j+1}, \mathcal{T}_j)^{\dag}\vec{u}_{[\mathcal{T}_j]}-\mSC(G_j, \mathcal{T}_j)^{\dag}\vec{u}_{[\mathcal{T}_j]} 
\end{array}\right]\right\|_{\mZ_{j,1}\mL(G_j)\mZ_{j,1}^T}\\
=&\left\|\mSC(G_{j+1}, \mathcal{T}_j)^{\dag}\vec{u}_{[\mathcal{T}_j]}-\mSC(G_j, \mathcal{T}_j)^{\dag}\vec{u}_{[\mathcal{T}_j]}\right\|_{(\mZ_{j,1}\mL(G_j)\mZ_{j,1}^T)_{[\mathcal{T}_j, \mathcal{T}_j]}}. 
\end{aligned}
\end{equation}
By (\ref{412_Zj1LGjZj1T_Tj}), we have 
\begin{equation}
\label{412_SCGj+1GjTj}
\begin{aligned}
&\left\|\mSC(G_{j+1}, \mathcal{T}_j)^{\dag}\vec{u}_{[\mathcal{T}_j]}-\mSC(G_j, \mathcal{T}_j)^{\dag}\vec{u}_{[\mathcal{T}_j]}\right\|_{(\mZ_{j,1}\mL(G_j)\mZ_{j,1}^T)_{[\mathcal{T}_j, \mathcal{T}_j]}}\\
\le&e^{\epsilon/2}\left\|\mSC(G_{j+1}, \mathcal{T}_j)^{\dag}\vec{u}_{[\mathcal{T}_j]}-\mSC(G_j, \mathcal{T}_j)^{\dag}\vec{u}_{[\mathcal{T}_j]}\right\|_{\mSC(G_j, \mathcal{T}_j)}. 
\end{aligned}
\end{equation}
By the fact $\mSC(G_{j+1}, \T_j)\approx_{\epsilon}\mSC(G_j, \T_j)$ and applying Lemma \ref{lem:ApproxToError}, we have 
\begin{align}
\label{412_SCGj+1GjTjuTj}
    \left\|\mSC(G_{j+1}, \mathcal{T}_j)^{\dag}\vec{u}_{[\mathcal{T}_j]}-\mSC(G_j, \mathcal{T}_j)^{\dag}\vec{u}_{[\mathcal{T}_j]}\right\|_{\mSC(G_j, \mathcal{T}_j)}\le\epsilon\left\|\vec{u}_{[\mathcal{T}_j]}\right\|_{\mSC(G_{j}, \mathcal{T}_j)^{\dag}}.
\end{align}
Recall that in (\ref{412_0u1_LGj+1}) and (\ref{bound_10_part2}) we have 
\begin{align}
\label{412_uTj_SCGjTj}
\left\|\vec{u}_{[\mathcal{T}_j]}\right\|_{\mSC(G_{j}, \mathcal{T}_j)^{\dag}}=\left\|\vec{u}\right\|_{\mA}=\left\|\mZ_{j,1}\vec{b}\right\|_{\mA}=\left\|\vec{b}\right\|_{\mZ_{j,1}^T\mA\mZ_{j,1}}\le e^{\epsilon/2}\left\|\vec{b}\right\|_{\mL(G_j)^{\dag}}.
\end{align}
Combining (\ref{bound_9_part2}), (\ref{412_SCGj+1GjTj}), (\ref{412_SCGj+1GjTjuTj}) and (\ref{412_uTj_SCGjTj}) gives
\begin{align}
\label{412_0_SCGj+1TjuTj_SCGjTjuTj}
    \left\|\left[\begin{array}{c}
     \vec{0}_{[V(G_j)\setminus\mathcal{T}_j]}  \\
     \mSC(G_{j+1}, \mathcal{T}_j)^{\dag}\vec{u}_{[\mathcal{T}_j]}-\mSC(G_j, \mathcal{T}_j)^{\dag}\vec{u}_{[\mathcal{T}_j]} 
\end{array}\right]\right\|_{\mZ_{j,1}\mL(G_j)\mZ_{j,1}^T}\le\epsilon e^{\epsilon}\left\|\vec{b}\right\|_{\mL(G_j)^{\dag}}.
\end{align}

Finally, combining (\ref{bound_9_part1}), (\ref{412_0_SCGj+1TjuTj_SCGjTjuTj}) with (\ref{bound_8}), we have that for the case $j$, 
\begin{align*}    \left\|\mZ_{t,1}^{T}\left[\begin{array}{cc}
         \mZ_{t,2}\vec{u}_{[V(G_t)\setminus\mathcal{T}_t]}  \\
         \vec{v} 
    \end{array}\right]-\mL(G_t)^{\dag}\vec{b}\right\|_{\mL(G_t)}&\le\left[2(t-j)\epsilon e^{2\epsilon}+\epsilon e^{\epsilon}+\epsilon)\right]\left\|\vec{b}\right\|_{\mL(G_j)^{\dag}}\\
    &\lesssim 2(t-j+1)\epsilon\left\|\vec{b}\right\|_{\mL(G_j)^{\dag}}. 
\end{align*}
Since $t\le\log{n}$, we can prove that 
\[
\left\|\mZ_{1,1}^{T}\left[\begin{array}{cc}
         \mZ_{1,2}\vec{u}_{[V(G_1)\setminus\mathcal{T}_1]}  \\
         \vec{v} 
    \end{array}\right]-\mL(G_1)^{\dag}\vec{b}\right\|_{\mL(G_1)}\le 2\epsilon\log{n}\left\|\vec{b}\right\|_{\mL(G_1)^{\dag}},
\]
that is,
\[
\left\|\vec{x}-\mL(G)^{\dag}\vec{b}\right\|_{\mL(G)}\le 2\epsilon\log{n}\left\|\vec{b}\right\|_{\mL(G)^{\dag}}.
\]
\end{proof}

\begin{algorithm}[h]
\caption{Distributed Laplacian Solver}
\SetKwProg{Proc}{procedure}{}{}
\Proc{$\Solve(G)$}{
	$G' \assign \Sparsify(G)$\\
	$(G_1, \mZ_{1,1}, \mZ_{1,2}, \T_1, G_2) \assign \Ultrasparsify (G', k)$ \\
	$\{(G_i, \mZ_{i, 1}, \mZ_{i, 2}, \T_i)\}_{i=2}^t \assign \BuildChain (G_2, d, \epsilon, k)$\\
	solve $\mL(G)\vec x = \vec b$ by preconditioned Chebyshev with $G_1$ as preconditioner s.t.
	$\mL(G_1)\vec {y} = \vec{c}$  is approximated by $\PseudoinverseMulti (\{(G_i, \mZ_{i, 1}, \mZ_{i, 2}, \T_i)\}_{i = 1}^t, \vec {c})$.\\
}
\Proc{$\BuildChain(G, d, \epsilon, k)$}{
	\If{$|V(G)| \leq k$} { return \;}
	$(\mZ_1, \mZ_2, C) \assign \Eliminate (G, d, \epsilon)$. \\
	$H \assign \ApproxSC (G, C, \epsilon)$\\
	return $(G, \mZ_1, \mZ_1, C) \cup \BuildChain (H, d, \epsilon, k)$\\
}
\end{algorithm}

\begin{proof}(of Theorem~\ref{thm:Main})
The parameters are set as follow:
\begin{itemize}
\item $\epsilon = (\frac{1}{\log \overline n})^{10}$.
\item $d = (\log \log \overline n)^2$
\item $k = 2^{(\log \overline n)^{2/3}}$
\end{itemize}

The correctness of the algorithm is obtained by Lemma~\ref{lemma:Communication}, Lemma~\ref{lem:UltraSparsify}, Lemma~\ref{lem:Elimination}, Theorem~\ref{thm:MinorSC} and Lemma~\ref{lem:applying_pseudoinverse}.

Now we bound the number of rounds required.
By Lemma~\ref{lem:Elimination}, Theorem~\ref{thm:MinorSC} we have that the Schur-complement chain obtained for graph $G$ satisfying the following conditions
\[|V(G)| = |V(G_1)| \geq |V(G_2)| 2^{O(\sqrt{\log n \log \log n})} / k = |V(G_2)| / k^{1 - o(1)}\]
and
\[|V(G_i)| / |V(G_{i+1})| \geq \epsilon^{-2}\log^2 n / 0.99^d = 2^{\Theta((\log \log \overline n)^2)}.\]
Hence, the Schur-complement chain obtained by the \BuildChain algorithm is a $(2^{\Theta((\log \log \overline n)^2)}, \eps)$-Schur-complement chain of length $O(\log \overline n / (\log \log \overline n)^2)$.
By Lemma~\ref{lem:applying_pseudoinverse}, \PseudoinverseMulti for the
Schur-complement chain takes $O(\rho {\overline n}^{o(1)} (\overline n^{1/2} + D))$
rounds.

Let $f(n, \rho)$ denote the number of rounds required by Algorithms \Solve and~\BuildChain
on a graph with $n$ vertices that is $\rho$-minor distributes to $\overline G$, and let
$g(n, \rho)$ denote the number of rounds of \BuildChain with $n$ vertices that
is $\rho$-minor distributes to $\overline G$.

Since preconditioned Chebyshev needs to call the Laplacian solver of the preconditioner $O(\sqrt{k})$ times,
by Lemma~\ref{lemma:Communication}, Lemma~\ref{lem:UltraSparsify}, Lemma~\ref{lem:Elimination}, Theorem~\ref{thm:MinorSC}, and Lemma~\ref{lem:applying_pseudoinverse} we have
\begin{align*}
f\left(n, \rho\right)
&=
O\left(\left(\log^{14} n \log^{60} \overline n\right)^{\left(\log \log \overline n\right)^2}\left(\rho \overline n^{1/2} \log \overline n + D\right)\right)
+
\\ & \qquad \quad
g\left(n / k^{1-o\left(1\right)}, \rho\right)
+
O\left(\sqrt{k}\rho {\overline n}^{o\left(1\right)} \left(\overline n^{1/2}
+ D\right)\right) \\
&=
O\left(\overline n^{o\left(1\right)}
\left(\rho \overline n^{1/2} \log \overline n + D\right)\right)
+
g\left(n / k^{1-o\left(1\right)}, \rho\right).
\end{align*}
and
\begin{align*}
g(n, \rho)
&=
O\left(\left(\log^{14} n \log^{60} \overline n\right)
  ^{\left(\log \log \overline n\right)^2}
     \left(\rho \overline n^{1/2} \log \overline n + D\right)\right)
+
f\left(n, 2\rho\right) \log^{10} n \cdot \epsilon^{-3}
+
g\left(n / k, 2\rho\right)\\
&=
O\left(\overline n^{o\left(1\right)}\left(\rho \overline n^{1/2}  + D\right)\right)
+
\text{polylog}\left(\overline n\right) f\left(n, 2\rho\right)
+
g\left(n / k, 2\rho\right).
\end{align*}
Since the depth of the recursion is $O(\log \overline n / (\log \log \overline n)^2)$,
the overall increase in congestion is at most
\[
2^{O\left(\log \overline n / \left(\log \log \overline{n}\right)^2\right)}\rho
\leq
n^{o\left( 1\right)} \rho
\]
so all the graphs constructed $\overline n^{o(1)} \rho$-minor distribute into $\overline G$.

Hence, the algorithm $\Solve$
takes $\rho \overline n^{o(1)} (\overline n^{1/2} + D)$ rounds.
\end{proof}

\section{Minor Schur Complement}
\label{sec:SC}
In this section we give the algorithm for constructing minor based approximate Schur complements. Due to the recursive invocation of this routine and solver constructions in Section~\ref{sec:ClownFiesta}, we can treat the calls to solvers for SDD or Laplacian matrices as black-boxes. The formal guarantees of our constructions are stated in Theorem~\ref{thm:MinorSC}, which is restated below.
\MinorSC*

Before delving into technical details, we first discuss the high-level connections and differences between our algorithm and that of~\cite{LS18}. 

\paragraph{Comparison to~\cite{LS18}} Our starting point is the same as~\cite{LS18}, that is, randomly contracting an edge with probability being equal to its leverage score (and deleting otherwise) is exactly a matrix martingale on the spectral form of the graph. It gives a natural algorithm -- iteratively computing leverage scores of
edges and sampling them until the variance having been accumulated. The correctness of the algorithm is proved via the matrix martingale concentration inequality.

The main difference lies in the way of obtaining a nearly-linear running time. The leverage scores of all the edges keep changing as some edges get sampled, so a fast algorithm is needed to do better than recomputing the sampling probabilities of all the edges after each edge gets sampled. Li and Schild~\cite{LS18} address this issue by showing that a random spanning tree has the correct marginals, and use the fast random spanning tree sampling algorithm~\cite{LS18,ALGV20:arxiv} to obtain such trees. While there are distributed algorithms for sampling spanning trees from unweighted graphs~\cite{DNPT13}, partial states of elimination algorithms, namely Schur complements, are naturally weighted. Furthermore, we are unable to directly extend fast random walk simulations to  weighted graphs due to the higher congestion of weighted random walks. 

Instead, we devise a parallel version of this algorithm based on sampling large subsets of edges independently. We compute a large subset of \textit{steady edges} $Z$ that are mostly uncorrelated, which is obtained by the localization of electrical flows~\cite{SRS18}. We then identify such subsets, as well as compute all their effective resistances, using standard sketching methods that are also highly parallel. By ensuring that the size of these sets is at least $1/\mathrm{poly}(\log{n})$ of the total number of edges, we are able to ensure the rapid convergence of this process. 

In general, we track the cost of our algorithms via three quantities. The first is the number of Laplacian solvers to principal minors of $\mL(G)$ that we must call, and the second is how many additional rounds of communication between neighbors of $G$ that are necessary, each of which can be simulated in $O(\rho\sqrt{\overline{n}}\log{\overline{n}}+D)$ rounds in $\overline{G}$ by Lemma~\ref{lemma:Communication}. Finally, we must also ensure that the local computations on vertices $v\in V(G)$ are actually simple minimum/sum aggregations, as each vertex  $v\in V(G)$ actually corresponds to a connected component in $\overline{G}$. These can also be simulated in $O(\rho\sqrt{\overline{n}}\log{\overline{n}}+D)$ rounds in~$\overline{G}$ by Lemma \ref{lemma:Communication}. We note that the computations for solving Laplacian systems and computing leverage scores, etc. only involve matrix-vector multiplications and sampling Bernoulli/Cauchy random variables, which can all be aggregated in a distributed manner. 

\paragraph{Distributed storage conventions.} In this section, we work with graph $G$ that $\rho$-minor distributes into the original graph/communication network $\overline{G}$ and is stored distributedly (see Definition~\ref{def:Minor}). We work with vertex vectors $\vec{x}\in\mathbb{R}^{V(G)}$. In this case, for a vertex $v\in V(G)$ (corresponding to a connected component in $\overline{G}$), we assume that the root $V_{map}^{G\rightarrow\overline{G}}(v)\in V(\overline{G})$ stores the value of $x_v$. We also work with edge vectors $\vec{w}\in\mathbb{R}^{E(G)}$ of edge resistances or leverage scores. For an edge $e^G=(u^G, v^G)\in E(G)$, it corresponds to an edge in $\overline{G}$ with endpoints $u^H$ and $v^H$ that store the weight $w_{e^G}$. When an algorithm is said to compute vertex vectors or edge vectors, it means that these conditions are satisfied. 

The remaining part of this section is organized as follows. 
\begin{enumerate}
    \item In Section \ref{sec:Schur_sparsification}, we give the formal definition of steady edges, and present the algorithm for minor based approximate Schur complement. 
    \item In Section \ref{sec:steady edges}, we give the algorithm for finding the set of steady edges.
    \item In Section \ref{sec:matrix martingale analysis}, we prove the correctness of the algorithm in Section \ref{sec:steady edges} via matrix martingales, and Theorem \ref{thm:MinorSC}. 
\end{enumerate}
\subsection{Sparsification Algorithm}
\label{sec:Schur_sparsification}
We start by defining the key notion \textit{steady edges}, which are edges that intuitively do not interact with each other much. Here, we emphasize that these steady edges are stochastic, not deterministic. 
\begin{definition}
\label{def:alpha_delta_steady}
A stochastic subset of edges $Z\subseteq E(H)
$ is $(\alpha, \delta)$-steady if 
\begin{enumerate}
    \item \label{condition:quadratic_form}(Quadratic form) $\mathbb{E}\left[\sum_{e\in Z}r_e^{-1}\vec{b}_e\vec{b}_e^T\right]\preceq\alpha\mL(H)$;
    \item \label{condition:localization}(Localization) For each edge $e\in Z$, $\sum_{f\neq e\in Z}\frac{|\vec{b}_e^T\mL(H)^{\dag}\vec{b}_f|}{\sqrt{r_{e}r_f}}\le\delta$;
    \item \label{condition:variance}(Variance) For each edge $e\in Z$, 
    \begin{align*}
        r_e^{-1}\vec{b}_e^T\mL(H)^{\dag}\left[\begin{array}{cc}
             0&0  \\
             0&\mSC(H, \T) 
        \end{array}\right]\mL(H)^{\dag}\vec{b}_e\le\frac{18|\T|}{|E(H)|}.
    \end{align*}
\end{enumerate}
\end{definition}

Intuitively, the \textit{Quadratic form} constraint guarantees that no edge is picked in the set of steady edges with a high probability. The \textit{Localization} constraint bounds the ``correlation'' between edges, by restricting the electrical flow of each edge $e$ putting on the remaining edges in $Z$. Finally, the \textit{Variance} constraint says that the induced leverage score of edge $e$ on the Schur complement is bounded, and allows us to control the variance in the matrix martingale analysis. 

Now we describe the algorithm for computing a minor Schur complement. First, identify a set of steady edges and approximately compute their leverage scores by the Johnson-Lindenstrauss lemma. Then for each steady edge, contract it with probability being its approximate leverage score, and delete it otherwise. Repeat this process until the size of the resulting graph is small enough. The algorithm pseudocode is shown in algorithm \ApproxSC. For this algorithm we have the following theorem. 

\begin{algorithm}[ht]
\caption{Finding sparsifier of Schur complement onto terminals, but with extra Steiner vertices\label{algo:sparsify}}
\SetKwProg{Proc}{procedure}{}{}
\Proc{$\ApproxSC(G, \T, \eps)$}{
Initialize $G^{(0)} \assign G$ and $i \assign 0$. \\
Set $\d \assign \frac{\eps}{C \log^2 m}$ with $m=|E(G)|$. \Comment{$C$ is a large constant}\\
\While{$|E(G^{(i)})| \ge \frac{C|\T|\log^2 m}{\eps^2}$}{
\label{ln:ApproxSCWhile}
$H^{(i)} \assign \Split(G^{(i)}, \LevApx(e, G^{(i)}, 0.01))$. \Comment{Lemmas
\ref{lem:split} and \ref{lem:LevApx}} \\
$Z^{(i)} \assign \FindSteady(H^{(i)}, \T, \d).$ \Comment{Lemma \ref{lemma:FindSteady}}\label{line:local} \\
For each edge $e \in Z^{(i)}$, set $p_e \assign \LevApx(e, H^{(i)}, \d)$.
\Comment{Lemma \ref{lem:LevApx}}\label{line:LeverageEstimate} \\
For each edge $e \in Z^{(i)}$, contract $e$ with probability $p_e$ and delete $e$ with probability $1-p_e$. Perform the contractions and deletions via Corollary~\ref{corollary:ActualMinor} and let the resulting graph be $I^{(i)}$. \\
$G^{(i+1)} \assign \Unsplit(I^{(i)})$. \Comment{Lemma \ref{lem:split}} \\
$i \assign i+1$. \\
}
$H\assign G^{(i)}$.\\
\Return $H$.
}
\end{algorithm}

\begin{theorem}
\label{thm:approxsc_algebra}
Given a graph $G$ with $m$ edges and a set of terminals $\T\subseteq V(G)$, and parameter $\epsilon\in (0,1)$, the algorithm $\ApproxSC(G, \T, \epsilon)$ returns a graph $H$ such that $|E(H)|\le O(|\T|\epsilon^{-2}\log^2{m})$ and $\mSC(H, \T)\approx_{\epsilon}\mSC(G, \T)$ with probability at least $1-1/\mathrm{poly}(m)$. 
\end{theorem}
The proof of Theorem \ref{thm:approxsc_algebra} is deferred to Section \ref{sec:matrix martingale analysis}. Now we describe the subroutines in algorithm \ApproxSC. 

The \Split and sampling process depend on the leverage score of each edge. Instead of computing the leverage scores precisely, we use \LevApx to compute approximate leverage scores following the standard random projection scheme devised by Spielman and Srivastava \cite{SS11}. Specifically, the subroutine \LevApx satisfies the following guarantees.  
\begin{lemma}[Approximate leverage scores]
\label{lem:LevApx}
Given a graph $G'$ that $\rho$-minor distributes into $\overline{G}$, an error parameter $\delta > 0$ and the distributed Laplacian solver \Solve, for each edge $e\in E(G')$, the algorithm $\LevApx(e, G', \delta)$ returns the approximation of $\mathrm{lev}_{G'}(e)=r_e^{-1}\vec{b}_e^T\mL(G')^{\dag}\vec{b}_e$ to within a factor of $1+\delta$ with high probability. Furthermore, it takes 
\begin{enumerate}
    \item $O(\delta^{-2}\log|V(G')|)$ calls to \Solve with accuracy $1/\mathrm{poly}(|V(G)|)$ on graphs that $\rho$-minor distribute into $\overline{G}$;
    \item An additional $O(\rho\delta^{-2}\sqrt{\overline{n}}\log{\overline{n}}\log{|V(G')|}+D)$ rounds of communication in $\overline{G}$. 
\end{enumerate}
\end{lemma}
Before proving Lemma \ref{lem:LevApx}, we present the Johnson-Lindenstrauss lemma which is essential for proving Lemma \ref{lem:LevApx}. 
\begin{lemma}[Johnson-Lindenstrauss Lemma]
\label{lem:JL}
Given $n$ vectors $\vec{v}_1, \cdots, \vec{v}_n\in\mathbb{R}^d$ and a parameter $\delta>0$, let $\mQ\in\mathbb{R}^{k\times d}$ with $k\ge 24\delta^{-2}\log{n}$ be a random $\pm 1/\sqrt{k}$ matrix with each entry being an independent Bernoulli random variable. Then with probability at least $1-1/n$,
\begin{align*}
\|\mQ(\vec{v}_i-\vec{v}_j)\|_2^2\approx_{\delta}\|\vec{v}_i-\vec{v}_j\|_2^2,
\end{align*}
for all $i,j\in[n]$.
\end{lemma}
\begin{proof}[Proof of Lemma \ref{lem:LevApx}]
Recall that the effective resistance of $e=(u,v)\in E(G')$ is defined by $\res_{G'}(e)=\vec{b}_e^T\mL(G')^{\dag}\vec{b}_e$. More specifically, we have 
\begin{align*}
    \res_{G'}(e)&=\vec{b}_e^T\mL(G')^{\dag}\vec{b}_e=\vec{b}_e^T\mL(G')^{\dag}\mL(G')\mL(G')^{\dag}\vec{b}_e\\  \tag{Setting $\mL(G')=\mB^T\mR^{-1}\mB$}
    &=\vec{b}_e^T\mL(G')^{\dag}\mB^T\mR^{-1}\mB\mL(G')^{\dag}\vec{b}_e\\
    &=\left\|\mR^{-1/2}\mB\mL(G')^{\dag}\vec{b}_e\right\|_2^2,
\end{align*}
which is equal to the squared Euclidean distance between the $u$-th and $v$-th column vectors of the matrix $\mR^{-1/2}\mB\mL(G')^{\dag}$. To compute the effective resistance for each edge, it suffices to compute the pairwise distances among the column vectors of matrix $\mR^{-1/2}\mB\mL(G')^{\dag}$. Applying Johnson-Lindenstruass lemma, we generate a random matrix $\mQ\in\mathbb{R}^{t\times|E(G')|}$ with $t=O(\delta^{-2}\log|V(G')|)$ such that with high probability 
\begin{align*}
\left\|\mQ\mR^{-1/2}\mB\mL(G')^{\dag}\vec{b}_e\right\|_2^2\approx_{\delta}\res_{G'}(e), 
\end{align*}
where computing the matrix $\mQ\mR^{-1/2}\mB\mL(G')^{\dag}$ requires matrix multiplication of $\mQ\mR^{-1/2}\mB$ and solving $O(\delta^{-2}\log|V(G')|)$ Laplacian linear systems. 

Now we implement the above operations in the distributed settings. First generate $\mQ$ on the endpoints of the images of the edges in $G'$. Note that $\mR^{-1/2}$ is a rescaling of the resistances of the edges, which are also stored together with their endpoints. Let each edge $e\in E(G')$ store the corresponding column of $\mQ\mR^{-1/2}$ on both of its endpoints, i.e., both endpoints of $E_{map}^{G'\rightarrow\overline{G}}(e)$ store $(\mQ\mR^{-1/2})_{:,e}$. Computing the matrix $\mQ\mR^{-1/2}\mB$ is reduced to computing the matrix-vector multiplication $\mQ\mR^{-1/2}\mB_{:,v}$ for each vertex $v\in V(G')$. Note that each edge can choose its direction arbitrarily, as the direction factor backs in when we apply the multiplication by $\vec{b}_e$ at the end, which is also a local step. Lemma \ref{lemma:Communication} allows us to perform this process in $O(\rho\delta^{-2}\sqrt{\overline{n}}\log{\overline{n}}\log{|V(G')|}+D)$ rounds. 
\end{proof}

The algorithm {\ApproxSC} requires that the leverage scores of all the edges are bounded away from $0$ and $1$, which can be done by the two subroutines, \Split and \Unsplit that have the following guarantees. 
\begin{lemma}[\Split and \Unsplit, see Proposition 3.4 and 3.5 in \cite{LS18}]
\label{lem:split}
Given a graph $G'$ that $\rho$-minor distributes into $\overline{G}$ and the approximate leverage score $\mathrm{lev}'_{G'}(e)\approx_{0.01}\mathrm{lev}_{G'}(e)$ for each $e\in E(G')$, the algorithm $\Split(G', \mathrm{lev}'_{G'}(e))$ returns a graph $H'$ in $\widetilde{O}(\rho(\sqrt{\overline{n}}+D))$ rounds such that 
\begin{enumerate}
    \item $H'$ is electrically equivalent to $G'$;
    \item $H'$ $2\rho$-minor distributes into $\overline{G}$;
    \item For each edge $e'\in E(H')$, $\mathrm{lev}_{H'}(e')\in[3/16, 13/16]$. 
\end{enumerate}
The algorithm \Unsplit returns a graph resulting from collapsing paths, parallel edges, and removing non-terminal leaves, along with a $\rho$-minor distribution into $\overline{G}$. 
\end{lemma}
\begin{proof}
In algorithm \Split, for each edge $e\in E(G')$, if $\mathrm{lev}'_{G'}(e)\le 1/2$, replace $e$ by a path of two edges $e_1$ and $e_2$ with resistance $r_e/2$; if $\mathrm{lev}'_{G'}(e)\ge 1/2$, replace $e$ by two parallel edges $e_1$ and $e_2$ with resistance $2r_e$. In the first case, both edges have leverage score $\mathrm{lev}_{H'}(e_1)=\mathrm{lev}_{H'}(e_2)=\frac{1}{2}+\frac{\mathrm{lev}_{G'}(e)}{2}$, in which $\mathrm{lev}_{G'}(e)\in\left[0,\frac{1}{2(1-0.01)}\right]$; in the second case, both edges have leverage score $\mathrm{lev}_{H'}(e_1)=\mathrm{lev}_{H'}(e_2)=\frac{\mathrm{lev}_{G'}(e)}{2}$, in which $\mathrm{lev}_{G'}(e)\in\left[\frac{1}{2(1+0.01)}, 1\right]$. It is easy to verify that each edge $e'\in E(H')$ has $\mathrm{lev}_{H'}(e')\in[3/16,13/16]$. 

The bound on the cost and embeddability follows since each edge is turned into a path of at most two edges. The new vertex can be placed at either endpoints of $E_{map}^{G\rightarrow\overline{G}}(e)$, and the congestion on both edge $e$ and the endpoints of $e$ goes up by a factor of $2$. This $2$-minor embedding of the new graph into $G'$ then meets the definition of Lemma \ref{lemma:MinorCompose}, which means that $H'$ $2\rho$-minor distributes into $\overline{G}$ with an overhead of $\widetilde{O}(\rho(\sqrt{\overline{n}}+D))$ rounds.

The execution of \Unsplit is straightforward because the resulting graph is a minor of $G'$, and all changes happen on $O(1)$ neighbors, and only involve local endpoints of edges of $G'$. Therefore, they can be implemented using $O(\rho)$ rounds of communications among neighbors of $G'$. 
\end{proof}
\subsection{Algorithm for Finding Steady Edges}
\label{sec:steady edges}
In this section, we present the subroutine \FindSteady, shown in algorithm \ref{algo:FindSteady}, that returns the set of steady edges in algorithm \ApproxSC. 
\begin{algorithm}[ht]
\caption{Given a graph $H$ with a set of terminals $\T$ and parameter $\d$, return the set of steady edges
\label{algo:FindSteady}}
\SetKwProg{Proc}{procedure}{}{}
\Proc{$\FindSteady(H, \T, \d)$}{
Set $\alpha \assign \frac{\d}{46\clocal \log^2{|E(H)|}}$. \\
For each $e \in E(H)$, let $v_e\assign\DiffApx(e, H, \T).$ \Comment{Lemma
	\ref{lem:DiffApx}}\label{line:difffinal} \\
For each $e \in E(H)$, let $s_e \assign \ColumnApx(e, H, E(H))$. \Comment{Lemma
	\ref{lem:ColumnApx}} \\
	$Z_1 \assign \{e \in E(H)\mid v_e \le 16|\T|/|E(H)|, \ s_e\le
	16\clocal
	\log^2{|E(H)|}\}$. \label{line:Z1_ve_se}\\
	Let $Z_2$ be the set of sampled edges from $Z_1$ such that each $e\in Z_1$ is sampled with probability $\alpha$.\\
	For each $e \in Z_2$, let $s_e'\assign \ColumnApx(e, H, Z_2)$. \Comment{Lemma
	\ref{lem:DiffApx}} \\
	$Z\assign \{e \in Z_2\mid s_e'\le \d/1.1 \}$. \label{line:localfinal} \\
	\Return $Z$.
}
\end{algorithm}

The algorithm \FindSteady has the following lemma. 
\begin{lemma}
\label{lemma:FindSteady}
Given a graph $H$ that $\rho$-minor distributes into $\overline{G}$, a set of terminals $\T\subseteq V(H)$ and constant $\delta\in(0,1)$, the algorithm $\FindSteady(H, \T, \delta)$ has access to the distributed Laplacian solver \Solve and returns an edge set $Z$ with at least $\alpha|E(H)|/2$ edges in expectation that is $(\alpha, \delta)$-steady. Furthermore, it takes 
\begin{enumerate}
    \item $O(\log^2{|V(H)|})$ calls to \Solve with $1/\mathrm{poly}(|V(H)|)$ error on graphs that $\rho$-minor distribute into $\overline{G}$;
    \item An additional $O((\rho\sqrt{\overline{n}}\log{\overline{n}}\log{|V(H)|}+D)\log{|V(H)|})$ rounds of communication in $\overline{G}$.
\end{enumerate}
\end{lemma}

Before proving Lemma \ref{lemma:FindSteady}, we first introduce the subroutines \DiffApx and \ColumnApx. 
\begin{lemma}[Difference sketch, Lemma 1.4 in \cite{LS18}]
\label{lem:DiffApx}
Given a graph $H$ that $\rho$-minor distributes into~$\overline{G}$ and a set of terminals $\T\subseteq V(H)$, for each edge $e\in E(H)$, the algorithm $\DiffApx(e, H, \T)$ returns an approximation to 
\begin{align*}
    r_e^{-1}\vec{b}_e^T\mL(H)^{\dag}\left[\begin{array}{cc}
         0&0  \\
         0&\mSC(H, \T) 
    \end{array}\right]\mL(H)^{\dag}\vec{b}_e
\end{align*}
within a factor of $1.1$ with high probability. Furthermore, it requires
\begin{enumerate}
    \item $O(\log{|V(H)|})$ calls to \Solve with accuracy $1/\mathrm{poly}(|V(H)|)$ on graphs that $\rho$-minor distribute into $\overline{G}$;
    \item An additional $O(\rho\sqrt{\overline{n}}\log{\overline{n}}\log{|V(H)|}+D)$ rounds of communication in $\overline{G}$. 
\end{enumerate}
\end{lemma}
\begin{proof}
By Lemma \ref{lemma:Inverse}, we have that 
\begin{align*}
    \left[\begin{array}{cc}
         0&0  \\
         0&\mSC(H, \T) 
    \end{array}\right]=\left[\begin{array}{cc}
         0&0  \\
         0&\mSC(H, \T) 
    \end{array}\right]\mL(H)^{\dag}\left[\begin{array}{cc}
         0&0  \\
         0&\mSC(H, \T) 
    \end{array}\right],
\end{align*}
and then
\begin{equation}
\label{eq:re_-1_be_LH_SCHT_LH_be}
\begin{aligned}
&r_e^{-1}\vec{b}_e^T\mL(H)^{\dag}\left[\begin{array}{cc}
         0&0  \\
         0&\mSC(H, \T) 
    \end{array}\right]\mL(H)^{\dag}\vec{b}_e\\
=&r_e^{-1}\vec{b}_e^T\mL(H)^{\dag}\left[\begin{array}{cc}
         0&0  \\
         0&\mSC(H, \T) 
    \end{array}\right]\mL(H)^{\dag}\left[\begin{array}{cc}
         0&0  \\
         0&\mSC(H, \T) 
    \end{array}\right]\mL(H)^{\dag}\vec{b}_e.
\end{aligned}
\end{equation}
Using the fact $\mL(H)^{\dag}=\mL(H)^{\dag}\mL(H)\mL(H)^{\dag}$ and setting $\mL(H)=\mB^T\mR^{-1}\mB$, (\ref{eq:re_-1_be_LH_SCHT_LH_be}) becomes 
\begin{align}
\label{quantity:re_-1_be_LH}
r_e^{-1}\vec{b}_e^T\mL(H)^{\dag}\left[\begin{array}{cc}
         0&0  \\
         0&\mSC(H, \T) 
    \end{array}\right]\mL(H)^{\dag}\mB^T\mR^{-1}\mB\mL(H)^{\dag}\left[\begin{array}{cc}
         0&0  \\
         0&\mSC(H, \T) 
    \end{array}\right]\mL(H)^{\dag}\vec{b}_e.
\end{align}
Formulating (\ref{quantity:re_-1_be_LH}) in another way, it becomes 
\begin{align}
\label{quantity:re_-1_l2norm_square}
r_e^{-1}\left\|\mR^{-1/2}\mB\mL(H)^{\dag}\left[\begin{array}{cc}
         0&0  \\
         0&\mSC(H, \T) 
    \end{array}\right]\mL(H)^{\dag}\vec{b}_e\right\|_2^2.
\end{align}
Combining (\ref{eq:re_-1_be_LH_SCHT_LH_be}), (\ref{quantity:re_-1_be_LH}) and (\ref{quantity:re_-1_l2norm_square}), we have 
\[
r_e^{-1}\vec{b}_e^T\mL(H)^{\dag}\left[\begin{array}{cc}
         0&0  \\
         0&\mSC(H, \T) 
    \end{array}\right]\mL(H)^{\dag}\vec{b}_e=r_e^{-1}\left\|\mR^{-1/2}\mB\mL(H)^{\dag}\left[\begin{array}{cc}
         0&0  \\
         0&\mSC(H, \T) 
    \end{array}\right]\mL(H)^{\dag}\vec{b}_e\right\|_2^2.
\]

By Lemma \ref{lem:JL}, we generate a random matrix $\mQ\in\mathbb{R}^{t\times|E(H)|}$ with $t=O(\log{|V(H)|})$ such that
\begin{align*}
&r_e^{-1}\left\|\mQ\mR^{-1/2}\mB\mL(H)^{\dag}\left[\begin{array}{cc}
         0&0  \\
         0&\mSC(H, \T) 
    \end{array}\right]\mL(H)^{\dag}\vec{b}_e\right\|_2^2\\
\approx_{0.1}&r_e^{-1}\vec{b}_e^T\mL(H)^{\dag}\left[\begin{array}{cc}
         0&0  \\
         0&\mSC(H, \T) 
    \end{array}
    \right]\mL(H)^{\dag}\vec{b}_e.
\end{align*}
Therefore, the round complexity is determined by computing the matrix 
\begin{align*}
    \mQ\mR^{-1/2}\mB\mL(H)^{\dag}\left[\begin{array}{cc}
         0&0  \\
         0&\mSC(H,\T) 
    \end{array}\right]\mL(H)^{\dag}.
\end{align*}

Firstly, we can compute the matrix $\mA_1=\mQ\mR^{-1/2}\mB$ as Lemma \ref{lem:LevApx}, which takes
\[
O(\rho\sqrt{\overline{n}}\log{\overline{n}}\log{|V(H)|}+D)
\]
rounds. Then multiplying by $\mL(H)^{\dag}$ requires calling $O(\log{|V(H)|})$ times \Solve on $\mL(H)$, which corresponds to the $t$ rows of the matrix $\mA_1$. 

Recall that
\begin{align*}
    \mSC(H,\T)=\mL(H)_{[\T,\T]}-\mL(H)_{[\T,V(H)\setminus\T]}\left(\mL(H)_{[V(H)\setminus\T, V(H)\setminus\T]}\right)^{-1}\mL(H)_{[V(H)\setminus\T, \T]},
\end{align*}
then we have 
\begin{align*}
    &\mA_1\left[\begin{array}{cc}
         0&0  \\
         0&\mSC(H,\T) 
    \end{array}\right]={\mA_{1}}_{[:,\T]}\mSC(H,\T)\\
    =&{\mA_{1}}_{[:,\T]}\mL(H)_{[\T,\T]}-{\mA_{1}}_{[:,\T]}\mL(H)_{[\T,V(H)\setminus\T]}\left(\mL(H)_{[V(H)\setminus\T, V(H)\setminus\T]}\right)^{-1}\mL(H)_{[V(H)\setminus\T, \T]},
\end{align*}
which can be computed by calling $t$ times \Solve on $\mL(H)_{[V(H)\setminus\T, V(H)\setminus\T]}$ and three matrix-matrix multiplications; each matrix-matrix consists of $t$ matrix-vector multiplications.

Let the matrix $\mA_2=\mA_1\left[\begin{array}{cc}
     0&0  \\
     0&\mSC(H,\T) 
\end{array}\right]$, then computing $\mA_2\mL(H)^{\dag}$ requires calling $t$ times \Solve on $\mL(H)$. 

Therefore, the algorithm $\DiffApx(e, H, \T)$ requires calling $O(\log{|V(H)|})$ times \Solve on graphs that $\rho$-minor distributes into $\overline{G}$  and additional $O(\rho\sqrt{\overline{n}}\log{\overline{n}}\log{|V(H)|}+D)$ rounds of communication. 
\end{proof}

\begin{lemma}[Analog to Proposition 4.3 in \cite{LS18}]
\label{lem:ColumnApx}
Given a graph $H$ that $\rho$-minor distributes into $\overline{G}$ and a subset $W\subseteq E(H)$, for each edge $e\in W$, the algorithm $\ColumnApx(e, H, W)$
returns an approximation to 
\begin{align*}
    \sum_{f\neq e\in W}\frac{|\vec{b}_e^T\mL(H)^{\dag}\vec{b}_f|}{\sqrt{r_{e}r_f}}
\end{align*}
within a factor of $1.1$ with high probability. Furthermore, it takes
\begin{enumerate}
    \item $O(\log^2{|V(H)|})$ calls to \Solve with accuracy $1/\mathrm{poly}(|V(H)|)$ on graphs that $\rho$-minor distribute into $\overline{G}$;
    \item An additional $O((\rho\sqrt{\overline{n}}\log{\overline{n}}\log{|V(H)|}+D)\log{|V(H)|})$ rounds of communication in $\overline{G}$.
\end{enumerate}
\end{lemma}
The proof of Lemma \ref{lem:ColumnApx} depends on the following $\ell_1$ sketch. 
\begin{lemma}[Theorem $3$ in \cite{Indyk06}]
\label{lem:l1_sketch}
Given an integer $d\ge 1$ and two constants $0<\delta, \epsilon<1$, there exists a matrix $\mC\in\mathbb{R}^{t\times d}$ with $t=O(\epsilon^{-2}\log(1/\delta))$ and an algorithm $\Recover(\vec{u}, d, \delta, \epsilon)$ such that 
\begin{enumerate}
    \item The entries of $\mC$ are independently sampled from a Cauchy distribution;
    \item For any vector $\vec{v}\in\mathbb{R}^d$, the algorithm $\Recover(\mC\vec{v}, d, \delta, \epsilon)$ outputs an estimator $r$ such that 
    \begin{align*}
        r\approx_{\epsilon}\|\vec{v}\|_1,
    \end{align*}
    with probability $1-\delta$. 
\end{enumerate}
\end{lemma}

\begin{proof}[Proof of Lemma \ref{lem:ColumnApx}]
We use the $\ell_1$ sketch in Lemma \ref{lem:l1_sketch} in a way analogous to the $\ell_2$ resistance estimation procedure in Lemma \ref{lem:LevApx}. Randomly partitioning the set $W$ such that $W=U\cup(W\setminus U)$ and $\Pr[e\in U]=1/2$ for each $e\in W$, then for each $e\in U$, we have 
\begin{align*}
    \mathbb{E}\left[\sum_{f\in W\setminus U}\frac{|\vec{b}_e^T\mL(H)^{\dag}\vec{b}_f|}{\sqrt{r_{e}r_f}}\right]=\frac{1}{2}\sum_{f\neq e\in W}\frac{|\vec{b}_e^T\mL(H)^{\dag}\vec{b}_f|}{\sqrt{r_{e}r_f}}.
\end{align*}
Denote the random variable $X_i=\sum_{f\in W\setminus U}\frac{|\vec{b}_e^T\mL(H)^{\dag}\vec{b}_f|}{\sqrt{r_{e}r_f}}$, and repeat $t_1=O(\log|V(H)|)$ times to obtain $X_1,\cdots,X_{t_1}$. Let $X$ be $X=\sum_{i=1}^{t_1}X_i$, then by Chernoff bound we have that with high probability, 
\begin{align*}
\frac{2X}{t_1}\approx_{0.1}\sum_{f\neq e\in W}\frac{|\vec{b}_e^T\mL(H)^{\dag}\vec{b}_f|}{\sqrt{r_{e}r_f}}. 
\end{align*}
Let $\mR_{W\setminus U}$ and $\mB_{W\setminus U}$ be the diagonal resistance matrix and incidence matrix restricted to $W\setminus U$, and $\vec{v}=r_e^{-1/2}\mR_{W\setminus U}^{-1/2}\mB_{W\setminus U}\mL(H)^{\dag}\vec{b}_e$. Then we have 
\begin{align*}
    \|\vec{v}\|_1=\sum_{f\in W\setminus U}\frac{|\vec{b}_e^T\mL(H)^{\dag}\vec{b}_f|}{\sqrt{r_{e}r_f}}. 
\end{align*}
By Lemma \ref{lem:l1_sketch}, setting the matrix $\mC\in\mathbb{R}^{t_2\times|W\setminus U|}$ with $t_2=O(\log{|V(H)|})$, then with probability $1-1/\mathrm{poly}(|V(H)|)$ the algorithm $\Recover(\mC\vec{v}, |W\setminus U|, 1/\mathrm{poly}(|V(H)|), 0.01)$ outputs a $0.01$-approximation of the quantity $\|\vec{v}\|_1$. 

Now we analyze the round complexity which is analogous to the $\ell_2$ sketch presented in the proof of Lemma \ref{lem:LevApx}. Computing the matrix $\mC\mR^{-1/2}_{W\setminus U}\mB_{W\setminus U}\mL(H)^{\dag}$ consists of computing $\mC\mR^{-1/2}_{W\setminus U}\mB_{W\setminus U}$, which takes $O(\rho\sqrt{\overline{n}}\log{\overline{n}}\log{|V(H)|}+D)$ rounds, and solving $t_2$ Laplacian linear systems in $\mL(H)$. Note that we repeat that for $t_1$ times, therefore, the algorithm $\ColumnApx(e, H, W)$ requires calling $O(\log^2{|V(H)|})$ times \Solve with accuracy $1/\mathrm{poly}(|V(H)|)$ on graphs that $\rho$-minor distribute into $\overline{G}$, and an additional $O((\rho\sqrt{\overline{n}}\log{\overline{n}}\log{|V(H)|}+D)\log{|V(H)|})$ rounds of communication in $\overline{G}$.
\end{proof}

\begin{lemma}
\label{lem:atmost_|T|}
The graph $H$ with a set of terminals $\T\subseteq V(H)$ satisfies that 
\begin{align}
\label{ineq:summation_re-1_be_LH-1_SCHT_LH_be}
\sum_{e\in E(H)}r_e^{-1}\vec{b}_e^T\mL(H)^{\dag}\left[\begin{array}{cc}
         0&0  \\
         0&\mSC(H,\T) 
    \end{array}\right]\mL(H)^{\dag}\vec{b}_e\le|\T|.
\end{align}
\end{lemma}
\begin{proof}
Since the LHS of (\ref{ineq:summation_re-1_be_LH-1_SCHT_LH_be}) is a scalar, it holds that 
\begin{equation}
\label{eq:summation_re-1_be_LH-1}
\begin{aligned}
&\sum_{e\in E(H)}r_e^{-1}\vec{b}_e^T\mL(H)^{\dag}\left[\begin{array}{cc}
         0&0  \\
         0&\mSC(H,\T) 
\end{array}\right]\mL(H)^{\dag}\vec{b}_e\\
=&\sum_{e\in E(H)}\Tr\left(r_e^{-1}\vec{b}_e^T\mL(H)^{\dag}\left[\begin{array}{cc}
         0&0  \\
         0&\mSC(H,\T) 
    \end{array}\right]\mL(H)^{\dag}\vec{b}_e\right).
\end{aligned} 
\end{equation}
By the properties of the trace operation, (\ref{eq:summation_re-1_be_LH-1}) becomes 
\begin{equation}
\label{eq:summation_einE(H)_trace}
\begin{aligned}
&\sum_{e\in E(H)}\Tr\left(\mL(H)^{\dag}\left[\begin{array}{cc}
         0&0  \\
         0&\mSC(H,\T) 
    \end{array}\right]\mL(H)^{\dag}r_e^{-1}\vec{b}_e\vec{b}_e^T\right)\\
=&\Tr\left(\mL(H)^{\dag}\left[\begin{array}{cc}
         0&0  \\
         0&\mSC(H,\T) 
    \end{array}\right]\mL(H)^{\dag}\sum_{e\in E(H)}r_e^{-1}\vec{b}_e\vec{b}_e^T\right).
\end{aligned}
\end{equation}
By the fact $\sum_{e\in E(H)}r_e^{-1}\vec{b}_e\vec{b}_e^T=\mL(H)$ and the properties of trace operation, (\ref{eq:summation_einE(H)_trace}) becomes 
\begin{equation}
\label{eq:trace_LH-1_SCHT_LH-1_LH}
\begin{aligned}
&\Tr\left(\mL(H)^{\dag}\left[\begin{array}{cc}
         0&0  \\
         0&\mSC(H,\T) 
    \end{array}\right]\mL(H)^{\dag}\mL(H)\right)=\Tr\left(\mL(H)^{\dag}\mL(H)\mL(H)^{\dag}\left[\begin{array}{cc}
         0&0  \\
         0&\mSC(H,\T) 
    \end{array}\right]\right)\\
    =&\Tr\left(\mL(H)^{\dag}\left[\begin{array}{cc}
         0&0  \\
         0&\mSC(H,\T) 
    \end{array}\right]\right)=\Tr\left(\left(\mL(H)^{\dag}\right)_{[\T, \T]}\mSC(H,\T)\right)\\
    =&\Tr\left(\left(\mL(H)^{\dag}\right)_{[\T, \T]}\mSC(H, \T)\mSC(H, \T)^{\dag}\mSC(H, \T)\right)\\
    =&\Tr\left(\mSC(H, \T)\left(\mL(H)^{\dag}\right)_{[\T, \T]}\mSC(H, \T)\mSC(H, \T)^{\dag}\right). 
\end{aligned}
\end{equation}
Lemma \ref{lemma:Inverse} gives $\mSC(H, \T)\left(\mL(H)^{\dag}\right)_{[\T, \T]}\mSC(H, \T)=\mSC(H, \T)$, so (\ref{eq:trace_LH-1_SCHT_LH-1_LH}) gets 
\[
\Tr\left(\mSC(H,\T)\mSC(H,\T)^{\dag}\right)=\Tr(\mP)=|\T|-1\le|\T|,
\]
where $\mP$ is the projection matrix of the space spanned by $\mSC(H, \T)$.

This completes the proof. 
\end{proof}
Now we review the \emph{flow localization} theorem. 
\begin{theorem}[Flow localization \cite{SRS18}]
\label{thm:flow_localization}
For a graph $H$, let 
\begin{align*}
    s_e=\sum_{f\in E(H)}\frac{|\vec{b}_e^T\mL(H)^{\dag}\vec{b}_f|}{\sqrt{r_{e}r_f}},
\end{align*}
then there exists an universal constant $C_{\textnormal{local}}$ such that $\sum_{e\in E(H)}s_e\le C_{\textnormal{local}}|E(H)|\log^2{|E(H)|}$. 
\end{theorem}
Now we are ready to prove Lemma \ref{lemma:FindSteady}.
\begin{proof}[Proof of Lemma \ref{lemma:FindSteady}]
We prove that $Z\subseteq E(H)$ is $(\alpha, \delta)$-steady according to the Definition \ref{def:alpha_delta_steady}.  
\begin{enumerate}
    \item (Quadratic form) 
    \begin{align*}
        \mathbb{E}\left[\sum_{e\in Z}r_e^{-1}\vec{b}_e\vec{b}_e^T\right]\preceq\mathbb{E}\left[\sum_{e\in Z_2}r_e^{-1}\vec{b}_e\vec{b}_e^T\right]=\alpha\sum_{e\in Z_1}r_e^{-1}\vec{b}_e\vec{b}_e^T\preceq\alpha\mL(H). 
    \end{align*}
    \item (Localization) By Lemma \ref{lem:ColumnApx}, we have that for each $e\in Z_2$, $s_e'\approx_{0.1}\sum_{f\neq e\in Z_2}\frac{|\vec{b}_e^T\mL(H)^{\dag}\vec{b}_f|}{\sqrt{r_{e}r_f}}$, which gives $\sum_{f\neq e\in Z_2}\frac{|\vec{b}_e^T\mL(H)^{\dag}\vec{b}_f|}{\sqrt{r_{e}r_f}}\le 1.1s_e'$. Combining with line \ref{line:localfinal} in algorithm \FindSteady, we know that for each edge $e\in Z$, 
\begin{align*}
\sum_{f\neq e\in Z}\frac{|\vec{b}_e^T\mL(H)^{\dag}\vec{b}_f|}{\sqrt{r_{e}r_f}}\le\delta.
\end{align*}
    \item (Variance) Lemma \ref{lem:DiffApx} gives that for each $e\in E(H)$,
\[
    v_e\approx_{0.1}r_e^{-1}\vec{b}_e^T\mL(H)^{\dag}\left[\begin{array}{cc}
         0&0  \\
         0&\mSC(H, \T) 
    \end{array}\right]\mL(H)^{\dag}\vec{b}_e,
    \]
which implies that $r_e^{-1}\vec{b}_e^T\mL(H)^{\dag}\left[\begin{array}{cc}
         0&0  \\
         0&\mSC(H, \T) 
    \end{array}\right]\mL(H)^{\dag}\vec{b}_e\le1.1v_e$. Combining with line \ref{line:Z1_ve_se} in algorithm \FindSteady, we have that for each edge $e\in Z$, \begin{align*}
r_e^{-1}\vec{b}_e^T\mL(H)^{\dag}\left[\begin{array}{cc}
             0&0  \\
             0&\mSC(H,\T) 
        \end{array}\right]\mL(H)^{\dag}\vec{b}_e\le\frac{18|\T|}{|E(H)|}.
    \end{align*}
\end{enumerate}

Now we bound $|Z|$. By Lemma \ref{lem:DiffApx} and Lemma \ref{lem:atmost_|T|}, we know that at most $1.1|E(H)|/16$ edges for $e\in E(H)$ satisfy that 
\begin{align*}
    v_e\ge 16|\T|/|E(H)|.
\end{align*}
Similarly, Lemma \ref{lem:ColumnApx} and Theorem \ref{thm:flow_localization} tell us that at most $1.1|E(H)|/16$ edges satisfy that 
\begin{align*}
    s_e\ge 16C_{\textnormal{local}}\log^2{|E(H)|}. 
\end{align*}
We conclude that 
\[
|Z_1|\ge(1-2.2/16)|E(H)|\ge 13|E(H)|/16.
\]
In addition, 
\[
\mathbb{E}[|Z_2|]=\alpha|Z_1|\ge13\alpha|E(H)|/16.
\]
By the definition of $Z_2$, we know that for each $e\in Z_2$, 
\begin{align*}
\mathbb{E}[s_e']\le 16\alpha C_{\textnormal{local}}\log^2{|E(H)|}.
\end{align*}
By Markov's inequality, we have that
\begin{align*}
\Pr[s_e'\ge\delta/1.1]\le\frac{\mathbb{E}[s_e']}{\delta/1.1}\le\frac{16\alpha C_{\textnormal{local}}\log^2{|E(H)|}}{\delta/1.1}=\frac{44}{115}. 
\end{align*}
Therefore, 
\begin{align*}
    \mathbb{E}[|Z|]\ge\left(1-\frac{44}{115}\right)\mathbb{E}[|Z_2|]\ge\alpha|E(H)|/2.
\end{align*}

The round complexity mainly comes from Lemma \ref{lem:DiffApx} and Lemma \ref{lem:ColumnApx}, and the remaining steps in algorithm \FindSteady can be trivially implemented. 
\end{proof}
Finally, we bound the number of while loops, which is a key ingredient to prove Theorem \ref{thm:approxsc_algebra} and Theorem \ref{thm:MinorSC}. 
\begin{lemma}
\label{lemma:runtime}
The while loop in algorithm {\ApproxSC} executes $O(\alpha^{-1}\log{m})$ times with 
\begin{align*}
    \alpha=\frac{\delta}{46\clocal\log^2{m}}=\frac{\epsilon}{46\clocal C\log^4{m}}
\end{align*}
with probability at least $1-1/\mathrm{poly}(m)$. 
\end{lemma}

\begin{proof}
To bound the number of iterations, it suffices to argue that
\[
\E[|E(G^{(i+1)})|] \le (1-\Omega(\alpha))|E(G^{(i)})|.
\]
Recall that in Lemma \ref{lemma:FindSteady}, we have $\E[|Z^{(i)}|]\ge\alpha|E(G^{(i)})|/2$, so it suffices to argue that each
original edge of $G^{(i)}$ is removed with at least a constant probability, even
considering the \Split~operation. We break the analysis into two cases.
\paragraph{Case 1: $e$ is split into two parallel edges $e_1$ and $e_2$.} Recall that by the definition of \Split, both $e_1$ and $e_2$ have leverage score in $[3/16,13/16].$ Therefore, if $e_1 \in Z^{(i)}$, then it will be contracted with probability at least $3/16.$ In that case both $e_1$ and $e_2$ disappear, as desired.

\paragraph{Case 2: $e$ is split into a path consisting of $e_1$ and $e_2$.} Recall that by the definition of \Split, both $e_1$ and $e_2$ have leverage score in $[3/16,13/16].$ Therefore, if $e_1 \in Z^{(i)}$, the it will be deleted with probability at least $3/16.$ Then $e_2$ becomes a leaf, so it will be removed during the \Unsplit~operation, as desired.
\end{proof}
\subsection{Matrix Martingale Analysis of Approximation}
\label{sec:matrix martingale analysis}
In this section, we prove Theorem \ref{thm:approxsc_algebra} by defining several stochastic sequences of matrices that capture the change of the quadratic form of the Schur complement. We also prove Theorem \ref{thm:MinorSC}. 

Let $\tau$ denote the final value of $i$ in algorithm \ApproxSC. Recall that Lemma \ref{lemma:runtime} gives $\tau=O(\frac{\log{m}}{\alpha})$. Let the hidden constant be $C'$, i.e., $\tau=\frac{C'\log{m}}{\alpha}$. For $0\le i\le\tau$ and $0\le t\le|Z^{(i)}|$, let $e_{i,t}$ be the $t$-th edge in $Z^{(i)}$ under an arbitrary ordering, and 
\begin{align}
{\hmY}^{(i,0)}=\mS_{0}\mL(H^{(i)})^{\dag}\mS_{0}^T,
\end{align}
with $\mS_0=[0, \mSC(G, \T)^{1/2}]$, then we have the following iteration equation 
\begin{equation}
\label{iter_Y_it+1_Y_it}
\begin{aligned}
{\hmY}^{(i, t+1)}=\left\{\begin{array}{ll}
{\hmY}^{(i, t)}+r_{e_{i,t}}^{-1}(1-p_{e_{i,t}})^{-1}\mS_{0}\mL(H^{(i)})^{\dag}\vec{b}_{e_{i,t}}\vec{b}_{e_{i,t}}^T\mL(H^{(i)})^{\dag}\mS_0^T &\textnormal{if $e_{i,t}$ is deleted,}  \\
{\hmY}^{(i, t)}-r_{e_{i,t}}^{-1}p_{e_{i,t}}^{-1}\mS_{0}\mL(H^{(i)})^{\dag}\vec{b}_{e_{i,t}}\vec{b}_{e_{i,t}}^T\mL(H^{(i)})^{\dag}\mS_0^T &\textnormal{if $e_{i,t}$ is contracted.} 
    \end{array}\right.
\end{aligned}
\end{equation}
Especially for ${\hmY}^{(0,0)}$, by Lemma \ref{lemma:Inverse} we have 
\begin{align*}
    {\hmY}^{(0,0)}&=\mS_0\mL(G)^{\dag}\mS_0^T=\mSC(G,\T)^{1/2}\left(\mL(G)^{\dag}\right)_{[\T,\T]}\mSC(G,\T)^{1/2}\\
    &=\mSC(G,\T)^{\dag/2}\mSC(G,\T)\mSC(G,\T)^{\dag/2}=\mP,
\end{align*}
where $\mP$ is the projection matrix of the space spanned by $\mSC(G,\T)$. 

In the proof of Theorem \ref{thm:approxsc_algebra}, if we assume that $\mSC(H^{(i)}, \T)\approx_{0.1}\mSC(G, \T)$ for all $(i,t)$, then we have the following claim to bound $\left\|\mS_0\mL(H^{(i)})^{\dag}\mS_0^T\right\|_2$. 
\begin{claim}
\label{claim:S0LH-1S0T_l2norm}
$\left\|\mS_0\mL(H^{(i)})^{\dag}\mS_0^T\right\|_2\le 1.1$. 
\end{claim}
\begin{proof}
The fact $\mSC(H^{(i)}, \T)\approx_{0.1}\mSC(G,\T)$ gives us $\mSC(H^{(i)}, \T)^{\dag}\approx_{0.1}\mSC(G,\T)^{\dag}$, i.e., 
\begin{align*}
\mSC(H^{(i)}, \T)^{\dag}\preceq 1.1\cdot\mSC(G,\T)^{\dag}. 
\end{align*}
By Lemma \ref{lemma:Inverse}, it holds that
\begin{align*}
\mP\left(\mL(H^{(i)})^{\dag}\right)_{[\T,\T]}\mP\preceq 1.1\cdot\mP\left(\mL(G)^{\dag}\right)_{[\T,\T]}\mP. 
\end{align*}
Moreover, we have 
\begin{align}
\label{ineq_LHiTT_1.1_LG_TT}
    \mSC(G, \T)^{1/2}\left(\mL(H^{(i)})^{\dag}\right)_{[\T,\T]}\mSC(G, \T)^{1/2}\preceq 1.1\cdot\mSC(G, \T)^{1/2}\left(\mL(G)^{\dag}\right)_{[\T,\T]}\mSC(G, \T)^{1/2}.
\end{align}
Lemma \ref{lemma:Inverse} has
\begin{align}
\label{eq:lem_inverse}
\mSC(G, \T)\left(\mL(G)^{\dag}\right)_{[\T, \T]}\mSC(G, \T)=\mSC(G, \T), 
\end{align}
which implies that 
\begin{align}
\label{eq:lem_inverse_SCGT_dag/2}
\mSC(G, \T)^{1/2}\left(\mL(G)^{\dag}\right)_{[\T,\T]}\mSC(G, \T)^{1/2}=\mSC(G, \T)^{\dag/2}\mSC(G, \T)\mSC(G, \T)^{\dag/2}
\end{align}
by multiplying the both sides of the LHS and RHS of (\ref{eq:lem_inverse}) by $\mSC(G, \T)^{\dag/2}$. Combing (\ref{ineq_LHiTT_1.1_LG_TT}) and (\ref{eq:lem_inverse_SCGT_dag/2}), we have 
\begin{align}
\label{ineq:SCGT1/2_preceq1.1}
    \mSC(G, \T)^{1/2}\left(\mL(H^{(i)})^{\dag}\right)_{[\T,\T]}\mSC(G, \T)^{1/2}\preceq 1.1\cdot\mSC(G, \T)^{\dag/2}\mSC(G, \T)\mSC(G, \T)^{\dag/2}.
\end{align}
Furthermore, 
\begin{align*}
\textnormal{LHS of (\ref{ineq:SCGT1/2_preceq1.1})}&=\mS_0\mL(H^{(i)})^{\dag}\mS_0^T,\\
\textnormal{RHS of (\ref{ineq:SCGT1/2_preceq1.1})}&=1.1\cdot\mP
\end{align*}
Therefore, $\mS_0\mL(H^{(i)})^{\dag}\mS_0^T\preceq 1.1\cdot\mP$ and $\left\|\mS_0\mL(H^{(i)})^{\dag}\mS_0^T\right\|_2\le 1.1$.
\end{proof}

Now we define the difference sequence for ${\hmY}^{(i,t)}$ by 
\begin{align*}
    \mX^{(i,t)}=\left\{\begin{array}{ll}
         0&\textnormal{if $t=0$,}  \\
         {\hmY}^{(i,t)}-{\hmY}^{(i, t-1)}&\textnormal{if $t>0$.} 
    \end{array}\right.
\end{align*}
The operator norm of $\mX^{(i,t)}$ has the following bound. 
\begin{lemma}
\label{lem_X_i_t+1_2norm}
For all $(i,t)$, it holds that 
\begin{align*}
    \left\|\mX^{(i,t)}\right\|_2\le\frac{162|\T|}{|E(H^{(i)})|}.
\end{align*}
\end{lemma}
\begin{proof}
By the definition of $\mX^{(i,t)}$ and equation (\ref{iter_Y_it+1_Y_it}), we have 
\begin{equation}
\label{ineq:Xit_l2norm_upper_bound}
\begin{aligned}
&\left\|\mX^{(i,t)}\right\|_2=\left\|{\hmY}^{(i,t)}-{\hmY}^{(i,t-1)}\right\|_2\\
\le&(r_{e_{i,t}}\cdot\min\{1-p_{e_{i,t}}, p_{e_{i,t}}\})^{-1}\left\|\mS_0\mL(H^{(i)})^{\dag}\vec{b}_{e_{i,t}}\vec{b}_{e_{i,t}}^T\mL(H^{(i)})^{\dag}\mS_0^T\right\|_2.
\end{aligned}
\end{equation}
By Lemma \ref{lem:split} and Lemma \ref{lem:LevApx}, we have that for each $e_{i,t}\in Z^{(i)}$, $\lev_{H^{(i)}}(e_{i,t})\in[3/16, 13/16]$ and $ p_{e_{i,t}}\approx_{\d}\lev_{H^{(i)}}(e_{i,t})$. Since $\d\le 0.01$, we have $p_{e_{i,t}}\in[1/8,7/8]$. Then (\ref{ineq:Xit_l2norm_upper_bound}) becomes  
\begin{align*}
\left\|\mX^{(i, t)}\right\|_{2}&\le8r_{e_{i,t}}^{-1}\vec{b}_{e_{i,t}}^T\mL(H^{(i)})^{\dag}\mS_0^T\mS_0\mL(H^{(i)})^{\dag}\vec{b}_{e_{i,t}}\\
    &\le 8r_{e_{i,t}}^{-1}\vec{b}_{e_{i,t}}^T\mL(H^{(i)})^{\dag}\left[\begin{array}{cc}
     0&0  \\
     0&\mSC(G, \T) 
\end{array}\right]\mL(H^{(i)})^{\dag}\vec{b}_{e_{i,t}}.
\end{align*}
The assumption $\mSC(H^{(i)}, \T)\approx_{0.1}\mSC(G, \T)$ implies that 
\[
\left[\begin{array}{cc}
     0&0  \\
     0&\mSC(G, \T) 
\end{array}\right]\preceq 1.1\left[\begin{array}{cc}
     0&0  \\
     0&\mSC(H^{(i)}, \T) 
\end{array}\right],
\]
which gives 
\[
\left\|\mX^{(i,t)}\right\|_2\le9r_{e_{i,t}}^{-1}\vec{b}_{e_{i,t}}^T\mL(H^{(i)})^{\dag}\left[\begin{array}{cc}
     0&0  \\
     0&\mSC(H^{(i)}, \T) 
\end{array}\right]\mL(H^{(i)})^{\dag}\vec{b}_{e_{i,t}}.
\]
Combining with condition \ref{condition:variance} in Definition \ref{def:alpha_delta_steady}, we have 
\[
\left\|\mX^{(i, t)}\right\|_2\le\frac{162|\T|}{|E(H^{(i)})|}.
\]
\end{proof}

\begin{claim}
\label{claim_Yhat_martingale}
For each fixed $0\le i\le\tau$, the sequence ${\hmY}^{(i, 0)}, \cdots, {\hmY}^{(i,|Z^{(i)}|)}$ is a martingale. 
\end{claim}
\begin{proof}
By the definition of $\mX^{(i,t)}$ and Lemma \ref{lem_X_i_t+1_2norm}, we have 
\begin{align*}
\left\|{\hmY}^{(i,t)}\right\|_2\le\left\|{\hmY}^{(i,t-1)}\right\|_2+\left\|\mX^{(i,t)}\right\|_2\le\left\|{\hmY}^{(i,t-1)}\right\|_2+\frac{162|\T|}{|E(H^{(i)})|},
\end{align*}
which implies that for a fixed $i$ and $0\le t\le |Z^{(i)}|$, 
\begin{align}
\label{ineq_Y_it_expectation_bounded}
    \mathbb{E}\left[\left\|{\hmY}^{(i,t)}\right\|_2\right]<\infty.
\end{align} 

Now considering the quantity $\mathbb{E}\left[{\hmY}^{(i, t+1)}\bigg|{\hmY}^{(i,0)},\cdots,{\hmY}^{(i,t)}\right]$, we have 
\begin{equation}
\label{eq_martingale_proof_Y_it}
\begin{aligned}
&\mathbb{E}\left[{\hmY}^{(i, t+1)}\bigg|{\hmY}^{(i,0)},\cdots,{\hmY}^{(i,t)}\right]=\mathbb{E}\left[{\hmY}^{(i, t+1)}\bigg|{\hmY}^{(i, t)}\right]\\
=&(1-p_{e_{i,t}})\cdot\left({\hmY}^{(i, t)}+r_{e_{i,t}}^{-1}(1-p_{e_{i,t}})^{-1}\mS_{0}\mL(H^{(i)})^{\dag}\vec{b}_{e_{i,t}}\vec{b}_{e_{i,t}}^T\mL(H^{(i)})^{\dag}\mS_0^T\right)\\
+&p_{e_{i,t}}\cdot\left({\hmY}^{(i, t)}-r_{e_{i,t}}^{-1}p_{e_{i,t}}^{-1}\mS_{0}\mL(H^{(i)})^{\dag}\vec{b}_{e_{i,t}}\vec{b}_{e_{i,t}}^T\mL(H^{(i)})^{\dag}\mS_0^T\right)\\
=&{\hmY}^{(i,t)}. 
\end{aligned}    
\end{equation}

Putting (\ref{ineq_Y_it_expectation_bounded}) and (\ref{eq_martingale_proof_Y_it}) together proves this claim. 
\end{proof}

However, the \Unsplit and \Split operations lead to ${\hmY}^{(i+1, 0)}\neq{\hmY}^{(i, |Z^{(i)}|)}$. In order to treat the $\tau$ sequences as a whole, we define the new sequence $\mY^{(i,t)}$ such that 
\begin{align*}
    \left\{\begin{array}{l}
        \mY^{(0, 0)}={\hmY}^{(0,0)}=\mP, \\
        \mY^{(i,t)}-\mY^{(i, t-1)}=\mX^{(i,t)},\\
        \mY^{(i,0)}=\mY^{(i-1, |Z^{(i)}|)}.
    \end{array}\right.
\end{align*}
We prove Theorem \ref{thm:approxsc_algebra} by considering two parts: the martingale $\mY^{(i,t)}$ and the errors resulting from ${\hmY}^{(i+1, 0)}-{\hmY}^{(i,|Z^{(i)}|)}$. Before that, we prove that the sequence $\mY^{(i,t)}$ is also a martingale following the proof of Claim \ref{claim_Yhat_martingale}.
\begin{itemize}
\item Note that $\mathbb{E}\left[\left\|\mY^{(i,t)}\right\|_2\right]<\infty$ since \begin{align*}
\left\|\mY^{(i,t)}\right\|_2\le\left\|\mY^{(i, t-1)}\right\|_2+\left\|\mX^{(i,t)}\right\|_2\le\left\|\mY^{(i,t-1)}\right\|_{2}+\frac{162|\T|}{|E(H^{(i)})|}.
\end{align*}
\item For the special case, $\mathbb{E}\left[\mY^{(i+1,0)}\bigg|\mY^{(i,0)},\cdots\mY^{(i,|Z^{(i)}|)}\right]=\mathbb{E}\left[\mY^{(i,Z^{(i)})}\right]$. More generally, 
\begin{align*}
&\mathbb{E}\left[\mY^{(i,t+1)}\bigg|\mY^{(i,0)},\cdots,\mY^{(i,t)}\right]=\mathbb{E}\left[\mY^{(i,t+1)}\bigg|\mY^{(i,t)}\right]\\
=&\mY^{(i,t)}+\mathbb{E}\left[\mX^{(i, t+1)}\right]=\mY^{(i,t)}+\mathbb{E}\left[{\hmY}^{(i,t+1)}-{\hmY}^{(i,t)}\right]=\mY^{(i,t)}. 
\end{align*}
\end{itemize}
This completes the proof. 

\begin{lemma}
\label{lem_W_it_2norm}
Let $C$ be a sufficiently large constant in algorithm \ApproxSC. Then for all $(i,t)$, it holds that
\begin{align*}
\left\|\mW^{(i,t)}\right\|_2\le\frac{\epsilon^2}{100\log{m}}.
\end{align*}
\end{lemma}
\begin{proof}
Recall that in Lemma \ref{lemma:freedman}
\begin{align*}
    \mW^{(k)}=\sum_{j=1}^{k}\mathbb{E}\left[\left(\mX^{(j)}\right)^2\bigg|\mX^{(j-1)}\right].
\end{align*}
Here we have 
\begin{align}
\label{eq:W_i+10_W_i0}
\mW^{(i+1, 0)}-\mW^{(i, 0)}=\sum_{1\le t\le|Z^{(i)}|}\mathbb{E}\left[\left(\mX^{(i,t)}\right)^2\bigg|\mX^{(i,t-1)}\right]=\sum_{1\le t\le|Z^{(i)}|}\mathbb{E}\left[\left(\mX^{(i,t)}\right)^2\right],
\end{align}
and set $\mW^{(0,0)}=0$. By the fact $\left(\mX^{(i,t)}\right)^2\preceq\left\|\mX^{(i,t)}\right\|_2\cdot\mX^{(i,t)}$, we have 
\begin{align}
\label{ineq:summation_E_Xit_square}
\sum_{1\le t\le|Z^{(i)}|}\mathbb{E}\left[\left(\mX^{(i,t)}\right)^2\right]\preceq\sum_{1\le t\le|Z^{(i)}|}\left\|\mX^{(i,t)}\right\|_2\mathbb{E}\left[\mX^{(i,t)}\right].
\end{align}
By Lemma \ref{lem_X_i_t+1_2norm}, it holds that
\begin{align}
\label{ineq:summation_Xit_l2_E_Xit}
\sum_{1\le t\le|Z^{(i)}|}\left\|\mX^{(i,t)}\right\|_2\mathbb{E}\left[\mX^{(i,t)}\right]\preceq\frac{162|\T|}{|E(H^{(i)})|}\sum_{1\le t\le|Z^{(i)}|}\mathbb{E}\left[\mX^{(i,t)}\right].
\end{align}
By the definition of $\mX^{(i,t)}$ and the fact $p_{e_{i,t}}\in[1/8, 7/8]$, we have 
\begin{equation}
\begin{aligned}
\sum_{1\le t\le|Z^{(i)}|}\mathbb{E}\left[\mX^{(i,t)}\right]&\preceq 8\sum_{1\le t\le|Z^{(i)}|}\mathbb{E}\left[r_{e_{i,t}}^{-1}\mS_0\mL(H^{(i)})^{\dag}\vec{b}_{e_{i,t}}\vec{b}_{e_{i,t}}^T\mL(H^{(i)})^{\dag}\mS_0^T\right]\\
&=8\cdot\mS_0\mL(H^{(i)})^{\dag}\mathbb{E}\left[\sum_{1\le t\le|Z^{(i)}|}r_{e_{i,t}}^{-1}\vec{b}_{e_{i,t}}\vec{b}_{e_{i,t}}^T\right]\mL(H^{(i)})^{\dag}\mS_0^T.
\end{aligned}
\end{equation}
By the condition \ref{condition:quadratic_form} in Definition \ref{def:alpha_delta_steady}, (\ref{ineq:summation_E_Xit}) becomes 
\begin{align}
\label{ineq:summation_E_Xit}
    \sum_{1\le t\le |Z^{(i)}|}\mathbb{E}\left[\mX^{(i,t)}\right]\preceq 8\alpha\cdot\mS_0\mL(H^{(i)})^{\dag}\mL(H^{(i)})\mL(H^{(i)})^{\dag}\mS_0^T=8\alpha\cdot\mS_0\mL(H^{(i)})^{\dag}\mS_0^T. 
\end{align}
Combining (\ref{ineq:summation_E_Xit}),  (\ref{ineq:summation_Xit_l2_E_Xit}), (\ref{ineq:summation_E_Xit_square}) with (\ref{eq:W_i+10_W_i0}), we have 
\[
\mW^{(i+1, 0)}-\mW^{(i, 0)}\preceq\frac{1296|\T|\alpha}{|E(H^{(i)})|}\mS_0\mL(H^{(i)})^{\dag}\mS_0^T.
\]
and 
\begin{align*}
    \left\|\mW^{(i+1, 0)}-\mW^{(i,0)}\right\|_2\le\frac{1296|\T|\alpha}{|E(H^{(i)})|}\left\|\mS_0\mL(H^{(i)})^{\dag}\mS_0^T\right\|_2\le\frac{1426|\T|\alpha}{|E(H^{(i)})|},
\end{align*}
where the last inequality follows from Claim \ref{claim:S0LH-1S0T_l2norm}.

For any $i,t$, we have
\begin{align*}
    \left\|\mW^{(i,t)}\right\|_2&=\left\|\mW^{(i,t)}-\mW^{(0,0)}\right\|_2\le\sum_{0\le k\le i}\left\|\mW^{(k+1, 0)}-\mW^{(k, 0)}\right\|_2\le\tau\cdot\frac{1426|\T|\alpha}{|E(H^{(i)})|}\\
    &=\frac{C'\log{m}}{\alpha}\cdot\frac{1426|\T|\alpha}{|E(H^{(i)})|}\le\frac{C'\log{m}}{\alpha}\cdot\frac{1426|\T|\alpha\epsilon^2}{C|\T|\log^2{m}}\le\frac{\epsilon^2}{100\log{m}}.
\end{align*}
\end{proof}
Now combining Lemma \ref{lem_X_i_t+1_2norm} and Lemma \ref{lem_W_it_2norm}, we have the following lemma. 
\begin{lemma}
\label{lem_bound_Y_it}
With probability at least $1-1/\mathrm{poly}(m)$, for all  $(i,t)$, it holds that $\|\mY^{(i,t)}-\mP\|_2\le\epsilon/2$. 
\end{lemma}
\begin{proof}
Lemma \ref{lem_X_i_t+1_2norm} gives  $\left\|\mX^{(i,t)}\right\|_2\le\frac{162|\T|}{|E(H^{(i)})|}$. The algorithm {\ApproxSC} implies that $|E(H^{(i)})|=\Omega\left(\frac{|\T|\log^2{m}}{\epsilon^2}\right)$. Then we have $\left\|\mX^{(i,t)}\right\|_2=O\left(\frac{\epsilon^2}{\log^2{m}}\right)$. Setting $R=\frac{\epsilon^2}{100\log^2{m}}$ and $\sigma^2=\frac{\epsilon^2}{100\log{m}}$, and applying Lemma \ref{lemma:freedman}, we obtain 
\begin{align*}
    &\Pr\left[\exists(i,t)\bigg|\left\|\mY^{(i,t)}-\mY^{(0,0)}\right\|_2\ge\epsilon/2,\ \left\|\mW^{(i,t)}\right\|_2\le\sigma^2\right]\\
    \le&2|\T|\cdot\exp\left(\frac{-\epsilon^2/12}{\sigma^2+R\epsilon/6}\right)\le 2|\T|\cdot\exp\left(\frac{-\epsilon^2/12}{\frac{\epsilon^2}{100\log{m}}+\frac{\epsilon^3}{600\log^2{m}}}\right)=1/\mathrm{poly}(m).
\end{align*}
Moreover, since $\left\|\mW^{(i,t)}\right\|_2\le\sigma^2$, we have 
\begin{align*}
\Pr\left[\left\|\mY^{(i,t)}-\mY^{(0,0)}\right\|_2\ge\epsilon/2\right]\le 1/\mathrm{poly}(m),
\end{align*}
which gives that 
\begin{align*}
    \Pr\left[\left\|\mY^{(i,t)}-\mP\right\|_2\le\epsilon/2\right]\ge 1-1/\mathrm{poly}(m).
\end{align*}
\end{proof}

Since \Split and \Unsplit preserve the Schur complement,we conclude that ${\hmY}^{(i+1, 0)}$ satisfies
\begin{align*}
    {\hmY}^{(i+1, 0)}&=\mS_0\mL(H^{(i+1)})^{\dag}\mS_0^T\\
    &=\mSC(G,\T)^{1/2}\left(\mL(H^{(i+1)})^{\dag}\right)_{[\T,\T]}\mSC(G,\T)^{1/2}\\
    &=\mSC(G,\T)^{1/2}\left(\mL(I^{(i)})^{\dag}\right)_{[\T,\T]}\mSC(G,\T)^{1/2}\\
    &=\mS_0\mL(I^{(i)})^{\dag}\mS_0^T,
\end{align*}
where $I^{(i)}$ corresponds to $H^{(i)}$ in algorithm \ApproxSC. The Laplacian matrices $\mL(I^{(i)})$ and $\mL(H^{(i)})$ satisfy the following relation
\begin{align*}
    \mL(I^{(i)})=\mL(H^{(i)})+\mU\mC^{(i)}\mU^T,
\end{align*}
where $\mC^{(i)}$ is a diagonal matrix such that 
\begin{align*}
    \mC_{ff}^{(i)}=\left\{\begin{array}{ll}
         -1&\textnormal{if $f\in Z^{(i)}$ is deleted}  \\
         \infty&\textnormal{if $f\in Z^{(i)}$ is contracted} 
    \end{array}\right.
\end{align*}
and $\mU=\mB_{Z^{(i)}}^{T}\mR_{Z^{(i)}}^{-1/2}$, where $\mB_{Z^{(i)}}$ and $\mR_{Z^{(i)}}$ are the matrices with restriction of $\mB$ and $\mR$ to the indices corresponding to the set $Z^{(i)}$. Then by Woodbury matrix formula (see Lemma \ref{lemma:woodbury}) we have 
\begin{align}
\label{eq_woodbury_LI_LH}
    \mL(I^{(i)})^{\dag}=\mL(H^{(i)})^{\dag}-\mL(H^{(i)})^{\dag}\mU\left((\mC^{(i)})^{-1}+\mU^T\mL(H^{(i)})^{\dag}\mU\right)^{-1}\mU^T\mL(H^{(i)})^{\dag}.
\end{align}
Furthermore, multiplying the left side and right side of (\ref{eq_woodbury_LI_LH}) by $\mS_0$ and $\mS_0^T$ respectively, one can obtain 
\begin{align}
\label{eq_Yhat_i+1_Yhat_i}
    {\hmY}^{(i+1, 0)}-{\hmY}^{(i, 0)}=-\mS_0\mL(H^{(i)})^{\dag}\mU\left((\mC^{(i)})^{-1}+\mU^T\mL(H^{(i)})^{\dag}\mU\right)^{-1}\mU^T\mL(H^{(i)})^{\dag}\mS_0^T. 
\end{align}

For the quantity ${\hmY}^{(i, |Z^{(i)}|)}-{\hmY}^{(i, 0)}$, by virtue of the equation (\ref{iter_Y_it+1_Y_it}), we have 
\begin{equation}
\label{eq_Yhat_iZi_Yhat_i0}
\begin{aligned}
{\hmY}^{(i, |Z^{(i)}|)}-{\hmY}^{(i,0)}&=\sum_{0\le t<|Z^{(i)}|}{\hmY}^{(i, t+1)}-{\hmY}^{(i, t)}\\
&=\mS_0\mL(H^{(i)})^{\dag}\mU\mP^{(i)}\mU^{T}\mL(H^{(i)})^{\dag}\mS_0^T,
\end{aligned}
\end{equation}
where $\mP^{(i)}$ is a diagonal matrix such that

\begin{align*}
\mP^{(i)}_{ff}=\left\{\begin{array}{ll}
         (1-p_f)^{-1}&\textnormal{if $f$ is deleted,}  \\
         -p_f^{-1}&\textnormal{if $f$ is contracted.} 
    \end{array}\right.
\end{align*}

Subtracting (\ref{eq_Yhat_iZi_Yhat_i0}) from (\ref{eq_Yhat_i+1_Yhat_i}) gives 
\begin{align*}
    &{\hmY}^{(i+1, 0)}-{\hmY}^{(i, |Z^{(i)}|)}\\
    =&-\mS_0\mL(H^{(i)})^{\dag}\mU\left[\left((\mC^{(i)})^{-1}+\mU^T\mL(H^{(i)})^{\dag}\mU\right)^{-1}+\mP^{(i)}\right]\mU^T\mL(H^{(i)})^{\dag}\mS_0^T.
\end{align*}
Define ${\hmX}^{(i)}={\hmY}^{(i+1,0)}-{\hmY}^{(i, |Z^{(i)}|)}$. Before bounding $\left\|{\hmX}^{(i)}\right\|_2$, we have the following lemma. 
\begin{lemma}
\label{lem_DQP}
$\left\|\left((\mC^{(i)})^{-1}+\mU^T\mL(H^{(i)})^{\dag}\mU\right)^{-1}+\mP^{(i)}\right\|_2\le 36\delta$. 
\end{lemma}
\begin{proof}
Let $\mD^{(i)}$ be the diagonal matrix with entries being the diagonal entries of the matrix $(\mC^{(i)})^{-1}+\mU^T\mL(H^{(i)})^{\dag}\mU$, specifically, 
\begin{align*}
    \mD^{(i)}_{ff}=\left\{\begin{array}{ll}
         -1+\mathrm{lev}_{H^{(i)}}(f)&\textnormal{if $f$ is deleted,}  \\
         \mathrm{lev}_{H^{(i)}}(f)&\textnormal{if $f$ is contracted.} 
    \end{array}\right.
\end{align*}
Define another matrix $\mQ^{(i)}$ by 
\begin{align*}
    \mQ^{(i)}=\mD^{(i)}-\left((\mC^{(i)})^{-1}+\mU^T\mL(H^{(i)})^{\dag}\mU\right).
\end{align*}
Note that all the diagonal entries of $\mQ^{(i)}$ are $0$. Considering the summation of non-diagonal entries of $\mQ^{(i)}$, which is equal to the summation of non-diagonal entries of $\mU^T\mL(H^{(i)})^{\dag}\mU$, we have 
\begin{align*}
    \sum_{f\neq g}\bigg|\mQ^{(i)}_{fg}\bigg|=\sum_{f\neq g\in Z^{(i)}}\frac{|\vec{b}_e^T\mL(H^{(i)})^{\dag}\vec{b}_f|}{\sqrt{r_{f}r_{g}}}\le\delta,
\end{align*}
where the inequality follows from the localization condition of Definition \ref{def:alpha_delta_steady}. Lemma \ref{lemma:sdd} gives
\[
\left\|\mQ^{(i)}\right\|_2\le\sum_{f\neq g}\bigg|\mQ^{(i)}_{fg}\bigg|, 
\]
and thus $\left\|\mQ^{(i)}\right\|_2\le\delta$.  

Consider the left side of the target inequality,
\begin{equation}
\label{ineq_DQP}
\begin{aligned}
    &\left\|\left((\mC^{(i)})^{-1}+\mU^T\mL(H^{(i)})^{\dag}\mU\right)^{-1}+\mP^{(i)}\right\|_2\\
    =&\left\|\left(\mD^{(i)}-\mQ^{(i)}\right)^{-1}+\mP^{(i)}\right\|_2\le\left\|\left(\mD^{(i)}-\mQ^{(i)}\right)^{-1}-(\mD^{(i)})^{-1}\right\|_2+\left\|(\mD^{(i)})^{-1}+\mP^{(i)}\right\|_2,
\end{aligned}
\end{equation}
in which 
\begin{equation}
\label{ineq_DQP_part1}
\begin{aligned}
    &\left\|\left(\mD^{(i)}-\mQ^{(i)}\right)^{-1}-(\mD^{(i)})^{-1}\right\|_2=\left\|\left(\mD^{(i)}-\mQ^{(i)}\right)^{-1}\mQ^{(i)}(\mD^{(i)})^{-1}\right\|_2\\
    \le&\left\|\left(\mD^{(i)}-\mQ^{(i)}\right)^{-1}\right\|_2\left\|\mQ^{(i)}\right\|_2\left\|(\mD^{(i)})^{-1}\right\|_2\le\left(\frac{3}{16}-\delta\right)^{-1}\cdot\d\cdot\frac{16}{3}\le 30\delta,
\end{aligned}    
\end{equation}
and 
\begin{align}
\label{ineq_DQP_part2}
    \left\|(\mD^{(i)})^{-1}+\mP^{(i)}\right\|_2=\max_{f\in Z^{(i)}}\left\{\bigg|\frac{1}{-1+\mathrm{lev}_{H^{(i)}}(f)}+\frac{1}{1-p_f}\bigg|,\  \bigg|\frac{1}{\mathrm{lev}_{H^{(i)}}(f)}-\frac{1}{p_f}\bigg|\right\}\le 6\delta,
\end{align}
where $p_f\approx_{\delta}\mathrm{lev}_{H^{(i)}}(f)$. 

Substituting (\ref{ineq_DQP_part1}) and (\ref{ineq_DQP_part2}) to (\ref{ineq_DQP}) completes the proof.
\end{proof}
Based on Lemma \ref{lem_DQP}, we can give the following bounds on ${\hmX}^{(i)}$. 
\begin{lemma}
\label{lem_operator_bound_Xhat_EXhat}
For all $(i,t)$, it holds that
\begin{align*}
    \left\|{\hmX}^{(i)}\right\|_2&\le 40\delta,\\
    \left\|\mathbb{E}_{Z^{(i)}}\left[{\hmX}^{(i)}\right]\right\|_2&\le40\alpha\delta,\\
    \left\|\mathbb{E}_{Z^{(i)}}\left[\left({\hmX}^{(i)}\right)^2\right]\right\|_2&\le1600\alpha\delta^2.
\end{align*}
\end{lemma}
\begin{proof}
Recall that 
\begin{align*}
    {\hmX}^{(i)}&={\hmY}^{(i+1, 0)}-{\hmY}^{(i, |Z^{(i)}|)}\\
    &=-\mS_0\mL(H^{(i)})^{\dag}\mU\left[\left((\mC^{(i)})^{-1}+\mU^T\mL(H^{(i)})^{\dag}\mU\right)^{-1}+\mP^{(i)}\right]\mU^T\mL(H^{(i)})^{\dag}\mS_0^T, 
\end{align*}
then
\begin{align*}
    \left\|{\hmX}^{(i)}\right\|_2&=\left\|\mS_0\mL(H^{(i)})^{\dag}\mU\left[\left((\mC^{(i)})^{-1}+\mU^T\mL(H^{(i)})^{\dag}\mU\right)^{-1}+\mP^{(i)}\right]\mU^T\mL(H^{(i)})^{\dag}\mS_0^T\right\|_2\\
    &\le\left\|\left((\mC^{(i)})^{-1}+\mU^T\mL(H^{(i)})^{\dag}\mU\right)^{-1}+\mP^{(i)}\right\|_2\left\|\mS_0\mL(H^{(i)})^{\dag}\mU\mU^T\mL(H^{(i)})^{\dag}\mS_0^T\right\|_2\\
    &\le 36\delta\cdot\left\|\mS_0\mL(H^{(i)})^{\dag}\mL(H^{(i)})\mL(H^{(i)})^{\dag}\mS_0^T\right\|_2\\
    &=36\delta\cdot\left\|\mS_0\mL(H^{(i)})\mS_0^T\right\|_2\\
    &\le 40\delta, \tag{By Claim \ref{claim:S0LH-1S0T_l2norm}}
\end{align*}
where the second inequality follows from Lemma \ref{lem_DQP} and the fact $\mU\mU^T\preceq\mL(H^{(i)})$. 

For $\mathbb{E}_{Z^{(i)}}\left[{\hmX}^{(i)}\right]$, we have 
\begin{align*}
    &\left\|\mathbb{E}_{Z^{(i)}}\left[{\hmX}^{(i)}\right]\right\|_2\\
    =&\left\|\mathbb{E}_{Z^{(i)}}\left[\mS_0\mL(H^{(i)})^{\dag}\mU\left[\left((\mC^{(i)})^{-1}+\mU^T\mL(H^{(i)})^{\dag}\mU\right)^{-1}+\mP^{(i)}\right]\mU^T\mL(H^{(i)})^{\dag}\mS_0^T\right]\right\|_2. 
\end{align*}
Since $\left((\mC^{(i)})^{-1}+\mU^T\mL(H^{(i)})^{\dag}\mU\right)^{-1}+\mP^{(i)}\preceq\left\|\left((\mC^{(i)})^{-1}+\mU^T\mL(H^{(i)})^{\dag}\mU\right)^{-1}+\mP^{(i)}\right\|_2\cdot\mI$, we have 
\begin{align*}
&\mathbb{E}_{Z^{(i)}}\left[\mS_0\mL(H^{(i)})^{\dag}\mU\left[\left((\mC^{(i)})^{-1}+\mU^T\mL(H^{(i)})^{\dag}\mU\right)^{-1}+\mP^{(i)}\right]\mU^T\mL(H^{(i)})^{\dag}\mS_0^T\right]\\
\preceq&\left\|\left((\mC^{(i)})^{-1}+\mU^T\mL(H^{(i)})^{\dag}\mU\right)^{-1}+\mP^{(i)}\right\|_2\mathbb{E}_{Z^{(i)}}\left[\mS_0\mL(H^{(i)})^{\dag}\mU\mU^T\mL(H^{(i)})^{\dag}\mS_0^T\right] 
\end{align*}
and 
\begin{equation}
\label{ineq:E_Zi_Xhati_l2norm}
\begin{aligned}
&\left\|\mathbb{E}_{Z^{(i)}}\left[{\hmX}^{(i)}\right]\right\|_2\\
\le&\left\|\left((\mC^{(i)})^{-1}+\mU^T\mL(H^{(i)})^{\dag}\mU\right)^{-1}+\mP^{(i)}\right\|_2\left\|\mathbb{E}_{Z^{(i)}}\left[\mS_0\mL(H^{(i)})^{\dag}\mU\mU^T\mL(H^{(i)})^{\dag}\mS_0^T\right]\right\|_2\\
=&\left\|\left((\mC^{(i)})^{-1}+\mU^T\mL(H^{(i)})^{\dag}\mU\right)^{-1}+\mP^{(i)}\right\|_2\left\|\mS_0\mL(H^{(i)})^{\dag}\mathbb{E}_{Z^{(i)}}\left[\mU\mU^T\right]\mL(H^{(i)})^{\dag}\mS_0^T\right\|_2.
\end{aligned}
\end{equation}
By Lemma \ref{lem_DQP} and $\mU=\mB_{Z^{(i)}}^{T}\mR_{Z^{(i)}}^{-1/2}$, (\ref{ineq:E_Zi_Xhati_l2norm}) becomes 
\begin{align*}
\left\|\mathbb{E}_{Z^{(i)}}\left[{\hmX}^{(i)}\right]\right\|_2\le&36\delta\left\|\mS_0\mL(H^{(i)})^{\dag}\mathbb{E}_{Z^{(i)}}\left[\sum_{f\in Z^{(i)}}r_f^{-1}\vec{b}_f\vec{b}_f^T\right]\mL(H^{(i)})^{\dag}\mS_0^T\right\|_2\\
    \le&36\alpha\delta\left\|\mS_0\mL(H^{(i)})\mS_0^T\right\|_2\\
    \le&40\alpha\delta,
\end{align*}
where the second inequality follows from the condition \ref{condition:quadratic_form} of Definition \ref{def:alpha_delta_steady} and the last inequality follows from Claim \ref{claim:S0LH-1S0T_l2norm}. 

For $\mathbb{E}_{Z^{(i)}}\left[\left({\hmX}^{(i)}\right)^2\right]$, using the fact $\left(\hmX^{(i)}\right)^2\preceq\left\|\hmX^{(i)}\right\|_2\cdot\hmX^{(i)}$, we have 
\begin{align*}
    &\left\|\mathbb{E}_{Z^{(i)}}\left[\left({\hmX}^{(i)}\right)^2\right]\right\|_2\le\left\|{\hmX}^{(i)}\right\|_2\left\|\mathbb{E}_{Z^{(i)}}\left[{\hmX}^{(i)}\right]\right\|_2\le 40\delta\cdot 40\alpha\delta=1600\alpha\delta^2.
\end{align*}
\end{proof}

Now we bound $\left\|{\hmY}^{(i,0)}-\mY^{(i,0)}\right\|_2$. 
\begin{lemma}
\label{lem_widehat_Y_i0}
With probability at least $1-1/\mathrm{poly}(m)$, $\left\|{\hmY}^{(i,0)}-\mY^{(i,0)}\right\|_2\le\epsilon/2$.
\end{lemma}
\begin{proof}
We first consider the difference matrix ${\hmY}^{(i,0)}-{\hmY}^{(0,0)}$, which can be decomposed into two parts: the summation of $\hmY^{(j+1, 0)}-\hmY^{(j, |Z^{(j)}|)}$ and the summation of $\sum_{k=1}^{|Z^{(j)}|}\left(\hmY^{(j,k)}-\hmY^{(j, k-1)}\right)$ for $j=0,\cdots,i-1$,
\[
{\hmY}^{(i,0)}-{\hmY}^{(0,0)}=\sum_{j=0}^{i-1}\left(\hmY^{(j+1, 0)}-\hmY^{(j, |Z^{(j)}|)}\right)+\sum_{j=0}^{i-1}\sum_{k=1}^{|Z^{(j)}|}\left({\hmY}^{(j,k)}-{\hmY}^{(j, k-1)}\right).
\]
Recall that $\hmX^{(j)}=\hmY^{(j+1, 0)}-\hmY^{(j, |Z^{(j)}|)}$ and $\hmY^{(j,k)}-\hmY^{(j, k-1)}=\mY^{(j, k)}-\mY^{(j, k-1)}$, then
\begin{align}
\label{eq:Yhat_i0_minus_Yhat_00}
\hmY^{(i,0)}-\hmY^{(0, 0)}=\sum_{j=0}^{i-1}{\hmX}^{(j)}+\sum_{j=0}^{i-1}\sum_{k=1}^{|Z^{(j)}|}\left(\mY^{(j,k)}-\mY^{(j, k-1)}\right).
\end{align}
Recall that $\mY^{(j, 0)}=\mY^{(j-1, |Z^{(j-1)}|)}$, so we have 
\begin{equation}
\label{eq:summation_Yjk_minus_Yjk-1}
\begin{aligned}
\sum_{j=0}^{i-1}\sum_{k=1}^{|Z^{(j)}|}\left(\mY^{(j,k)}-\mY^{(j, k-1)}\right)&=\sum_{j=0}^{i-1}\sum_{k=1}^{|Z^{(j)}|}\left(\mY^{(j,k)}-\mY^{(j, k-1)}\right)+\sum_{j=1}^{i}\left(\mY^{(j,0)}-\mY^{(j-1,|Z^{(j-1)}|)}\right)\\
&=\mY^{(i,0)}-\mY^{(0,0)}.
\end{aligned}
\end{equation}
Substituting (\ref{eq:summation_Yjk_minus_Yjk-1}) to (\ref{eq:Yhat_i0_minus_Yhat_00}) gives 
\[
\hmY^{(i,0)}-\hmY^{(0,0)}=\sum_{j=0}^{i-1}{\hmX}^{(j)}+\mY^{(i,0)}-\mY^{(0,0)}.
\]
Recall that $\hmY^{(0,0)}=\mY^{(0,0)}$, then we can obtain 
\begin{align}
\label{eq:Yhat_minus_Yi0_simplify}
{\hmY}^{(i,0)}-\mY^{(i,0)}=\sum_{j=0}^{i-1}{\hmX}^{(j)}.
\end{align}

Define the new sequence $\mU^{(i)}$ such that
\begin{align}
\label{eq:def_Ui}
\mU^{(i)}={\hmX}^{(i)}-\mathbb{E}_{Z^{(i)}}\left[{\hmX}^{(i)}\right],    
\end{align}
and $\{\mV^{(i)}\}$ to be the martingale with difference sequence $\mU^{(i)}$ and $\mV^{(0)}=0$. In order to apply Lemma \ref{lemma:freedman} to martingale $\{\mV^{(i)}\}$, we first give the bounds of $\left\|\mU^{(i)}\right\|_2$ and $\left\|\sum_{i}\mathbb{E}_{Z^{(i)}}\left[(\mU^{(i)})^2\bigg|\mU^{(i-1)}\right]\right\|_2$. By the definition of $\mU^{(i)}$ and Lemma \ref{lem_operator_bound_Xhat_EXhat}, we have 
\begin{align*}
    \left\|\mU^{(i)}\right\|_2&=\left\|{\hmX}^{(i)}-\mathbb{E}_{Z^{(i)}}\left[{\hmX}^{(i)}\right]\right\|_2\le\left\|{\hmX}^{(i)}\right\|_2+\left\|\mathbb{E}_{Z^{(i)}}\left[{\hmX}^{(i)}\right]\right\|_2\\
    &\le 40\d+40\alpha\d\le 80\delta.
\end{align*}
In addition, 
\begin{align*}
    \left\|\mathbb{E}_{Z^{(i)}}\left[(\mU^{(i)})^2\bigg|\mU^{(i-1)}\right]\right\|_2
    &=\left\|\mathbb{E}_{Z^{(i)}}\left[(\mU^{(i)})^2\right]\right\|_2=\left\|\mathbb{V}\left[\hmX^{(i)}\right]\right\|_2\\
    &\le\left\|\mathbb{E}_{Z^{(i)}}\left[\left({\hmX}^{(i)}\right)^2\right]\right\|_2\le 1600\alpha\delta^2,  
\end{align*}
where the first inequality follows from the fact $\mathbb{V}\left[\hmX^{(i)}\right]\preceq\mathbb{E}_{Z^{(i)}}\left[\left(\hmX^{(i)}\right)^2\right]$. Moreover, by triangle inequality, we have 
\begin{align*}
\left\|\sum_ {i}\mathbb{E}_{Z^{(i)}}\left[(\mU^{(i)})^2\bigg|\mU^{(i-1)}\right]\right\|_2&\le\sum_{i}\left\|\mathbb{E}_{Z^{(i)}}\left[(\mU^{(i)})^2\bigg|\mU^{(i-1)}\right]\right\|_2\le\tau\cdot 1600\alpha\delta^2\\
&\le\frac{C'\log{m}}{\alpha}\cdot 1600\alpha\delta^2=1600C'\delta^2\log{m}. 
\end{align*}
Setting $R=80\delta$ and $\sigma^2=1600C'\delta^2\log{m}$, Lemma \ref{lemma:freedman} gives that 
\begin{align*}
&\Pr\left[\exists i\bigg|\left\|\mV^{(i)}\right\|_2\ge\epsilon/4,\ \left\|\sum_{i}\mathbb{E}_{Z^{(i)}}\left[(\mU^{(i)})^2\bigg|\mU^{(i-1)}\right]\right\|_2\le\sigma^2\right]\\
\le&2|\T|\cdot\exp\left(\frac{-\epsilon^2/48}{1600C'\delta^2\log{m}+20\delta\epsilon/3}\right)=1/\mathrm{poly}(m),
\end{align*}
that is, with probability at least $1-1/\mathrm{poly}(m)$,
\begin{align}
\label{ineq:Vi_le_epsilon/4}
\left\|\mV^{(i)}\right\|_2\le\epsilon/4. 
\end{align}

Now we finish the proof. By equality (\ref{eq:Yhat_minus_Yi0_simplify}) and the definition of $\mU^{(i)}$ (see (\ref{eq:def_Ui})), we have 
\begin{align*}
&\left\|{\hmY}^{(i,0)}-\mY^{(i,0)}\right\|_2=\left\|\sum_{j=0}^{i-1}{\hmX}^{(j)}\right\|_2=\left\|\sum_{j=0}^{i-1}\left(\mU^{(i)}+\mathbb{E}_{Z^{(i)}}\left[\hmX^{(j)}\right]\right)\right\|_2\\
=&\left\|\sum_{j=0}^{i-1}\mU^{(i)}+\sum_{j=0}^{i-1}\mathbb{E}_{Z^{(i)}}\left[{\hmX}^{(j)}\right]\right\|_2=\left\|\mV^{(i)}+\sum_{j=0}^{i-1}\mathbb{E}_{Z^{(i)}}\left[{\hmX}^{(j)}\right]\right\|_2\\
\le&\left\|\mV^{(i)}\right\|_2+\sum_{j=1}^{i-1}\left\|\mathbb{E}_{Z^{(i)}}\left[{\hmX}^{(j)}\right]\right\|_2.\tag{By triangle inequality}
\end{align*}
By inequality (\ref{ineq:Vi_le_epsilon/4}) and Lemma~\ref{lem_operator_bound_Xhat_EXhat}, we have that with probability at least $1-1/\mathrm{poly}(m)$, 
\[
\left\|{\hmY}^{(i,0)}-\mY^{(i,0)}\right\|_2\le\frac{\epsilon}{4}+\frac{C'\log{m}}{\alpha}\cdot 40\alpha\delta=\frac{\epsilon}{4}+\frac{40C'}{C\log{m}}\cdot\epsilon\le\frac{\epsilon}{2}.
\]
\end{proof}

Now we prove Theorem \ref{thm:approxsc_algebra}. By algorithm \ApproxSC, the returned graph $H$ satisfies that $|E(H)|=O(|\T|\epsilon^{-2}\log^2{m})$. Therefore, it remains to prove that $\mSC(H, \T)\approx_{\epsilon}\mSC(G, \T)$ with probability at least $1-1/\mathrm{poly}(m)$. Specifically, we bound the $\ell_2$ norm of the difference matrix $\hmY^{(i,0)}-\mP$ by considering two parts: the martingale $\mY^{(i,t)}$ and the errors accumulated by $\hmY^{(i,0)}-\hmY^{(i-1, |Z^{(i-1)}|)}$, which correspond to Lemma \ref{lem_bound_Y_it} and Lemma \ref{lem_widehat_Y_i0} respectively. 

\begin{proof}[Proof of Theorem \ref{thm:approxsc_algebra}]
By Lemma \ref{lem_bound_Y_it} and Lemma \ref{lem_widehat_Y_i0}, we have that 
\begin{align}
\label{hat_Y_tau0_minus_I}
\left\|{\hmY}^{(\tau,0)}-\mP\right\|_{2}\le\left\|{\hmY}^{(\tau, 0)}-\mY^{(\tau, 0)}\right\|_{2}+\left\|\mY^{(\tau, 0)}-\mP\right\|_2\le\epsilon/2+\epsilon/2=\epsilon.
\end{align}
Note that 
\begin{align*}
    {\hmY}^{(\tau, 0)}=\mS_{0}\mL(H^{(\tau)})^{\dag}\mS_{0}^T=\mS_{0}\mL(H)^{\dag}\mS_{0}^T=\mSC(G, \T)^{1/2}\left(\mL(H)^{\dag}\right)_{[\T, \T]}\mSC(G, \T)^{1/2}, 
\end{align*}
and $\mP=\mSC(G,\T)^{\dag/2}\mSC(G,\T)\mSC(G,\T)^{\dag/2}$, then inequality (\ref{hat_Y_tau0_minus_I}) tells us that 
\begin{equation}
\label{eq:Yhat_tau0_approx_P}
\mSC\left(G, \T\right)^{1/2}
\left(\mL\left(H\right)^{\dag}\right)_{\left[\T, \T\right]}
\mSC\left(G,\T\right)^{1/2}
\approx_{\epsilon}
\mSC\left(G,\T\right)^{\dag/2}
\mSC\left(G,\T\right)
\mSC\left(G,\T\right)^{\dag/2}.
\end{equation}
Multiplying the both sides of the LHS and RHS of (\ref{eq:Yhat_tau0_approx_P}) by $\mSC(G, \T)^{\dag/2}$ gives
\[
\mP\left(\mL(H)^{\dag}\right)_{[\T, \T]}\mP\approx_{\epsilon}\mSC(G, \T)^{\dag}.
\]
By Lemma \ref{lemma:Inverse}, it holds that $\mP\left(\mL(H)^{\dag}\right)_{[\T, \T]}\mP=\mSC(H, \T)^{\dag}$, therefore, 
\[
\mSC(H, \T)^{\dag}\approx_{\epsilon}\mSC(G, \T)^{\dag},
\]
that is, 
\[
\mSC(H, \T)\approx_{\epsilon}\mSC(G, \T).
\]
\end{proof}

Finally, we prove Theorem~\ref{thm:MinorSC}. Since the algorithm {\ApproxSC} only applies deletions and contractions on the input graph, it follows that the resulting sparsifier is a minor. The correctness and the bound on the number of edges follow from Theorem \ref{thm:approxsc_algebra}. Therefore, it remains to bound the computation cost.

\begin{proof}[Proof of Theorem \ref{thm:MinorSC}]
By Lemma~\ref{lemma:runtime}, the number of iterations in the main while loop of the algorithm \ApproxSC~
(Line~\ref{ln:ApproxSCWhile} of algorithm~\ref{algo:sparsify}) is
$O(\alpha^{-1} \log m) = O(\eps^{-1}\log^5 n)$.
By  Lemma \ref{lem:LevApx}, \ref{lem:split}, \ref{lemma:FindSteady}, the communication cost of each iteration is dominated by line \ref{line:local} and line \ref{line:LeverageEstimate}, which require solving $O(\d^{-2}\log n)=O(\eps^{-2}\log^5 n)$ Laplacian linear systems, and $O(\rho\epsilon^{-2}\sqrt{\overline{n}}\log{\overline{n}}\log^5{n}+D\log{n})$ rounds of communication in $\overline{G}$.

Therefore, the total number of required Laplacian solvers is 
\[
O\left(\eps^{-1}\log^5 n \cdot \eps^{-2}\log^5 n\right)
=O\left(\eps^{-3} \log^{10} n\right).
\]
The total overhead cost of communication in $\overline{G}$ can be bounded in the same way, that is, 
\[
O\left(\epsilon^{-1}\log^5{n}\left(\rho\epsilon^{-2}\sqrt{\overline{n}}\log{\overline{n}}\log^5{n}+D\log{n}\right)\right)=O\left(\rho\epsilon^{-3}\sqrt{\overline{n}}\log{\overline{n}}\log^{10}{n}+D\epsilon^{-1}\log^6{n}\right).
\]
\end{proof}

\section{Vertex and Edge Reductions}
\label{sec:Solvers}
Here we show our reductions via tree and elimination based preconditioners in Section~\ref{subsec:ultrasparsify} and Section~\ref{subsec:eliminiation} respectively.
This will prove Lemmas~\ref{lem:UltraSparsify}~and~\ref{lem:Elimination}.
\subsection{Ultra-Sparsifier}
\label{subsec:ultrasparsify}

We prove the high error reduction routine as stated in Lemma~\ref{lem:UltraSparsify}

\UltraSparsify*

We follow the construction from~\cite{KMP10}, which samples off-tree
edges with any upper bound on their stretches.
To find the tree, we utilize the distributed version of the Alon-Karp-Peleg-West
(AKPW) low stretch spanning tree, due to Ghaffari, Karrenbauer, Kuhn, Lenzen,
and Patt-Shamir~\cite{GKKLP15}.
They work with a definition of distributed $N$-node	cluster graphs
that was the basis of our definition of distributed $\rho$-minor.
We start by restating this definition, and describe how we simulate it
when $G$ is itself embedded.

\begin{definition}
	A distributed $N$-node cluster graph $\mathcal{G} = (\mathcal{V},
	\mathcal{E},\mathcal{L}, \mathcal{T},
	\psi)$ is defined by a set of $N$ clusters $\mathcal{V} = \{S_1, \ldots, S_N \}$
	partitioning
	the vertex set $V$ , a set of weighted multiedges, a set of cluster leaders
	$\mathcal{L}$, a set of cluster trees
	$\mathcal{T}$, as well as a function $\psi$ that maps the edges $\mathcal{E}$ of
	the cluster graph to edges in $E$. Formally, the tuple
	$(\mathcal{V}, \mathcal{E},\mathcal{L}, \mathcal{T}, \psi)$ has to satisfy the
	following conditions.
	\begin{enumerate}
		\item  The clusters $\mathcal{V} = (S_1, \dots , S_N)$ form a partition of the
		set of vertices $V$.
		\item For each cluster $S_i$, $|S_i\cap \mathcal{L}| = 1$. Hence, each cluster
		has exactly one cluster leader $\ell_i \in \mathcal{L} \cap S_i$. The
		ID of the node $\ell_i$ also serves as the ID of the cluster $S_i$ and for the
		purpose of distributed computations,
		we assume that all nodes $v \in S_i$ know the cluster $ID$ and the size $n_i
		:= |S_i|$ of their cluster $S_i$.
		\item  Each cluster tree $T_i = (S_i
		, E_i)$ is a rooted spanning tree of the subgraph $G[S_i]$ of $G$ induced by
		$S_i$.
		The root of $T_i$ is the cluster leader $\ell_i \in S_i \cap \mathcal{L}$.
		\item The function $\psi : \mathcal{E}\rightarrow E$ maps each edge of
		$\mathcal{E}$ to an (actual) edge of $E$ connecting
		the clusters.
	\end{enumerate}
\end{definition}

As a consequence of Lemma~\ref{lemma:Communication}, we get that shortest paths
can be ran on distributed $N$-node cluster graphs of $G$

\begin{lemma}
\label{lem:DfnsTalkWell}
	Let $G = (V, E)$ be a graph with $n$ vertices and $m$ edges
	that $\rho$-embed into the communication network $\overline{G} = (\overline{V},
	\overline{E})$,
	and $\mathcal{G} = (\mathcal{V}, \mathcal{E},\mathcal{L}, \mathcal{T}, \psi)$ be a
	distributed cluster graph for $G$.

	Then we have the following algorithms:
	\begin{enumerate}
		\item For each cluster $S_i$, the cluster leader $\ell_i$ broadcasts $O(\log
		\overline n)$ bit message $s_i$ to each vertex of $S_i$ in $O(\rho \overline n^{1/2} \log \overline n +
		D)$ rounds.
		\item
		Assume every vertex $v \in S_i$ for each $S_i \in \mathcal{V}$,
		the corresponding vertex $v' \in \overline{V}$ holds a value $f(v)$.
		Then computing $\min_{v\in S_i}\{f(v)\}$ at node $\ell_i$ for each $S_i \in
		\mathcal{V}$ needs $O(\rho \overline{n}^{1/2} \log \overline n+D)$
		rounds if the tree $T_i$
		with root $\ell_i$ is known.
	\end{enumerate}
\end{lemma}

\begin{proof}
  The definition of distributed $N$-node cluster graphs implies that
  $\mathcal{G}$ $1$-minor distributes over $G$.
  Lemma~\ref{lemma:MinorCompose} then gives that $\mathcal{G}$ $\rho$-minor
  distributes over $\overline{G}$, and this distributed mapping can be obtained
  using $O(\rho \overline{n}^{1/2} \log \overline n+D)$ rounds of
  computations.
  The broadcast, and the aggregation of minimums then follow from
  Lemma~\ref{lemma:Communication}.
\end{proof}

This in turn implies that the \SplitGraph algorithm in \cite{GKKLP15}
can be simulated on a graph that's $\rho$-minor distributed into $\overline{G}$
in  $O(n^{o(1)}(\rho\overline n^{1/2} \log \overline n+D))$ rounds.
Putting it together gives our variant of the AKPW low stretch spanning tree
algorithm, with the main difference being that it's ran on a $\rho$-minor
distributed over our overall communication network.

\begin{lemma}
\label{lemma:LowStretch}
Let $G = (V, E)$ be a graph with $n$ vertices and $m$ edges
that $\rho$-embeds into the communication network $\overline{G} = (\overline{V},
\overline{E})$, and $\mathcal{G} = (\mathcal{V}, \mathcal{E},\mathcal{L}, \mathcal{T},
psi)$ be a
distributed cluster graph for $G$.

It takes $O(n^{o(1)}(\rho\overline n^{1/2} \log \overline n+D))$ rounds to construct a
spanning tree $T$ of $G$, along with stretch upper bounds that sum to
\[
m \cdot 2^{O(\sqrt{\log{n}\log\log{n}})}.
\]
\end{lemma}

These upper bounds are sufficient for sampling the edges by stretch.
The following was shown in~\cite{KMP10}, or Theorem 2.2.4 in \cite{P13:thesis}.

\begin{lemma}
\label{lemma:SampleByStretch}
Given a graph $G$, a tree $T$,
upper bounds on stretches of edges of $T$ w.r.t. $G$ that sum to $\alpha$,
along with a parameter $k$,
there is an independent sampling / rescaling distribution computable
locally from the stretch upper bounds that gives a graph $H$ such that
with high probability
\begin{enumerate}
\item $\mL(G) \pe \mL(H) \pe k\mL(H)$
\item $H$ contains (rescaled) $T$, plus $O(\alpha \log{n} / k)$ edges.
\end{enumerate}
\end{lemma}

We then need to contract the tree so that its size becomes similar
to the number of off-tree edges.

\begin{lemma}
\label{lemma:Deg1Elim}
	Let $H = (V, E)$ be a graph with $n$ vertices and $m$ edges
	that $\rho$-embed into the communication network $\overline{G} = (\overline{V},
	\overline{E})$.

	Let $T$ be a spanning tree of $H$ and $W = E - T$ be the set of off-tree edges of
	$H$ with respect to $T$.
	There is an algorithm to compute a graph $\widehat{G}$  that's $1$-embeddable into $H$
	satisfying the following conditions in $O(\rho \overline n^{1/2} \log \overline n + D)$ rounds:
	\begin{enumerate}
		\item $\widehat{G}$ contains $O(|W|)$ vertices and edges.
		\item There are operators $\mZ_1$
		and $\mZ_2$ that can be evaluated in $O(\rho \overline n^{1/2} \log \overline n + D)$ rounds with
		\[
		\mL\left( H \right)^{\dag}
		=
		\mZ_1^{\tomato}
		\left[
		\begin{array}{cc}
		\mZ_{2} & 0\\
		0 & \mL\left( \widehat{G} \right)^{\dag}
		\end{array}
		\right]
		\mZ_1
		\]
	\end{enumerate}
\end{lemma}
\begin{proof}
We use the parallel elimination procedure from Section 6.3 of \cite{BGKMPT14}, specifically Lemma 26. At a high level, it eliminates degree $1$ and $2$ vertices by random sampling a subset of vertices which have degree $1$ or $2$, computing an independent set, and eliminating them. The algorithm requires $O(\log n)$ rounds in PRAM, and therefore can be implemented in $O(\log n(\rho\sqrt{\overline{n}}\log\overline{n}+D))$ rounds in the \congest~model by Lemma \ref{lemma:Communication}. The operators $\mZ_1, \mZ_2$ are computed as in Lemma 26 of \cite{BGKMPT14}.
\end{proof}

We can combine these pieces to prove the main ultrasparsification claim.
\begin{proof}[Proof of Lemma \ref{lem:UltraSparsify}]
The algorithm to prove Lemma \ref{lem:UltraSparsify} is as follows.
\begin{enumerate}
\item Compute a low-stretch tree using Lemma \ref{lemma:LowStretch}.
\item Sample edges using Lemma \ref{lemma:SampleByStretch}
with $\alpha = m \cdot 2^{O(\sqrt{\log n\log\log n})}$.
\item Compute the operators $\mZ_1, \mZ_2$ using Lemma \ref{lemma:Deg1Elim}.
\end{enumerate}
We verify the conditions of Lemma \ref{lem:UltraSparsify}.
The approximation guarantees and number of off-tree edges
are given by Lemma~\ref{lemma:SampleByStretch}

After eliminating  degree $1$ and degree $2$ vertices, the resulting graph has
size $O(\alpha \log n)$ by Lemma \ref{lemma:Deg1Elim} Part 1,
and $\mZ_1$ and $\mZ_2$ are computed by Lemma~\ref{lemma:Deg1Elim}.

 The round complexity in the \congest~model follows by summing the round complexities
 in Lemmas \ref{lemma:LowStretch}, \ref{lemma:SampleByStretch}, \ref{lemma:Deg1Elim}.
\end{proof}

\subsection{Elimination / Sparsified Cholesky}
\label{subsec:eliminiation}

The main goal of this section is to prove Lemma~\ref{lem:Elimination},
which allows the elimination of large subsets of vertices under small error.

\Elimination*

To prove the above the above lemma, we present a distributed implementation of the \emph{Schur Complement Chain} (SCC) construction due to Kyng, Lee, Peng, Sachdeva, and Spielman~\cite{KLPSS16}. The key components to this construction are (i) an algorithm that finds a large near-independent set $F$, and approximates the inverse of the matrix restricted to entries in $F$ and (ii) a procedure for spectrally approximating the Schur complement with respect to $C = V \setminus F$. We next discuss how to implement these components in the \congest~model.

\paragraph{Finding large $\alpha$-DD sets.} When doing Gaussian elimination, the goal is to find a large subset of vertices $F$ such that we can approximate the inverse of $\mL_{[F,F]}$ by an operator $\mZ$ that can be constructed efficiently. Ideally, $F$ forms an independent set. Unfortunately, we are not able to find a large independent set but we can instead find a large, almost-independent set, as made precise in the following definition.

\begin{definition}[$\alpha$-DD] \label{def:alphaSDD}
	A matrix $\mM $ is $\alpha$-diagonally dominant ($\alpha$-DD) if
	\[
	\forall i, \quad \mM_{i,i} \geq (1+\alpha) \sum_{j: j \neq i} \mM_{i,j}.
	\]
	An index set $ F $ is $\alpha$-DD if $ \mM_{[F,F]} $ is $\alpha$-DD.
\end{definition}

The algorithm due to \cite{KLPSS16} for finding $\alpha$-DD sets in a Laplacian proceeds as follows: (i) pick a random subsets of vertices and (ii) and discard all those that do not satisfy the condition in Definition~\ref{def:alphaSDD}. The pseudocode for computing such sets is given in Algorithm~\ref{algo:DDsubset}. In the \congest~model, the way the set is ``stored'' is that each vertex remembers whether it is in the set.

\begin{algorithm}[ht]
	\caption{Find an $\alpha$-DD subset $F$ of $\mL$ \label{algo:DDsubset}}
	\SetKwProg{Proc}{procedure}{}{}
	\Proc{$\DDSubset(\mL, \alpha)$}{
	Sample each index of $\{1,\ldots,n\}$ independently with probability $ \frac{1}{4 (1 + \alpha)} $ and let $ F' $ be the resulting set of sampled indices. \\
	Set \begin{equation*}
			F = \left\{i \in F' : |\mL_{i,i}| \geq (1+\alpha) \sum_{j \in F', j \neq i} |\mL_{i,j}| \right\}.
				\end{equation*} \\
	\If{$|F| < \frac{n}{8(1+\alpha)}$} { Goto Step 1. }
	\Return $F$.
	}
\end{algorithm}

We have the following lemma.
\begin{lemma} \label{lem:DDsubset}
Let $G=(V,E)$ be a graph that $\rho$-minor distributes into the communication network
$\overline{G} = (\overline{V}, \overline{E})$. Let $\mL$ be the Laplacian matrix
associated with $G$ and let $\alpha \geq 0$ be a parameter. Then
$\DDSubset(\mL,\alpha)$ computes an \emph{$\alpha$-DD} subset $F$ of $\mL$ of size
$n/(8(1+\alpha))$ in $O(\rho \sqrt{\overline{n}} \log \overline{n} + D)$ rounds.
\end{lemma}

\begin{proof}
In~\cite[Lemma 5.2]{LeePS15} (and more generally in \cite{KLPSS16}), it is shown that
Algorithm~\ref{algo:DDsubset} computes an $\alpha$-DD subset $F$ of size
$n/(8(1+\alpha))$. To bound the round complexity of the algorithm, consider the
following distributed implementation:
\begin{enumerate}
  \item \label{step:sampleDD} Include each index of $\{1,\ldots,n\}$ in $F'$ with
  probability $\frac{1}{4(1+\alpha)}$.
  \item \label{step:localDD}
   Each node corresponding to $i \in F'$ sums up the values $ |\mL_{i,j}| $ of the
   indices $ j $ corresponding to its neighbors in $ G $, and then decides whether
   $|\mL_{i,i}| \geq (1+\alpha) \sum_{j \in F', j \neq i}
  |\mL_{i,j}| $ and if so declares itself as belonging to $ F $.
 \item \label{step:sum} The size of $F$ is computed by an (arbitrarily decided) leader vertex,
 which aggregates the sum of the following values over all vertices $ v $ in $G$: $1$ if $ v $ is in $F$ and $0$ otherwise.
  \item \label{step:notification} The leader checks whether $|F| < n/(8(1+\alpha))$.
  If the latter holds, then the leaders informs all the vertices in $G$ to repeat the
  previous steps. Otherwise, the algorithm terminates.
\end{enumerate}

	In the \congest~model, Step~\ref{step:sampleDD} requires no communication between
	the nodes: each root vertex of supervertices does the sampling independently.
	In Step~\ref{step:localDD}, each node computes an aggregate of values stored by
	its neighbors in $ G $, which by Lemma~\ref{lemma:Communication} takes $O(\rho
	\sqrt{\overline{n}} \log \overline{n} + D)$ rounds.
	It is well-known that Steps~\ref{step:sum} and~\ref{step:notification} can be carried out in $ O (D) $ rounds by routing the messages via a BFS tree rooted at the leader.
	Together with the fact that Algorithm~\ref{algo:DDsubset} terminates in at most $2$ iterations in expectation (see~\cite[Lemma 5.2]{LeePS15}), it follows that the distributed implementation takes $O(\rho \sqrt{\overline{n}} \log \overline{n} + D)$ rounds in expectation.
\end{proof}

\paragraph{Jacobi Iteration on $\alpha$-DD matrices. } Using an $\alpha$-DD set $F$, we will construct an operator $\mZ$ that approximates $\mL^{-1}_{[F,F]}$ and can be applied efficiently to any vector. An important observation is that we can write $\mL_{[F,F]} = \mX_{[F,F]} + \mY_{[F,F]}$, where $\mX_{[F,F]}$ is a diagonal matrix and $\mY_{[F,F]}$ is a Laplacian matrix. We have the following lemma.

\begin{lemma} \label{lem:Jacobi} Let $G=(V,E)$ be a graph that $\rho$-minor
distributes into the communication network $\overline{G} = (\overline{V},
\overline{E})$.
  Let $\mL$ be the Laplacian matrix associated with $G$ and let $F$ be a subset of
  of $V$ such that $\mL_{[F,F]}$ is $\alpha$-DD for some $\alpha \geq 4$. Then
  $\Jacobi(\mL_{[F,F]}, \cdot, \epsilon)$  gives a linear operator $\mZ$
  that over vectors given on the root vertices of the supervertices
  such that for any vector $\vec{b}$ given by storing $\vec{b}_{v^{G}}$ on
  $V_{map}^{G \rightarrow \overline{G}}(\cdot)$,
  returns in $O((\rho\sqrt{\overline{n}} + D) \log(1 / \epsilon))$ rounds
  $\mZ\vec{b}$ stored on the same vertices, for some matrix $\mZ$ such that
  \[
    \mL_{[F,F]} \preceq \mZ^{(-1)} \preceq \mL_{[F,F]} + \epsilon \cdot \mSC(\mL,F).
  \]
\end{lemma}

Note that the matrix $\mZ$ is only used in the analysis,
and is never explicitly constructed by the algorithm.
We first give the pseudocode of this algorithm in the centralized setting,
and then show its distributed implementation.

\begin{algorithm}[ht]
	\caption{Solve $\mL_{[F,F]} \cdot \vec{x}_{F} = \vec{b}_{F}$ up to $\epsilon$ accuracy \label{algo:Jacobi}}
	\SetKwProg{Proc}{procedure}{}{}
	\Proc{$\Jacobi(\mL_{[F,F]}, \vec{b}_F, \epsilon)$}{
		Set $\mL_{[F,F]} = \mX_{[F,F]} + \mY_{[F,F]}$ such that $\mX_{[F,F]}$ is diagonal and $\mY_{[F,F]}$ is a Laplacian. \\
		Set $k$ to be an odd integer that is greater than $\log (3/\epsilon)$. \\
		Set $\vec{x}_{F}^{(0)} = \mX^{-1}_{[F,F]} \vec{b}_F$. \\
		\For{$i = 1,\ldots,k$}
		{
			Set $\vec{x}_{F}^{(i)} = -\mX_{[F,F]}^{-1} \mY_{[F,F]} \vec{x}_F^{(i-1)} +
			\mX^{-1}_{[F,F]} \vec{b}_F$.
            \label{line:step}
		}
		\Return $\vec{x}_{F}^{(k)}$.

	}
\end{algorithm}

To measure the quality of the operator produced by \Jacobi~procedure, we observe that $k$ iterations produce the operator
\begin{equation} \label{eq:operator}
\mZ^{\left(k\right)}
:=
\sum_{i=0}^{k} \mX^{-1}_{\left[F,F\right]}
\left(-\mY_{\left[F,F\right]} \mX^{-1}_{\left[F,F\right]} \right)^{i}
\end{equation}
by induction. Concretely, suppose we have
\[
\vec{x}_{F}^{\left(k - 1\right)}
=
\sum_{i=0}^{k - 1} \mX^{-1}_{\left[F,F\right]}
\left(-\mY_{\left[F,F\right]} \mX^{-1}_{\left[F,F\right]} \right)^{i} \vec{b}_{F},
\]
then substituting this into the step in Line~\ref{line:step} gives
\begin{align*}
\vec{x}_{F}^{\left(k\right)}
& =
-\mX_{\left[F,F\right]}^{-1} \mY_{\left[F,F\right]}
\vec{x}_{F}^{\left(k - 1\right)}
+
\mX^{-1}_{\left[F,F\right]} \vec{b}_F\\
& = \mX^{-1}_{\left[F,F\right]} \vec{b}_F
+ \left( -\mX_{\left[F,F\right]}^{-1} \mY_{\left[F,F\right]} \right)
\sum_{i=0}^{k - 1} \mX^{-1}_{\left[F,F\right]}
\left(-\mY_{\left[F,F\right]} \mX^{-1}_{\left[F,F\right]} \right)^{i} \vec{b}_{F}\\
& = \mX^{-1}_{\left[F,F\right]} \vec{b}_F
+ \sum_{i=1}^{k} \mX^{-1}_{\left[F,F\right]}
\left(-\mY_{\left[F,F\right]} \mX^{-1}_{\left[F,F\right]} \right)^{i} \vec{b}_{F}
= \sum_{i=0}^{k} \mX^{-1}_{\left[F,F\right]}
\left(-\mY_{\left[F,F\right]} \mX^{-1}_{\left[F,F\right]} \right)^{i} \vec{b}_{F}.
\end{align*}

We next review two lemmas from~\cite{KLPSS16} that help us prove the approximation accuracy of the Jacobi iteration on $\mL_{[F,F]}$. The first shows that $\alpha$-DD matrices admit good diagonal preconditioners. The second gives a way to bound the error produced by \Jacobi.

\begin{lemma}[\cite{KLPSS16}, Lemma 3.6.] \label{lem:Jacobipreconditioner}
	Let $\mL_{[F,F]}$ be an $\alpha$-DD matrix which can be written in the form $\mX_{[F,F]} + \mY_{[F,F]}$ where $\mX_{[F,F]}$ is diagonal and $\mY_{[F,F]}$ is a Laplacian. Then $\frac{\alpha}{2} \mY \preceq \mX$.
\end{lemma}

\begin{lemma}[\cite{KLPSS16}, Lemma E.1.] \label{lem:Jacobisolver}
	Let $\mL_{[F,F]}$ be an $\alpha$-DD matrix with $\mL_{[F,F]} = \mX_{[F,F]} + \mY_{[F,F]}$ where $0 \preceq \mY_{[F,F]} \preceq \beta \mX$ for some $0 < \beta < 1$. Then, for any odd $k$ and $\mZ^{(k)}$ as defined in Eq.~(\ref{eq:operator}), we have
	\[
		\mX_{[F,F]} + \mY_{[F,F]} \preceq (\mZ^{(k)})^{-1} \preceq \mX_{[F,F]} + (1+\delta) \mY_{[F,F]},
	\]
	where
	\[
		\delta = \beta^{k} \frac{1+\beta}{1 - \beta^{k+1}}.
	\]
\end{lemma}

\begin{proof}[Proof of Lemma~\ref{lem:Jacobi}]
Let $\mY_{[F,F]}$ be the matrix generated when calling \Jacobi~with $\mL_{[F,F]}$. Since $\mL_{[F,F]}$ is an $\alpha$-DD matrix, by extending $\mY_{[F,F]}$ with zero entries, we have $\mY \preceq \mL$. This in turn implies that $\mY_{[F,F]}  = \mSC(\mY,F) \preceq \mSC(\mL,F)$.

Lemma~\ref{lem:Jacobipreconditioner} gives that $\frac{\alpha}{2} \mY \preceq \mX$. As $\alpha \geq 4$, we can invoke Lemma~\ref{lem:Jacobisolver} with $\beta = 1/2$, which in gives that $(1+\beta)/(1-\beta^{k+1}) \leq 3$. Therefore, our choice of $k = \log (3/\epsilon)$ gives the desired error guarantee. To bound the round complexity of the algorithm, consider the following distributed implementation of \Jacobi:
\begin{enumerate}
  \item Store the values $\vec{b}_{F}(u)$, $\mX_{[F,F]}(u,u)$, $\mY_{[F,F]}(u,u)$
  at $V_{map}^{G \rightarrow \overline{G}}(u)$.
  Store the off-diagonal entries of $\mY_{[F,F]}$, together with their weights, in the
  endpoints of mapped edges $E_{map}^{G \rightarrow \overline{G}}(e)$:
  this is possible because $\mY_{[F,F]}$ is a Laplacian.
  \item Set $k = \log (1/\epsilon)$ and  $\vec{x}_{F}^{(0)}(u) =
  \mX^{-1}_{[F,F]}(u,u) \cdot \vec{b}_{F}(u)$ for each $u \in F$.
  \item \label{step:Jacobiupdate} For $i = 1,\ldots,k$ do
\begin{enumerate}
	\item For each $u \in F$ set
\[
\vec{x}_F^{(i)}(u) \leftarrow
\mX^{-1}_{[F,F]}(u,u) \cdot \vec{b}_{F}(u)
-\mX^{-1}_{[F,F]}(u,u)\sum_{v \in F : \mY_{[F,F]}(u,v) \neq 0} \mY_{[F,F]}(v,u)
\vec{x}^{(i-1)}(v)
\]
using the matrix-vector multiplication primitive from Corollary~\ref{corollary:MatVec},
with results stored on all root vertices, $V_{map}^{G \rightarrow \overline{G}}(u)$.
\end{enumerate}
    \item Every vertex $u \in F$ returns $\vec{x}^{(k)}_F(u)$.
\end{enumerate}

The complexity of the algorithm is dominated by the number of rounds to implement
Step~\ref{step:Jacobiupdate}, which is given by Corollary~\ref{corollary:MatVec}.
As there are $k = \log (1/\epsilon)$ iterations, we have that
Step~\ref{step:Jacobiupdate} requires $O(\log(1/\epsilon))$ rounds in $G$.

By Lemma~\ref{lemma:Communication} and using the fact that weights of the network, and
hence the solution vector's magnitudes, are polynomially bounded, we can simulate the
algorithm in the original communication network $\overline{G}$ in $O((\rho
\sqrt{\overline{n}} \log \overline{n} + D) \log (3/\epsilon) )$ rounds.
\end{proof}

\paragraph{Approximating Schur complements using low congestion random
walks.}
We next show that $\alpha$-DD sets are useful when approximating Schur complements. A key ingredient to our construction is the following combinatorial view of Schur complements.

It is well known that $\mSC(\mL, \T)$ is a Laplacian matrix of a graph on vertices
in $\T = V \setminus F$.
For our purposes, it will be useful interpret $\mSC(\mL, \T)$ in terms
of random walks. To this end, given a walk $W=u_0,\ldots,u_l$ of length $\ell$ in $G$
with a subset of vertices $\T$, we say that $W$ is a \emph{terminal-free} walk if
$u_0, u_{\ell} \in \T$ and $u_1,\ldots,u_{\ell-1} \notin \T$.
\begin{lemma}
For any undirected, weighted graph $G$ and any subset of vertices $\T$, the Schur
Complement $\mSC(G,\T)$ is given as a union over all multi-edges corresponding to
terminal-free walks $u_0,\ldots,u_\ell$ with weight
  \[
    \frac{\prod_{0 \leq  i < k} \vec{w}_{u_i u_{i+1}}}
      {\prod_{0 \leq  i \leq k} \sum_{u_iv \in E(G)} \vec{w}_{u_i v}}
  \]
\end{lemma}

The theorem below allows us to efficiently sample from this distribution of walks while paying a small cost in the approximation quality.

\begin{lemma}[Theorem~3.1 in~\cite{DGGP19}] \label{lem:sampleSC} Let $G=(V,E)$ be an undirected, weighted graph with
a subset of vertices $\T$.
Let $\epsilon \in (0,1)$ be an error parameter and $\mu = \Theta(\epsilon^{-2}\log n)$ be
some parameter related to the concentration of sampling. Let $H$ be an initially empty
graph, and for every edge $e=(u,v) \in G$, repeat $\mu$ times the following procedure, where
a random step from a vertex is taken proportional to the edge weights of its adjacent edges.
  \begin{enumerate}
    \item Simulate a random walk starting from $u$ until it hits $\T$ at vertex $t_1$.
    \item Simulate a random walk starting from $v$ until it hits $\T$ at vertex $t_2$.
    \item Let $\ell$ be the total length of this combined walk~(including edge $e$).
    Add the edge $(t_1,t_2)$ to $H$ with weight
    \[
    \frac{1}{\mu \sum_{i=0}^{\ell-1} (1/\vec{w}_{u_i,u_{i+1}})}.
    \]
  \end{enumerate}
  The resulting graph $H$ satisfies $\mL(H) \approx_{\epsilon} \mSC(\mL(G), \T)$
  with high probability.
\end{lemma}

The main idea to make use of the above theorem is to compute an $\alpha$-DD set $F$
using Lemma~\ref{lem:DDsubset} as this ensures that the random walks in the graph are
short in expectation. However, since we are dealing with weighted graphs, there might
be scenarios where the expected congestion of an edge is prohibitively large, which makes it
difficult to recursively repeat the algorithm. To alleviate this, we add new vertices
to the terminals, whenever they have too much congestion.
Note that because $G$ is distributed over $\overline{G}$ as a minor, we can only
accommodate small vertex congestion due to the need for each root node to inform the entire
supervertex.
As a result, our resulting congestion depends on the average degree, and we resolve
this via calling sparsification (Corollary~\ref{corollary:Sparsify}) at each step.

Let $W$ be the family of walks generated in Lemma~\ref{lem:sampleSC}. For $e \in E$,
let $\load_W(e)$ denote the number of walks from $W$ that use the edge $e$.
Algorith~\ref{algo:randomWalk} below computes $W$.
Note that it can also be run implicitly to generate the congestion on every edge
without exceeding the communication limit on any edges:
we simply pass around the congestion on every edge.

\begin{algorithm}[ht]
	\caption{Generate random walks from each edge until they hit terminals \label{algo:randomWalk}}
	\SetKwProg{Proc}{procedure}{}{}
	\Proc{$\RandomWalk(G, \T,\mu)$}{
	Set $W \gets \emptyset$ \\
	\For{$e=(u,v) \in E(G)$ in parallel}{
		\For{$i = 1,\ldots,\mu$ in parallel}
		{
			Generate a random walk $W(u,i)$ from $u$ until it hits $\T$ at vertex
			$t_1$. \label{line:rw1} \\
			Generate a random walk $W(v,i)$ from $v$ until it hits $\T$ at vertex
			$t_2$. \label{line:rw2} \\
			Set $W(e,i) = W(u,i) \cup (u,v) \cup W(v,i)$ and $W = W \cup \{W(e,i)\}$.
		}
	}
	\Return $W$.
	}
\end{algorithm}
\paragraph{Implementation of random walks.} We implement the random walks in lines \ref{line:rw1} and \ref{line:rw2} of \RandomWalk~as in Algorithm \ref{algo:randomWalk} as follows. When a non-terminal vertex gets a random walk edge into it, first we do a sum aggregation so that the leader of the corresponding cluster / super-vertex knows how many out-edges to compute. Below, we will ensure that all non-terminal vertices have low congestion as intermediate vertices of random walks, so we will focus on sampling a single out-edge. To do this, the vertex aggregates the sum of weights of out-edges, which we will call $W_{total}$. Now, the leader samples a random real number $r \in [0, W_{total}]$. Finally, the leader vertex does a binary search on the label of the out-edge and aggregates sums to figure out which out-edge corresponds to the sample $r$. All these steps can be done by Lemma \ref{lemma:Communication}.

We now move on to giving the distributed random walk based algorithm
that approximates Schur complements.
Psuedocode of this routine is in Algorithm~\ref{algo:randomWalkSchur}.
\begin{algorithm}[ht]
  \caption{Distributed Approximate Schur complements using random walks \label{algo:randomWalkSchur}}
  \SetKwProg{Proc}{procedure}{}{}
  \Proc{$\RandWalkSchur(G, \T,\epsilon,\gamma, \alpha)$}{
    Initialize $\widehat{\T} \leftarrow \T$ \\
    Set $H \gets \emptyset$  and $\mu \gets O(\eps^{-2}\log n)$ \\
    Implicitly compute the expected congestion of $W = \RandomWalk(G,\T,\mu)$ by
     propagating the expected congestion on vertices and edges evenly to neighbors
       for $O(\alpha \log{n})$ steps.
      \label{line:implicit} \\
    \For{all vertices $u$ with $\E[\load_W(u)] > \gamma$, in parallel} {
      Add $u$ to $\widehat{\T}$, $\widehat{\T} \gets \widehat{\T} \cup \{u\}$.
    }
    Set $W$ $\gets \RandomWalk(G,\widehat{\T},\mu)$ (with walks explicitly
    generated). \label{line:realwalks} \\
    Initialize the minor distribution of $H$ into $\overline{G}$ by associating
    each terminal $t$ with all vertices involved in all random walks that ended at $t$,
    and building $T^{H \rightarrow \overline{G}}(t)$ to be a spanning tree of
    all edges used, plus $T^{G \rightarrow \overline{G}}(u)$
    of all vertices on these walks.\\
    \For{$W(e,i) \in W$} {
      Let $(t_1,t_2)$ be the endpoints of the walk $W(e,i)$. \\
      Let $\ell$ be the length of $W(e,i)$.\\
      Set $H  = H \cup (t_1,t_2, w(t_1,t_2))$ with $
      \vec{w}(t_1,t_2) := 1/(\mu \sum_{i=0}^{\ell-1}
      (1/\vec{w}_{u_i,u_{i+1}})).$
    }
    \Return $H$, $\widehat{\T}$.
  }
\end{algorithm}

Here, we discuss subtleties in the distributed implementation of Algorithm \RandomWalk~and \RandomWalkSchur~(Algorithms \ref{algo:randomWalk} and \ref{algo:randomWalkSchur}) in the \congest~model.

\begin{lemma} \label{lem:randomWalkSchur}
Let $G=(V,E)$ be a graph that $\rho$-minor distributes into the communication network
$\overline{G} = (\overline{V}, \overline{E})$.
Let $F$ be an $\alpha$-DD set,
$\T = V \setminus F$, $\epsilon \in (0,1)$ be an error parameter
and $\gamma \geq 1$ be a congestion parameter.
Then the procedure $\RandWalkSchur(G, \T, \epsilon, \gamma, \alpha)$
outputs in $O(\alpha \gamma \epsilon^{-2} \log n (\rho\sqrt{\overline{n}} \log
\overline{n} + D))$ rounds a graph $H$ along with its
$\alpha \gamma \log{n}$-minor distribution into $\overline{G}$
such that with high probability,
\[
\mL(H) \approx_{\epsilon} \mSC(\mL(G), \widehat{\T})
\]
for some (slightly larger) superset $\widehat{\T} \supseteq \T$ of size
at most $n - |F| + O(\alpha m \epsilon^{-2} \log n /\gamma)$.
\end{lemma}
\begin{proof}
	The spectral guarantee $\mL(H) \approx_{\epsilon} \mSC(\mL(G),\widehat{\T})$ follows directly
	from Lemma~\ref{lem:sampleSC}. To bound the size of $\T$, first note that by
	definition $\T = V \setminus F$ and thus $|\T| = n - |F|$. Next, as $F$ is an
	$\alpha$-DD set, the expected length of a random walk that starts at an endpoint
	of any edge in $G$ and hits a vertex in $\T$ is $O(\alpha)$. Our algorithm
	simulates $O(m \epsilon^{-2}\log n)$ random walks for each edge, thus the total
	congestion generated by these walks is $O(\alpha m \epsilon^{-2}\log n)$.

  The latter gives that there can be at most $O(\alpha m \epsilon^{-2}\log n /
  \gamma)$ vertices whose expected congestion is larger than $\gamma$,
  and \RandWalkSchur~adds these vertices to the set $\T$.
  It follows that $|\widehat{\T}| \leq n - |F| +
  O(\alpha m \epsilon^{-2} \log n /\gamma)$.
  For each vertex $u$, the congestion incurred by other edges are independent random
  variables	bounded by the length of the walks, which is $O(\log n)$.
  So by a Chernoff bound, the congestion of all edges with expected congestion less
  than $\gamma > O(\log^2 n)$ is at most $O(\gamma)$ with high probability.
  So after line \ref{line:implicit} of Algorithm
  \ref{algo:randomWalkSchur} adds all vertices with high expected congestion
  into the terminals (to form $\widehat{\T}$),
  all subsequent vertices in $V \setminus \widehat{\T}$ have vertex
  congestion at most $O(\gamma)$ in the second random walk in line \ref{line:realwalks} with high probability.

We next study the round complexity.
To this end, recall that the expected length of each walk in $W$ is
$O(\alpha \log{n})$ with high probability.

When we are only passing the congestion of a vertex to neighbors,
that is, running the walks implicitly,
each round can be executed in one round of message passing
as described in Lemma~\ref{lemma:Communication}.
As we execute $\mu = O(\epsilon^{-2}\log n)$ rounds for each edge,
we have that the round complexity of the congestion estimation part of \RandomWalk\
is $O(\alpha \rho (\sqrt{\overline{n}} + D)\epsilon^{-2} \log^2n)$.

For the explicit generation, we aggregate, for all non-terminal vertices of $G$,
the walks that reach them, to the root node of the corresponding super vertex
in $\overline{G}$.
Then we broadcast these walks outward, we simulate the choice of a random edge
by total weights of edges in subtrees
(which we compute via Lemma~\ref{lemma:RootTrees}).
As the node congestions are at most $2 \gamma$, Lemma~\ref{lemma:Communication}
lets us perform these propagations in
$\overline{G}$ in $O(\alpha \gamma \epsilon^{-2}\log n
(\rho\sqrt{\overline{n}} \log \overline{n} + D))$ rounds.

Finally, to create the minor distribution of $H$, the graph with the new random
walk edges, into $\overline{G}$, we extend the terminals to include the supervertices
of all intermediate (non-terminal) vertices.
As the non-terminals have at most $O(\gamma)$ walks through them, the resulting
mapping is still a $\rho \gamma$-minor distribution into $\overline{G}$.
\end{proof}

We remark that in this scheme, the end points of the new edges in $H$
cannot actually know these edges in a centralized manner
(e.g. aggregate them at root vertices).
Instead, such walks are only passed to the vertices in $\overline{G}$ that
correspond to the first and last edges of the corresponding walk.
This is because we can only guarantee low node congestion of intermediate vertices.
Note that that in turn necessitates us sparsifying the graph at every
intermediate step as well.

\paragraph{Vertex Sparsifier Chain.}

Bringing together the above algorithmic components leads to
an algorithm for computing a vertex sparsifier chain,
whose pseudocode is given in Algorithm~\ref{algo:eliminate} below.

\begin{algorithm}[h]
	\caption{Eliminate a large subset of vertices for $d$ rounds}
    \label{algo:eliminate}
	\SetKwProg{Proc}{procedure}{}{}
	\Proc{$\Eliminate(G,d,\epsilon)$}{
		Set $\mL^{(0)} \gets \mL$ and $\widehat{\T}_0 = V$. \\
		Compute a spectral sparsifier $\mM^{(0)} \approx_{\epsilon} \mL^{(0)}$~(Corollary \ref{corollary:Sparsify}). \\
		\For{$0 < i \le d$ iteratively}
		{
			Let $F_{i}$ be an $\alpha$-DD set of $\mM^{(i-1)}$~(Lemma~\ref{lem:DDsubset}). \\
			Construct an operator $(\mZ^{(i)})^{-1}$ that approximates $\mM_{[F,F]}^{(i-1)}$~(Lemma~\ref{lem:Jacobi}). \\
			Compute $\tmM^{(i+1)} \approx_{\epsilon} \mSC(\mM^{i-1},\widehat{\T}_i)$~(Lemma~\ref{lem:randomWalkSchur}) with $\widehat{\T}_i = \widehat{\T}_{i-1} - F_{i} + U_{i}$, where $U_i$ is the set of extra vertices added to ensure low congestion. \\
			Compute an $\epsilon$-spectral sparsifier $\mM^{(i+1)}$ of
			$\tmM^{(i+1)}$.~(Corollary \ref{corollary:Sparsify})
		}
		Let $\mZ_1$, $\mZ_2$ be stored implicitly as the product of matrices using the Cholesky factorization (Lemma \ref{lemma:cholesky}). \\
		\Return $\mM^{(d)}, \mZ_1, \mZ_2$.
	}
\end{algorithm}

\begin{proof}[Proof of Lemma~\ref{lem:Elimination}]
We start by analyzing the round complexity of the algorithm. By Corollary
\ref{corollary:Sparsify}, there is a distributed algorithm for computing a sparsifier
with $O(n \log^{5} n \eps^{-2})$ edges in $O(\eps^{-2}\log^{7} n)$ rounds. We call
\RandWalkSchur~ with $\gamma = 1000 \cdot c \cdot \alpha \eps^{-2}\log^{6}n $, where
$c$ is a large enough constant. By Lemma~\ref{lem:DDsubset}, we find a $4$-DD set of
size $n/(8(1+4)) = n/40$. These together imply that the number of vertices in
$G^{(i)}$ after $i$ steps in our algorithm is
\[n_i \leq \left(1-\frac{1}{40} + \frac{1}{1000}\right)n_{i-1} \leq \left(1 - \frac{1}{50}\right) n_{i-1}.\]
By induction, $n_d \leq (\frac{49}{50})^{d}n$.

In an iteration of the algorithm, the dominating cost is (1) approximating the Schur complement and (2) computing the spectral sparsifier. By our choice of $\alpha$ and $\gamma$ and Lemma~\ref{lem:randomWalkSchur}, the number of rounds required to implement the first one is $O(\epsilon^{-4}\log^{7} n  (\rho\sqrt{\overline{n}} \log \overline{n} + D))$. The second one introduces a $O(\epsilon^{-2}\log^{7} n ) $ overhead, which then gives a round complexity of $O(\epsilon^{-6}\log^{14} n  (\rho\sqrt{\overline{n}} \log \overline{n} + D))$ per one step in~\Eliminate. Thus after $d$ steps, the round complexity is:
\[ O( (\epsilon^{-6}\log^{14} n )^d (\rho\sqrt{\overline{n}} \log \overline{n} + D)). \]
The error is $(1\pm\eps)^d$ because we do $d$ rounds of elimination, and each round accumulates $(1\pm\eps)$-multiplicative error by using Lemma~\ref{lem:Jacobi} to bound the quality of the inverse of $\mZ^{(i)}$, Lemma~\ref{lem:randomWalkSchur} to bound the quality of the Schur complement, and Corollary \ref{corollary:Sparsify} to spectrally sparsify the Schur complement, each of which accumulates error $\eps/4.$
\end{proof}

\section{Implications in Graph Algorithms}
\label{sec:Implications}
In this section we use Theorem \ref{thm:Main} to give improved algorithms for maximum flow, min-cost flow, and shortest paths with negative weights in the \congest~model (Theorems \ref{thm:max-flow}, \ref{thm:min-cost-flow}, and \ref{thm:shortest-paths}). Our goal is to show that our distributed Laplacian solver can be used to achieve improved complexities for these three problems, so we provide pseudocode in Algorithms \ref{alg:madry_renew}, \ref{algo:min-cost-flow_renew}, \ref{shortestpaths} (full details of these algorithms are given in Appendix~\ref{appendix:flow}), and analyze the runtimes in the distributed setting.
Our runtimes come from using the Laplacian system solver in Theorem \ref{thm:Main} to
implement an interior point method until the graph has a low amount of residual flow
remaining, which we then route with augmenting path~\cite{GU15} or shortest path with positive weights~\cite{CM20}, both taking $\widetilde O(n^{1/2} D^{1/4} + D)$ rounds per iteration.
In the remainder of this section, we formalize this reasoning. There are several
additional technical pieces, as the algorithms of \cite{M16,CMSV17} require changing
the graph by adding edges, etc. Throughout, we assume that our Laplacian system
solvers are exact -- it is justified in the papers \cite{M16,CMSV17} that solving to
accuracy $1/\poly(m,U)$, where $U$ is the maximum weight / capacity suffices to
implement the interior point methods.

In Section \ref{sec:flowrounding}, we simulate Cohen's flow rounding algorithm \cite{cohen1995approximate} in the \congest model. Given an $s-t$ flow $\Vec{f}$, it returns an integral $s-t$ flow $\vec{f'}$ in $O(\sqrt{\overline m}\log \overline m \cdot  D \cdot \log(1 /\Delta))$ rounds such that the flow value of $\vec{f'}$ is at least $\vec{f}$'s flow value. In Section \ref{sec:maxflow}, we implement maximum flow \cite{M16} in the \congest model, which takes $\widetilde{O}\left(\overline{m}^{3/7}U^{1/7}(\overline{n}^{o(1)}(\overline{n}^{1/2}+D)+\overline{n}^{1/2}D^{1/4})+\overline{m}^{1/2}D\right)$ rounds. In Section \ref{sec:mincostflow}, it takes  $\widetilde{O}\left(\overline{m}^{3/7}\overline{n}^{1/2}(n^{o(1)}+D^{1/4})+\overline{m}^{1/2}D\right)$ rounds to execute the min-cost flow \cite{CMSV17} that consists of Laplacian solver, flow rounding and single-source shortest path \cite{CM20} in the \congest model. In addition,  shortest paths with negative weights can be implemented in the same rounds since it utilizes min-cost flow to make edges non-negative and then compute the shortest paths \cite{CM20}.
\subsection{Flow Preconditioned Minor}
In this subsection, we define flow preconditioned minor.
Compared with the definition of $\rho$ congestion, this definition does not bound  the size of the pre-image of vertex mapping function $V^{G\rightarrow H}_{map}$.
Instead, we only allow a bounded number of vertices of $G$ that maps to more than one vertex of $H$.
\begin{definition}
\label{def:weak_minor}
Let $G = (V, E)$ be a minor of $H = (V_H, E_H)$ (as Definition~\ref{def:Minor}).
We say this minor is $(\rho, \alpha)$-flow-preconditioned if:
  \begin{enumerate}
    \item Each edge of $H$ appears as the image of the edge map
    $E_{map}^{G \rightarrow H}(\cdot)$, or in one of the trees
    connecting supervertices, $T^{G \rightarrow H}(v^{G})$ for some $v^{G}$,
    at most $\rho$ times.
\item There is a set $U \subset V$ with $|U| = \alpha$ such that the following two conditions hold:
\begin{enumerate}
\item 
$S^{G\rightarrow H}(v^G) = V_H$ 
and $T^{G\rightarrow H}(v^G)$ is a spanning tree of $H$ with depth at most the diameter of $H$  for each $v^G \in U$.
\item $|S^{G\rightarrow H}(v^G)| = 1$ for each $v^G \in V\setminus U$.
\end{enumerate}
\end{enumerate}

We say a $(\rho, \alpha)$-flow-preconditioned minor mapping is stored distributedly,
  or that $G$ is a $(\rho, \alpha)$-flow-preconditioned minor distributed over $H$ if it's stored by
  having all the images of the maps recording their sources (the same as Definition~\ref{def:Minor}).
\end{definition}

\begin{lemma}
  \label{lemma:Communication_edge_congested}
  Let $G = (V, E)$ be a graph with $n$ vertices and $m$ edges
  that $(\rho, \alpha)$-flow-preconditioned minor distributes into a communication network
  $\overline{G} = (\overline{V}, \overline{E})$ with $\overline{n}$ vertices,
  $\overline{m}$ edges, and diameter $D$.
  In the \congest model, the following operations can be performed
  using $O(t \alpha D)$
  rounds of communication on $\overline{G}$:
  \begin{enumerate}
    \item Each $V_{map}^{G \rightarrow \overline{G}}(v^{G})$ sends
    $O(t \log{n})$ bits of information to
    all vertices in $S^{G \rightarrow \overline{G}}(v^{G})$.
    \item Simultaneously aggregate the sum/minimum
    of $O(t \log{n})$ bits, from all vertices in
    $S^{G \rightarrow \overline{G}}(v^{G})$
    to
    $V_{map}^{G \rightarrow \overline{G}}(v^{G})$
    for all $v^{G} \in V(G)$.
  \end{enumerate}
\end{lemma}
\begin{proof}
Both operations can be achieved by running a BFS or a reverse BFS on $T^{G\rightarrow \overline G}(v^G)$ for each $v^G \in V$.
\end{proof}

We say a vector $\vec x \in \mathbb{R}^{V}$ on $G$ is distributed on $\overline G$ if for each $v^G$, all the vertices of 
$S^{G\rightarrow \overline G}_{map}(v^G)$ records $\vec x_{v^G}$.

We say a vector $\vec f \in \mathbb{R}^{E}$ defined on edges of $G$ is distributed to $\overline G$ if for each $e \in E$,
the two endpoints of  $E^{G
\rightarrow \overline G}_{map}(e)$ records $\vec f_e$.
Sometimes, we treat a vector $\vec f \in \mathbb{R}^{E}$  defined on edges of $G$ as a matrix, denoted as $\mM_{\vec f}$, such that $\mM_{\vec f, uv} = \vec f_{uv}$ if $(u, v)$ is an edge in $E$, otherwise $\mM_{\vec f, uv} = 0$.

\begin{lemma}
  \label{lemma:Communication_path_cycle}
  Let $\overline{G} = (\overline{V}, \overline{E})$ be a communication network $\overline{n}$ vertices,
  $\overline{m}$ edges, and diameter $D$.
  Let $\overline{\mathcal{P}}$ be a collection of paths/cycles of $\overline G$ such that every edge of $\overline G$ is used for at most $\rho$ times for some $\rho = O(\poly(\overline m))$, and for every two consecutive edges $(u, v), (v, w)$ of some path in $\overline{\mathcal{P}}$, $v$ knows that $(u, v), (v, w)$ are two consecutive edges of some path/cycle of  $\overline{\mathcal{P}}$.
  Then the following operations can be performed
  using $O(\rho \overline m^{1/2} \log \overline m + D)$ 
  rounds of communication on $\overline{G}$:
  \begin{enumerate}
	\item Every path/cycle of $\mathcal{P}$ is associated with a unique ID such that for each edge $e\in E$ in the path/cycle, 
	the two endpoints of $E^{G\rightarrow \overline G}_{map}(e)$ know the ID.
	    \item Let $\vec f$ be an edge vector of $G$. 
	    Compute the sum of $\vec f$ for each path/cycle of $\mathcal{P}$ and let the two endpoints of $E^{G\rightarrow \overline G}_{map}(e)$ know the result for each edge $e$ in the path/cycle.
  \end{enumerate}
\end{lemma}
\begin{proof}
We view $\overline {\mathcal{P}}$ as a graph such that every vertex and edge appears once by treating each appearance of the same vertex/edge as a new vertex/edge.
The resulted graph, denoted as $G$, is a graph of  $O(\rho \overline m)$ vertices and edges that corresponds to a set of edge disjoint paths.

We sample each vertex of $G$ with probability $\log \overline m / \overline m^{1/2}$. 
At most $O(\rho \overline m^{1/2} \log \overline m)$ vertices are sampled with high 
probability, and for each simple path of length $\overline m^{1/2}$ that corresponds to an induced subgraph of $G$, at least one vertex of the simple path is sampled. 
Each sampled vertices performs a BFS on $G$ until it reaches another sampled vertex or an endpoint of some path in $G$.
We aggregate the result of the BFS (the visit of sampled vertices or endpoints of some path in $G$) to an arbitrary vertex of $\overline G$ in $O(\rho \overline m^{1/2} \log \overline m + D)$ rounds, and the IDs for paths/cycles of length at least $\overline m^{1/2}$ can be assigned and broadcast to each vertex in these paths/cycles in $O(\rho \overline m^{1/2} \log \overline m + D)$ rounds.

Then, all the paths/cycles that do not contain any sampled vertex initiate a BFS from each vertex in the paths/cycles such that if two BFS collide, only the one initiated by small vertex ID is kept. 
Since every path/cycle that does not contain any sampled vertex is of length at most $O(\overline m^{1/2})$,
this step can be done in $O(\overline m^{1/2})$ rounds, and afterwards the IDs for these paths/cycles can be computed and broadcast to each vertex on these paths/cycles in $O(\overline m^{1/2})$ rounds.

Hence, the first operation can be done in $O(\rho \overline m^{1/2} \log \overline m + D)$ rounds.
The second operation can also be done in $O(\rho \overline m^{1/2} \log \overline m + D)$ rounds using sampled vertices in a way similar as the first operation.
\end{proof}

\subsection{Basic Operations on Flow Preconditioned Minor}

In the rest of this section, we will use the following basic operations on flow preconditioned minor in our algorithms: (Let $\overline{G} = (\overline V, \overline E)$ be a communication network,
and	$G = (V, E)$ be a graph that is a flow-preconditioned minor distributed to $\overline {G}$.)
\begin{enumerate}
\item Local Edge Vector Operation: 
Let $\vec f^{(1)}, \vec f^{(2)}, \dots, \vec f^{(t)}$ be $t$ edge vectors of $G$ that are distributed on $\overline G$.
Compute $\vec f$ that is an edge vector of $G$ distributed on $\overline G$ such that $\vec f_e$ is a function of $\vec f^{(1)}_e, \vec f^{(2)}_e, \dots, \vec f^{(t)}_e$ for each edge $e \in E$.
\item Local Vertex Vector Operation: Let $\vec x^{(1)}, \vec x^{(2)}, \dots, \vec x^{(t)}$ be $t$ vertex vectors of $G$ that are distributed on $\overline G$.
Compute $\vec x$ that is a vertex vector of $G$ distributed on $\overline G$ such that $\vec x_{v^G}$ is a function of $\vec x^{(1)}_{v^G}, \vec x^{(2)}_{v^G}, \dots, \vec x^{(t)}_{v^G}$ for each vertex $v^G \in V$.
\item Norm Operation: For a vertex vector $\vec x \in  \mathbb{R}^{V}$ or an edge vector $\vec f \in \mathbb{R}^E$ on $G$ distributed to $\overline G$ and a $p > 1$, 
compute $p$-norm of $\vec x$ or $\vec f$ and broadcast the result to each vertex of $\overline G$. 
\item Coordinate Selection Operation: For an edge vector $\vec f \in  \mathbb{R}^{E}$ of $G$ distributed
on $\overline G$ and an integer $k$, identify top $k$ coordinates of $\vec f$ with largest absolute value. 
\item Matrix Vector Multiplication Operation: For an edge vector $\vec f \in  \mathbb{R}^{E}$  of $G$ and a vertex vector $\vec x \in  \mathbb{R}^{V}$ on $G$ both distributed to $\overline G$, compute 
$\mM_{\vec f} \vec x$ that is distributed to $\overline G$. 

\end{enumerate}

We prove the following lemma to bound the number of rounds to perform each basic operation.

\begin{lemma}\label{lem:basic_operation}
Let $\overline{G} = (\overline V, \overline E)$ be a communication network with $\overline{n}$ vertices and $\overline{m}$ edges,
and	$G = (V, E)$ be a graph that is $(\rho, \alpha)$-flow-preconditioned minor distributed to $\overline {G}$. 
Then, we have

\begin{enumerate}
\item an algorithm to perform a local vertex vector operation or a local edge vector operation in $O(1)$ rounds;
\item an algorithm to perform a norm operation in $O(D)$ rounds;
\item an algorithm to perform a coordinate selection operation for $\vec f$ in $O(D \cdot \poly(\log \beta))$ rounds,
where $\beta = \frac{\max_{e\in E} |\vec f_e|}{\min_{e, e'\in E: \vec f_e \neq \vec f_{e'}} |\vec f_e - \vec f_{e'}|}$;
\item an algorithm to perform a matrix vector multiplication operation in 
$O(\rho + \alpha D)$ rounds.
\end{enumerate}
\end{lemma}
\begin{proof}
The local vertex vector operation or local edge vector operation can be computed locally by each vertex of $\overline G$.

The norm operation can be computed by a BFS on $\overline G$ such that every $G$ by aggregating the sum of $\vec x_{v^G}$ for each $v^G \in V$.

The coordinate selection operation can be computed by binary searching $k$-th largest absolute value of $\vec f_e$ among all the $e \in E$ and counting the number of coordinates with absolute value greater than or equal to the value binary searched.

The matrix vector multiplication operation can be implemented by computing $\vec f_{uv} \vec x_v$ at one endpoint of $E^{G \rightarrow \overline G}_{map}(uv)$ in $S^{G \rightarrow \overline G}_{map}(u)$ and taking the sum of $\sum_{v: uv \in E} \vec f_{uv} \vec x_v$ for each $u$ by $T^{G\rightarrow \overline G}_{map}(u)$. The number of rounds required is by Lemma~\ref{lemma:Communication_edge_congested}.
\end{proof}

\subsection{Flow Rounding}
\label{sec:flowrounding}

We simulate the flow rounding algorithm by Cohen~\cite{cohen1995approximate} in the \congest~model as a subroutine for maximum flow and min-cost flow.
The algorithm by Cohen~\cite{cohen1995approximate} is summarized as Algorith~\ref{algo:rounding}.

\begin{lemma}[Proposition 5.3 of~\cite{cohen1995approximate}]
	Let $G = (V, E)$ be a graph with $n$ vertices and $m$ edges,
	 $f:  E \rightarrow \mathbb{R}^{\geq 0}$ be a $s$-$t$ flow function,
	and $\Delta$ be a real value such that $1/\Delta$ is a power of $2$ and $f(e)$ is an integral multiplication of $\Delta$ for every $e \in E$.
	Then
	\begin{enumerate}
		\item Algorithm \FlowRounding~rounds $f$ on edge $e\in E$ to $\lfloor f(e) \rfloor$ or  $\lceil f(e) \rceil$ such that the resulted flow has the total flow value not less than $f$.
		 \item If the total flow value of $f$ is integral and there is an integral cost function  $c:  E \rightarrow \mathbb{Z}^{\geq 0}$, then
		 Algorithm \FlowRounding~rounds $f$ on edge $e\in E$ to $\lfloor f(e) \rfloor$ or  $\lceil f(e) \rceil$ such that the resulted flow has the total flow value not less than $f$, and the total cost not more than $f$.
	\end{enumerate}
\end{lemma}

\begin{algorithm}[ht]
\caption{Flow Rounding \label{algo:rounding}}
\SetKwProg{Proc}{procedure}{}{}
\Proc{$\FlowRounding(G, s, t, f, c, \Delta)$}{
	\If {the total flow of $f$ is not integral} {
		Add an edge from $t$ to $s$ with flow value the same as total flow.\\
	}
	\While{$\Delta < 1$} {\label{ln:FlowRoundingWhile}
		$ E' \assign \{(u, v) \in  E: f(u, v)/ \Delta \text{ is } odd\}$\\
		Find an Eulerian partition of $ E'$ (ignoring the directions of the edges)\\
		\For{every cycle of the Eulerian partition of $ E'$} {

			\If {cycle contains the edge $(t, s)$} {
			Traverse the cycle such that edge $(t, s)$ is a forward edge.
			}
			\ElseIf {cost function $c$ exists} {
				Traverse the cycle such that the sum of costs on forward edges is no more than the sum of costs on backward edges.
			}
			\Else{
				Traverse the cycle arbitrarily.
			}
		}
		\For{every edge $(u, v)$ $ E'$} {
			\If {$(u, v)$ is a forward edge w.r.t the traversal of the path containing $(u, v)$} {
				$f(u, v) \assign f(u, v) + \Delta$\\
			}
			\Else {$f(u, v) \assign f(u, v) - \Delta$\\}

		}
		$\Delta \assign 2\Delta$\\
	}
	\Return $f$.
}
\end{algorithm}

In this section, we present two distributed simulations of Cohen's algorithm, one for maximum flow rounding, and another one for min-cost flow rounding.
The underlying reason of  two algorithms is that 
the strategy of choosing path directions (Line 8-11 of Algorithm~\ref{algo:rounding}) are different: the former one always choose direction that does not decrease the flow value, and later one always chooses the direction that does not increase total cost.

\begin{lemma}\label{lem:flow_rounding}
	Let $\overline{G} = (\overline V, \overline E)$ be a communication network with $\overline{n}$ vertices, $\overline{m}$ edges and diameter $D$,
	$G = (V, E)$ be a flow network containing two vertices $s$ and $t$ such that $G$ is a $(\rho, \alpha)$-flow-preconditioned minor distributed to $\overline {G}$,
	$\Delta$ be a real value such that $1/\Delta$ is an integer that is a power of $2$, and
	(for two vertices $s, t \in \overline V$) $\vec f:$ be a $s$-$t$ flow function on $G$ such that $f(e)$ is an integral multiplication of $\Delta$ for every $e \in \overline E$, and $\vec f$ is distributed to $\overline G$.
	Then
	\begin{enumerate}
	\item
	 	There is a distributed algorithm which runs in $O((\rho \sqrt{\overline m}(\log \overline m)^2 + D) \cdot \log (1 / \Delta))$ rounds
		to compute an integer $s$-$t$ flow function $f' : E \rightarrow \mathbb{Z}^+$ such that the flow value of $f'$at least that of $f$ and $f'(e) \in \{\lfloor f(e) \rfloor, \lceil f(e) \rceil\}$ for every $e \in \overline E$.
	\item If the total flow value of $f$ is integral and there is an integral cost function  $c$, then
		 there is a distributed algorithm
		$O((\rho \sqrt{\overline m}(\log \overline m)^2 + D) \cdot \log (1 / \Delta))$ rounds
		to compute an integer $s$-$t$ flow function $f' : E \rightarrow \mathbb{Z}^+$
		 such that the resulted flow has the total flow value not less than $f$, and the total cost not more than $f$.

	\end{enumerate}
	\end{lemma}

\begin{proof}
We first extend $G$ such that $S^{G \rightarrow \overline G}(s) = S^{G \rightarrow \overline G}(t) = \overline V$.
The resulted $G$ is a $(\rho +2, \alpha + 2)$-flow-preconditioned minor distributed on $\overline G$.
Then in $O(1)$ rounds, we can add edge $(t, s)$ to $G$ by letting $E^{G\rightarrow \overline G}_{map}(t, s) = \overline v \overline v$ for an arbitrary vertex $\overline v \in \overline V$. The resulted $G$ is still a $(\rho +2, \rho + 2)$-flow-preconditioned minor distributed on $\overline G$.

Throughout the algorithm, we view the flow function as a vector on $E$, denoted as $\vec f$, that is distribured to $\overline G$.

For a $\vec f$ and a fixed $\Delta$, $E'$ which is an edge set of $G$ can be identified in $O(1)$ rounds.
Now we show that an Eulerian partition of $E'$ can be determined in $O((\alpha +2) D)$ rounds.
By the first and second condition of the lemma, for any vertex $v^G \in V$, the number of edges in $E'$ incident to $v^G$ is always even.
Hence,
to construct an Eulerian partition of $E'$, we only need to pair all the incident edges in $E'$ for each vertex of~$V$.
This pairing process can be done in $O(1)$ rounds for all the vertices $v^G$ such that $|V^{G\rightarrow \overline G}(v^G)| = 1$.
For each vertex $v^G \in V$ such that $V^{G\rightarrow \overline G}(v^G) = \overline V$, 
we simulate the following algorithm in $O(D)$ rounds such that the pairing process are simulated:
\begin{itemize}
\item 
Run a reverse BFS on $T^{G\rightarrow \overline G}_{map}(v^G)$ such that for each $v^{\overline G} \in \overline V$:
Let $E'_{v^{\overline G}}$ be the union of the edges sent to $v^{\overline G}$ from all the children of $v^{\overline G}$ in $T^{G\rightarrow \overline G}(v^G)$ and all the edges in $E'$ incident to $v^G$ whose images by $E^{G\rightarrow \overline G}_{map}$ are edges incident to $v^{\overline G}$.
If $|E'_{v^{\overline G}}|$ is even, then pair all the edges in $E'_{v^{\overline G}}$ arbitrarily, otherwise, send one edge of $E'_{v^{\overline G}}$ to the parent of $v^{\overline G}$ in $T^{G\rightarrow \overline G}(v^G)$, and pair the remaining edges of $E'_{v^{\overline G}}$ arbitrarily.
\end{itemize}
In addition, since for every vertex $v^{\overline G}$ of $T^{G\rightarrow \overline G}(v^G)$, every child of $v^{\overline G}$ with respect to $T^{G\rightarrow \overline G}(v^G)$ sends the information of at most one edge to $v^{\overline G}$,
the Eulerian partition of $E'$ corresponds to a union of cycles of $\overline G$ such that every edge of $\overline G$ is used at most $\rho + 2$ times.

Lemma~\ref{lemma:Communication_path_cycle} gives that
for each cycle in the  Eulerian partition,
aggregating the required information takes
$\widetilde O(\overline m^{1/2} \log \overline m + D)$ rounds.
And since all the cycles are edge disjoint for $G$, 
the traversal of each cycle of the  Eulerian partition can be determined in $O(\rho \sqrt{\overline m}\log \overline m + D)$ rounds, and the direction of the traverse of each cycle can be broadcasted to each edge of the cycle in $O(\rho \sqrt{\overline m}\log \overline m + D)$ rounds.

Hence, simulating one iteration of the while loop on Line~\ref{ln:FlowRoundingWhile}
of Algorithm \FlowRounding\ takes $O(\rho \sqrt{\overline m}\log \overline m + D)$ rounds.
So the overall number of rounds needed for Algorithm \FlowRounding\ is
$O((\rho \sqrt{\overline m}\log \overline m + D) \cdot \log (1 / \Delta))$.
\end{proof}


\subsection{Maximum Flow}
\label{sec:maxflow}
In this subsection, we present a distributed exact maximum flow algorithm for flow network with integral capacity in $	\widetilde{O}\left(\overline{m}^{3/7}U^{1/7}\overline{n}^{o(1)}(\overline{n}^{1/2}D^{1/4}+D)+\overline{m}^{1/2}\right)$ rounds in the \textsf{CONGEST} model, where $U$ is upper bound of the capacities among all the edges.
Based on the distributed Laplacian solver, our algorithm simulate  the sequential exact maximum flow algorithm by Madry~\cite{M16}.
Madry's sequential algorithm is briefly summarized in Algorithm~\ref{alg:madry_renew}, and the details are given in Section~\ref{sec:maxflow_algo}. 

\begin{theorem}[\cite{GU15, CM20}]\label{thm:directed_reachbility}
Let $G$ be a (undirected or directed) graph with $n$ vertices, $m$ edges and undirected diameter $D$,  $s$ be a vertex of $G$, and $\vec w$ be an edge vector such that for any edge $(u, v)$ of $G$, vertices $u$ and $v$ know $w_{u, v}$.  
Assume in every round, two vertices can send $O(\log n)$ bit information to each other if there is an edge between them in $G$ no matter the direction of the edge.
Then there is a distributed SSSP algorithm that in $\widetilde O(n^{1/2} D^{1/4} + D)$ rounds computes the distances with respect to $\vec w$ from $s$ to all of its reachable vertices as well as an implicit shortest path tree rooted at $ s $ such that every vertex knows its parent in the shortest path tree.
\end{theorem}

In our distributed maximum flow and min-cost flow algorithm, 
the graph which we run single source shortest path (SSSP) algorithm on is different to the communication network, 
because the algorithms we want to simulate~\cite{M16, CMSV17} add additional vertices and edges to the graph. 
Hence, we show that this SSSP algorithm can be simulated efficiently if the graph is a flow preconditioned minor distributed to the communication network. 
\begin{corollary}\label{cor:shortest_path}
Let $\overline{G} = (\overline V, \overline E)$ be a communication network with $\overline{n}$ vertices and $\overline{m}$ edges, 
and	$G = (V, E)$ be a (undirected or directed) graph with $n$ vertices, $m$ edges and diameter $D$ that is $(\rho, \alpha)$-flow-preconditioned minor distributed to $\overline {G}$. 
Then there is a distributed SSSP algorithm that, for any given source vertex $s^G \in V$, performs $\widetilde O(\rho(\overline n^{1/2} D^{1/4} + D)\cdot \alpha^2)$ rounds. 
\end{corollary}

\begin{algorithm}
\KwIn{directed graph $G_0=(V, E_0,\vec{u})$ with each $e\in E_0$ having two non-negative integer capacities $u_{e}^{-}$ and $u_{e}^{+}$; $|V|={n}$ and $|E_0|={m}$; source $s$ and sink $t$; the largest integer capacity $U$; target flow value $F\ge 0$;}
\caption{\MaxFlow ($G_0$, $s$, $t$, $U$, $F$)}\label{alg:madry_renew}
Add $m$ undirected edges $(t, s)$ with forward and backward capacities $2U$ to $G_0$\;\label{line:mf_1}
\For{each $e=(u,v)\in E_0$}{
Replace $e$ by three undirected edges $(u,v)$, $(s,v)$ and $(u,t)$ whose capacities are $u_{e}$\;}\label{line:mf_3}
Let the new graph be $G=(V,E)$\;
Initialize the flow vector $\vec{f}\assign\vec{0}$ and dual vector $\vec{y}\assign\vec{0}$\;\label{line:mf_5}
Update $\vec{f}$ and $\vec{y}$ by solving two Laplacian linear systems on $G$ and a constant number of local vertex/edge vector operations\;\label{line:mf_6}
Compute the congestion vector $\vec{\rho}$ by a constant number of local edge vector operations\;\label{line:mf_7}
\Repeat{$\widetilde{O}(m^{3/7}U^{1/7})$ times}{
\eIf{$\|\vec{\rho}\|_3$ is at most one computed threshold}{
Update $\vec{f}$ and $\vec{y}$ by solving two Laplacian linear systems on $G$ and a constant number of local vertex/edge vector operations\;\label{line:mf_10}
Update $\vec{\rho}$ by a constant number of local vertex/edge vector operation\;\label{line:mf_11}}{
Determine the set $S^*$ that contains the $m^{4\eta}$ edges with the largest $|\rho_e|$ by a coordinate selection operation\;\label{line:mf_13}
Update graph $G$ via replacing each edge in $S^*$ by a path and setting some quantities\;\label{line:mf_14}
}}
\While{there is an augmenting path from $s$ to $t$ w.r.t. $\vec{f}$ for $G$}{
Augment an augmenting path for $\vec{f}$ using the shortest path from $s$ to $t$ in the residual graph\;\label{line:mf_17}
}
\end{algorithm}

\begin{proof}
We assume that excluding the $\alpha$ vertices of $G$ that are mapped to all the vertices of $\overline G$,
at most one vertex is mapped to any vertex of $\overline V$. 
This is without loss of generality, because if multiple vertices are mapped to the same vertex of $\overline G$,
then any communication between these vertices are free.
In the following, we will call vertices of $G$ that are mapped to all the vertices of $\overline G$ \emph{simulated} vertices.

We first explain how the SSSP algorithm of~\cite{CM20} can be modified to work as an $s$-$t$ shortest path algorithm in a setting where only the source (start) vertex~$ s $ and the sink (target) vertex~$ t $ (and their incident edges) are simulated vertices.
In particular,  we argue that the algorithm can be simulated in $\widetilde O(\rho(n^{1/2} D^{1/4} + D))$ rounds (where the multiplicative $\rho$ factor simply comes from the fact that every edge of $\overline G$ is used at most $\rho$ times for edges in $G$).

The algorithm of~\cite{CM20} consists of eight steps.
In Steps 1 and 2, vertices sample themselves with certain probabilities.
These steps require no communication and therefore can also be carried out in our setting in which source and sink are just simulated.
We slightly modify Step 1 to ensure that the vertex~$ t $ is never sampled.
This does not affect the correctness of the algorithm if we are only interested in computing the shortest path from $ s $ to~$ t $ as shortest paths are simple and thus $ t $ will never be an inner vertex on this shortest path.

In Steps 4 and 7, certain auxiliary graphs are created implicitly in the sense that each vertex only knows its incident edges in the auxiliary graph and their respective edge weights.
Therefore these steps also require no communication and therefore can also be carried out in our setting in which source and sink are just simulated.

To implement the rest of the algorithm we will rely on the following observation: whenever a step of the algorithm is performed solely by broadcasting or aggregating values via a global BFS tree of the network, then this step immediately can be carried out in our setting with the two simulated vertices as well.
This is the case in Steps 5 and 6 of the algorithm.

In Step~8, a certain number of iterations of the Bellman-Ford algorithm is performed on a graph that in addition to the edges of the input graph contains edges from $ s $ to certain other vertices.
Similar to~\cite{CM20}, we carry out the first iteration of the Bellman-Ford algorithm -- in which the neighbors of~$ s $ set their tentative distance to the weight of the edge from~$ s $ -- by a global broadcast in $ O(D) $ rounds.
In~\cite{CM20}, the remaining iterations of Bellman-Ford are carried out in the standard way where vertices directly communicate with their neighbors.
For our modification of~\cite{CM20} we do the same, but ignore the vertex~$ t $ for these iterations.
In the end, we explicitly need to ensure that the simulated vertex~$ t $ also gets to know its distance from~$ s $.
We achieve this by additionally performing one iteration of the Bellman-Ford iteration in which only the incoming edges of~$ t $ (and the corresponding neighbors of~$ t $) are considered.
This can be carried out in $ O(D) $ rounds by broadcasting.
This works because the incoming neighbors of~$ t $ (which are part of the communication network and are not just simulated by it) already know their distance from~$ s $ at this stage due to the previous iterations of Bellman-Ford.

This leaves only Step 3 of the algorithm.
In Step 3, Lemma~2.4 of~\cite{FN18} is applied to compute approximate distances from each vertex of a set~$ S $ (where~$ S $ includes the vertex~$ s $, but not the vertex~$ t $.).
Essentially this Lemma amounts to running a ``weighted'' version of the breadth-first-search algorithm for each vertex of~$ S $ (which is repeated $ O (\log (nW)) $ times with a certain weight rounding applied to the edges in each iteration).
The start times of these BFS algorithms are chosen with random delay to guarantee that the congestion at each vertex is low.
For the BFS starting at vertex~$ s $, the first iteration can be carried out by broadcasting the random delay of~$ s $.
The neighbors of~$ s $ (knowing the weight of the edge from~$ s $) then know when the ``weighted'' BFS of~$ s $ reaches them and can continue with it at the respective time.
This gives an additional additive term of $ \tilde O(D) $ in the running time, which does not affect the asymptotic bounds stated in Lemma 2.4 of~\cite{FN18}.
This concludes our discussion of the $s$-$t$ shortest path algorithm.

Now observe that with the same approach we can obtain an SSSP algorithm in a setting where the source vertex~$ s $ is the only vertex simulated by the network: we simply need to remove the special handling we had for vertex~$ t $ in our approach above.

Note that these two algorithmic primitives are sufficient to compute SSSP in a setting where there are $ \alpha $ simulated vertices:
Let $ S $ denote the set of vertices consisting of $ s^G $ and the simulated vertices.
First, perform an SSSP computation from each vertex $ s \in S $ ignoring the other vertices of $ S $ (i.e., perform the SSSP computation in the graph $ G \setminus S \cup \{ s \} $).
Then, perform an $s$-$t$ shortest path computation for each pair of vertices $ s, t \in S $, ignoring the other vertices of $ S $ (i.e., perform the $s$-$t$ shortest path computation in the graph $ G \setminus S \cup \{ s, t \} $).
Now each vertex~$ v $ can reconstruct its distance from~$ s^G $ by the information it stored so far as the shortest path from $ s^G $ to $ v $ can be subdivided into subpaths between vertices of $ S $ containing no other vertices of $ S $.
Overall, we perform $ O (\alpha) $ SSSP computations with at most one simulated vertex and $ O (\alpha^2) $ $s$-$t$ shortest path computations with at most two simulated vertices.
Finally, note that as soon as each vertex $ v $ knows its distance from $ s^G $ an implicit shortest path tree (in which each vertex knows its parent in the tree) can be reconstructed by performing a single iteration of the Bellman-Ford algorithm.
For the $ \alpha $ simulated vertices, we perform this final step in $ O (\alpha D) $ rounds by broadcasting via a global BFS tree.
\end{proof}

\begin{theorem}
\label{thm:max-flow}
	Let $\overline{G} = (\overline V, \overline E)$ be a communication network with $\overline{n}$ vertices, $\overline{m}$ edges, and diameter $D$,
	$G_0$ be a  graph and $c$ be an integral capacity function for each edge of $G_0$ with maximum capacity $U$ satisfying one of the following two conditions:
	\begin{enumerate}
		\item $G_0$ is the same as $\overline G$, and for each edge $(u, v) \in \overline E$, $u$ and $v$ know the capacity of edge $(u, v)$. 
		\item $G_0$ is a directed graph obtained by associating each edge of $\overline G$ a direction such that for each edge $(u, v) \in \overline E$,
				$u$ and $v$ know the direction of edge $(u, v)$ and its capacity. 
	\end{enumerate}
	Then there is a distributed algorithm to compute exact $s$-$t$ maximum flow for two vertices $s$ and $t$ of $G_0$ in 
	\[\widetilde{O}\left(\overline{m}^{3/7}U^{1/7}\overline{n}^{o(1)}(\overline{n}^{1/2}D^{1/4}+D)+\overline{m}^{1/2}\right)\] rounds in the \textsf{CONGEST} model.
\end{theorem}

\begin{proof}
We simulate Algorithm~\ref{alg:madry_renew}. 
By \cite{M16}, the accuracy required throughout the algorithm is $1 / \poly(\overline m)$.
Without loss of generality, we assume all the values throughout multiplied by $2^\gamma$ are integers for some $\gamma = O(\log \overline m)$. 
Throughout the algorithm, we set $S^{G_0 \rightarrow \overline G}(s) = S^{G_0 \rightarrow \overline G}(t) = \overline V$, and for each vertex $v \in \overline V \setminus \{s, t\}$, $S^{G_0 \rightarrow \overline G}(v) = v$. 

In line \ref{line:mf_1}, we need to add $m$ parallel $(t, s)$ edges each with capacity $U$. 
This step can be simulated in $O(1)$ rounds by specifying an arbitrary vertex in $v^{\overline G} \in \overline V$ such that $m$ parallel $(t, s)$ edges are mapped to selfloops of $v^{\overline G}$. 

In line \ref{line:mf_3}, every edge $(u, v)$ with capacity $\alpha$ of $G_0$ is replace by three edges $(u, v)$, $(s, u)$ and $(v, t)$ with capacity $\alpha$.
Let $G$ denote the graph after line \ref{line:mf_3}. 
We always make sure that $E^{G\rightarrow \overline G}_{map}(s, u) = (u, u)$
and $E^{G\rightarrow \overline G}_{map}(v, t) = (v, v)$. 
Hence, $G$ is a $(3, 2)$-flow-preconditioned minor distributed to $\overline G$.  

In line \ref{line:mf_14}, 
if we replace an edge $(u, v)$ of $G$ by a path $(u, v_1, v_2, \dots, v_\ell, v)$,
then we consider the following two cases:
\begin{enumerate}
\item If $E^{G\rightarrow \overline G}_{map}(u, v)$ is a selfloop of some vertex $x$ in $\overline G$, then 
add all the new vertices $v_1, v_2, \dots, v_\ell$ such that $V^{G\rightarrow \overline G}_{map}(v_i) = x$,
and all the new edges are also selfloops on $x$. 
\item If $E^{G\rightarrow \overline G}_{map}(u, v)$ corresponds to an edge of $\overline G$, then 
add all the new vertices $v_1, v_2, \dots, v_\ell$ such that  $S^{G\rightarrow \overline G}(v_i) = \{V^{G\rightarrow \overline G}_{map}(u)\}$,
and add edges such that 
$E^{G \rightarrow \overline G}_{map}(v_\ell, v) = E^{G\rightarrow \overline G}_{map}(u, v)$ and the remaining edges to be selfloops on $V^{G\rightarrow \overline G}_{map}(u)$.
\end{enumerate}
Hence, we maintain the invariant that $G$ is $(3, 2)$-flow-preconditioned minor distributed to $\overline G$.

Note that the execution of each line takes a constant number of basic vector operations in Lemma~\ref{lem:basic_operation} 
or solves a constant number of Laplacian systems on graph $G$ (excepting replacing an edge by a path in line \ref{line:mf_14}, which can be done in $O(1)$ rounds). 
By Lemma~\ref{lem:basic_operation}, each basic vector operation can be simulated in $O(D \log \overline m)$ rounds.

To solve a Laplacian system,
we first eliminate all the vertices that are added in line \ref{line:mf_14}. 
Since all these vertices are of degree 2, 
this elimination can be done locally in each vertex of $\overline G$, and the resulted graph is $O(1)$-minor distributed to $\overline G$.
By Theorem~\ref{thm:Main}, the Laplacian system can be solved in $\overline n^{o(1)}(\overline n^{1/2} + D)$ rounds.
Then we 
obtain the solution of Laplacian system with respect to $G$ by adding the eliminated vertices back locally.

Since the repeat part takes $O(\overline m^{3/7}U^{1/7})$ iterations, 
the total number of rounds required to simulate all progress steps is 
$\widetilde O(\overline m^{3/7+o(1)}U^{1/7}(\overline n^{1/2} + D))$. 

By Lemma~\ref{lem:flow_rounding}, the flow rounding takes $O(\log m\cdot (\overline m^{1/2} \log^2 \overline m + D))$ rounds.
After the flow rounding, the difference between the flow value and maximum flow value is at most $O(\overline m^{3/7}U^{1/7})$. 
By Corollary~\ref{cor:shortest_path},  each iteration of finding an augmenting path takes $\widetilde O(D + \overline n^{1/2} D^{1/4})$ rounds. 
Hence,
the additional augmenting step takes $\widetilde O(\overline m^{3/7}U^{1/7} (\overline n^{1/2} D^{1/4} + D))$ rounds.
\end{proof}

\subsection{Unit Capacity Minimum Cost Flow}
\label{sec:mincostflow}
In this subsection, we present a distributed minimum cost unit capacity flow algorithm with integral cost in $\widetilde{O}\left(\overline{m}^{3/7+o(1)}(\overline n^{1/2}D^{1/4}+D)\right)$ rounds in the \congest~model.
Based on the distributed Laplacian solver, our algorithm simulate  the algorithm by Cohen et al.~\cite{CMSV17}.
The sequential algorithm is briefly summarized in Algorithm~\ref{algo:min-cost-flow_renew}, and the details of the algorithm are given in Section~\ref{sec:min-cost-flow-algo}.

\begin{algorithm}[t]
\KwIn{directed graph ${G_0}=({V_0},{E_0},\vec{c}_0)$ with each edge having unit capacity and cost $\vec{c}_0$; $|{V_0}|={n}$ and $|{E_0}|={m}$; integral demand vector $\vec{\sigma}$; the absolute maximum cost $W$;}
\caption{\MinCostFlow ($G$, $\vec{\sigma}$, $W$)}
\label{algo:min-cost-flow_renew}
Create a new vertex $v_{aux}$ with $\sigma(v_{aux})=0$ and add parallel edges $(v, v_{aux})$ or $(v_{aux}, v)$ according to $\rho(v)$ and degree of $v$ in $G_0$, and denote the resulted graph as $G_1=(V_1, E_1, \vec{c}_1)$\;\label{line:mcf_1}
Initialize the bipartite graph $G=(P\cup Q, E, \vec{c})$ with $P\leftarrow V_1$ and $Q\leftarrow\{e_{uv}\mid(u,v)\in E_1\}$ where $e_{uv}$ is a vertex corresponding to edge $(u,v)\in E_1$\;\label{line:mcf_2}
Add a new vertex $v_0$ and undirected edges $(v_0, v)$ for every $v\in P$ to $G$\;\label{line:mcf_3}
Initialize the resistance vector $\vec{r}$, flow vector $\vec{f}$, dual vector $\vec{s}$, measure vector $\vec{\nu}$ and congestion vector $\vec{\rho}$ by a constant number of local vector operations\;\label{line:mcf_4}
\For{$i=1$ \KwTo $1200\sqrt{3}m^{2/7}\log^{4/3}{W}$\label{line:mcf_5}
}{
Reset the resistances of the auxiliary edges $(v_0, v)$ for each $v\in P$ by a constant number of  local vector operations\;\label{line:mcf_6}
\For{$j=1$ \KwTo $m^{1/7}$}{
\While{$\left\|\vec{\rho}\right\|_{\vec{\nu},3}>400\sqrt{3}m^{3/7}\log^{1/3}{W}$}{
Increase the energy via resetting the resistance $r_e$ and measure $\nu_e$ for $e\in E$ by a constant number of local vector operations\;\label{line:mcf_9}
}
Update the flow vector $\vec{f}$ and dual vector $\vec{s}$ by solving two Laplacian linear systems in $\mL(G)$ and a constant number of local vector operation\;\label{line:mcf_10}
}}
Round the solution to be integral by calling \FlowRounding and obtain $\vec{M}$\;\label{line:mcf_11}
\Repeat{$\widetilde{O}(m^{3/7})$ times\label{line:mcf_16}}{
\label{line:mcf_12}
Construct the directed graph $\overrightarrow{G}_M$ of $G$ w.r.t. $\vec{M}$ by a constant number of local edge vector operations\;\label{line:mcf_13}
Compute the shortest path $\pi$ in $\overrightarrow{G}_M$\;\label{line:mcf_14}
Augment $\vec{M}$ using the augmenting path $\pi$\;\label{line:mcf_15}
}
\Return $\vec{M}$\;
\end{algorithm}

\begin{theorem}
\label{thm:min-cost-flow}
	Let $\overline{G} = (\overline V, \overline E)$ be a communication network with $\overline{n}$ vertices and $\overline{m}$ edges,
	$G_0$ be a directed unweighted graph also defined on $\overline V$ such that each edge $(u, v)$ of $G_0$ is an edge of $\overline E$ if the direction is ignored, and 
				$u$ and $v$ in the communication network know the direction of edge $(u, v)$ and its cost. 
	Then, given $G_0$ and a demand vector $\vec \sigma$, there is a distributed algorithm to compute the minimum cost  flow for graph $G_0$ with respect to $\vec \sigma$ in 
	\[\widetilde O(\overline m^{3/7+o(1)}(\overline n^{1/2}D^{1/4} + D)\poly(\log W) )\] rounds in the \textsf{CONGEST} model.
\end{theorem}
\begin{proof}
We simulate Algorithm~\ref{algo:min-cost-flow_renew} and the accuracy required throughout the algorithm is $1 / \poly(\overline m)$.
Without loss of generality, we assume all the values throughout the algorithm multiplied by $2^\gamma$ are integers for some $\gamma = O(\log \overline m)$. 

To build graph $G_1$ in line \ref{line:mcf_1}, we add vertex $v_{aux}$ to the graph such that $S^{G_1 \rightarrow \overline G}(v_{aux}) = \overline V$, and for all the vertices $v \in \overline V$,
$S^{G_1 \rightarrow \overline G}(v) = \{v\}$. 
For each edge $(u, v)$ of $G_1$ with $u, v \in \overline V$,
we let $E^{G_1 \rightarrow \overline G}_{map}(u, v) = (u, v)$.
For edge $(v, v_{aux})$ or  $(v_{aux}, v)$, we let 
$E^{G_1 \rightarrow \overline G}_{map}(v, v_{aux})$ to be a selfloop on $v$.
The construction of $G_1$ can be done locally, and the resulted graph is a $(2, 1)$-flow-preconditioned minor distributed to $\overline G$.

To construct $G$, for each vertex $e_{uv} \in Q$,
we map $e_{uv}$ to be one of $u$ and $v$ of $\overline G$ arbirarily. 
For vertices of $P$, the mapping is the same as that of $G_1$.
For each edge $(u, e_{uv})$ of $G$, if $V^{G \rightarrow \overline G}_{map}(u) = V^{G \rightarrow \overline G}_{map}(e_{uv})$, then we map edge 
$(u, e_{uv})$ to a selfloop on $V^{G \rightarrow \overline G}_{map}(u)$, 
otherwise, $(V^{G \rightarrow \overline G}_{map}(u), V^{G \rightarrow \overline G}_{map}(e_{uv}))$ is an edge of $\overline G$, and we set 
$E^{G \rightarrow \overline G}_{map}(u, e_{uv})$ to be the edge $(V^{G \rightarrow \overline G}_{map}(u), V^{G \rightarrow \overline G}_{map}(e_{uv}))$.
Hence, 
the first two lines can be simulated in $O(1)$ rounds, and the resulting graph 
$G$ is a $(2, 1)$-flow-preconditioned minor distributed to $\overline G$.

 To simulate Line \ref{line:mcf_3}, 
 we add an additional vertex $v_0$ to $G $ such that $S^{G\rightarrow \overline G}(v_0) = \overline V$,
 and add edges $(v_0, v)$ for each $v \in P$ by setting $E^{G\rightarrow \overline G}_{map}(v_0, v)$ be a selfloop on $v$ if $v \in \overline V$, and 
 $E^{G\rightarrow \overline G}_{map}(v_0, v)$ be a selfloop on an arbitrary vertex of $\overline V$ if $v = v_{aux}$.
 The resulting graph $G$ is a $(3, 2)$-flow-preconditioned minor distributed to $\overline G$. 

For lines \ref{line:mcf_5}-\ref{line:mcf_10} of  Algorithm~\ref{algo:min-cost-flow_renew}, the execution of each line takes a constant number of basic vector operations in Lemma~\ref{lem:basic_operation} 
or solves a Laplacian system on graph $G$. 
To solve a Laplacian system, since all the vertices of $Q$ have two incident edges,
we eliminate vertices of $Q$ in $O(1)$ rounds,
and the resulting graph is $O(1)$-minor distributed to $\overline G$. 
By Theorem~\ref{thm:Main}, every Laplacian system can be solved in $\overline n^{o(1)}(\overline n^{1/2} + D)$ rounds.
Based on the parameter setting, lines \ref{line:mcf_5}-\ref{line:mcf_10} can be simulated in $\widetilde O(\overline m^{3/7+o(1)}(\overline n^{1/2} + D) \poly (\log W))$ rounds.
 
Before the flow rounding, we add $s$ and $t$ such that $S^{G\rightarrow\overline G}(s) = S^{G\rightarrow\overline G}(t) = \overline V$, and add edges by adding selfloops on vertices of $\overline V$. 
 The resulting graph is a $(5, 4)$-flow-preconditioned minor distributed to $\overline G$.
By Lemma~\ref{lem:flow_rounding}, line \ref{line:mcf_11} can be simulated in $\widetilde O(\overline m^{1/2} + D)$ rounds.

To simulate lines \ref{line:mcf_12}-\ref{line:mcf_16},
all the operations except finding shortest path can be simulated in a way similar to that of lines \ref{line:mcf_5}-\ref{line:mcf_10}. 
To find the shortest path from $P\cap F_M$ to $Q\cap F_M$, we add an additional vertex $v'$, and edges $(v', v)$ for each $v \in P\cap F_m$ with  weight zero to $G$. The resulted graph is a $(4, 3)$-flow-preconditioned minor distributed to $\overline G$. 
By Corollary~\ref{cor:shortest_path}, the shortest path can be computed in $\widetilde O(D + \overline n^{1/2}D^{1/4})$ rounds. 
The total number of rounds for lines \ref{line:mcf_12}-\ref{line:mcf_16} is $\widetilde O(m^{3/7}(D + \overline n^{1/2} D^{1/4}))$ rounds.
 
Hence, the total number of rounds required is 
\[\begin{split} & \widetilde O(\overline m^{3/7+o(1)}(\overline n^{1/2} + D) \poly (\log W)) + \widetilde O(\overline m^{1/2 } + D) + \widetilde O(\overline m^{3/7}(D + \overline n^{1/2} D^{1/4}))\\
= &  \widetilde O(\overline m^{3/7+o(1)} (\overline n^{1/2}D^{1/4} + D)\poly(\log W) ).\end{split}\]
\end{proof}

\subsection{Negative shortest path}
We now give a distributed algorithm for computing
single source shortest path with negative weights.
It is a direct use of the reduction from shortest paths with negative weights
to min-cost flow by Cohen et al.~\cite{CMSV17}.
Pseudocode of this algorithm is in Algorithm~\ref{shortestpaths}.

\begin{algorithm}
\caption{\ShortestPaths ($\overline{G}$, $s$, $W$)}
\label{shortestpaths}
\KwIn{directed graph ${
G}=(V,E,w)$ with $|V|={n}$ and $|{E}|={m}$; source $s$; the absolute maximum weight $W$; }
\tcc{Reduction to a weighted perfect $\vec{1}$-matching problem}
Let the bipartite graph be $G_{12}=(V_{1}\cup V_2, E_{12}, w_{12})$ with $V_{1}=\{v_{1}\mid v\in V\}$, $V_2=\{V_2\mid v\in V\}$, $E_{12}=\{u_{1}v_2\mid uv\in E\}\cup\{v_{1}v_2\mid v\in V\}$ and $w_{12}(u_{1}v_2)=\left\{\begin{array}{ll}
     -w_{uv}& uv\in E  \\
     0&u=v
\end{array}\right.$\;
$(\vec{f},\vec{y})\leftarrow\MinCostFlow (G_{12}, \vec{1}, W)$\;
\For{each edge $(u,v)\in\overline{E}$}{
$w_{uv}'\leftarrow w_{uv}+y_{u}-y_{v}$\;
}
Compute the shortest paths with source $s$ on $\overline{G}=(\overline{V}, \overline{E}, w')$\;
\end{algorithm}

\begin{theorem}
\label{thm:shortest-paths}
	Let $\overline{G} = (\overline V, \overline E)$ be a communication network with $\overline{n}$ vertices and $\overline{m}$ edges, 
	and $w: \overline E\rightarrow \{-W, -W+1, \dots, -1, 0, 1, \dots, W\}$ be an integral weight function. 
	For a vertex $s \in \overline V$,  there is a distributed algorithm to compute 
	the shortest path from $s$ to all the other vertices that has a shortest path from $s$ in 
	\[\widetilde O(\overline m^{3/7+o(1)} (\overline n^{1/2}D^{1/4} + D)\poly(\log W) )\] rounds in the \textsf{CONGEST} model.
\end{theorem}
\begin{proof}
The strategy of Algorithm \ref{shortestpaths} is that transferring the given $G$ to a bipartite graph $G_{12}$ and calling the Algorithm \MinCostFlow on $G_{12}$ to obtain the dual solution $\vec{y}$, which is utilized to transform the original edge weights in ${G}$ to be non-negative, and then using the single source shortest path algorithm with non-negative weights on the new instance ${G}'=({V},{E},w$) to compute the shortest paths.

In a constant number of rounds, we can construct $G_{12}$ based on $\overline G$
that is a $(1, 0)$-flow-preconditioned minor distributed to $\overline G$. 
By Theorem~\ref{thm:min-cost-flow} and Corollary~\ref{cor:shortest_path}, Algorithm \ref{shortestpaths} can be simulated in the desired number of rounds.
\end{proof}

\section*{Acknowledgements}
\label{sec:acknowledgements}
Sebastian Forster is supported by the Austrian Science Fund (FWF): P 32863-N.
Yang P. Liu was supported by the Department of Defense (DoD) through
the National Defense Science and Engineering Graduate Fellowship (NDSEG) Program.
Richard Peng is supported by the National Science Foundation (NSF)
under Grant No. 1846218.
Xiaorui Sun is supported by start-up funds from University of Illinois at Chicago. 

\bibliographystyle{alpha}
\bibliography{refs}

\appendix

\section{Lower Bound}
\label{sec:LowerBound}

\LowerBound*

\begin{proof}
	Note first that any low-accuracy solver with accuracy $ \epsilon \leq \tfrac{1}{2} $ can be boosted to a solver with precision $ \tilde{\epsilon} $ at the cost of increasing the running time by a factor of $ O (\log \tilde{\epsilon}^{-1}) $ (see Lemma~1.6.8. in~\cite{P13:thesis}).
	This is done by running $ O (\log \tilde{\epsilon}^{-1}) $ iterations of the iterative refinement method, which in each iteration performs one matrix-vector product, one vector substraction, one vector addition, and one call to the low-accuracy solver.
	Excluding the call to the low-accuracy solver, all of these operations can be performed in a constant number of rounds.
	The running time of the iterative refinement method is therefore dominated by the $ O (\log \tilde{\epsilon}^{-1}) $ calls to the low-accuracy solver.
	We thus assume in the following that we are given a high-accuracy solver with accuracy $ \epsilon = 1/\poly(n) $, which as just argued requires only an overhead of $ O (\log n) $ in the number of rounds compared to a low-accuracy solver.
	
	To prove the lower bound we use the framework of Das Sarma et.\ 
	al.~\cite{SHKKNPPW12} who established an $ \Omega (\sqrt{\overline{n}}/(\log 
	\overline{n}) + D) $ lower bound for the following verification problem: Given a 
	subgraph $ H $ of the communication network $ \overline{G} $ (which has $ 
	\overline{n} $ nodes and diameter $ D $) and two nodes $ s $ and $ t $, the 
	network needs to decide whether $ s $ and $ t $ are connected (i.e., lie in the 
	same connected component).
	In their lower bound construction, the distance between $ s $ and $ t $ in $ 
	\overline{G} $ is $ O (\log n) $.
	We show that any algorithm for solving Laplacian Systems up to small enough error 
	$ \epsilon = 1/\poly(n) $ can be used to give $ s $ and $ t $ the information 
	whether they are connected in $ H $ in additional $ 
	\operatorname{dist}_{\overline{G}} (s, t) = O (\log n) $ rounds.
	In particular, we exploit that such a solver can be used to compute an 
	approximation to the effective $s$-$t$ resistance.
	
	Define the weighted graph $ H' $ (which we view as a resistor network) as having 
	the same nodes and edges as $ \overline{G} $ and resistances $ r_e = 1 $ for every 
	edge $ e \in E (H) $ and $ r_e = n $ for every edge $ e \notin E (H) $.
	Let $ \chi_{s,t} $ be the $n$-dimensional vector is $ 1 $ at the coordinate 
	corresponding to $ s $, $ -1 $ at the coordinate corresponding to $ t $, and $ 0 $ 
	otherwise.
	It is well known~(see, e.g., \cite{ChristianoKMST11, Vishnoi12}) that $ \res_{H'} 
	(s, t) = \phi (s) - \phi (t) $ for any vector $ \phi $ satisfying $ \mL(H') \phi = 
	\chi_{s,t} $, where $ \phi (s) $ and $ \phi (t) $ are the values of the 
	coordinates corresponding to $ s $ and $ t $, respectively.
	Let $ \vec{\phi}' $ be an approximate solution with error $ \epsilon $ to the 
	linear system $\mL(H')\vec{\phi} = \vec{\chi}_{s,t} $, i.e.,
	\begin{equation}
	\norm{\vec{\phi'} - \mL\left( H' \right)^{\dag} \vec{\chi}_{s,t}}_{\mL\left( H' 
	\right)}
	\leq
	\epsilon \cdot \norm{\vec{\chi}_{s,t}}_{\mL\left( H' \right)^{\dag}} \, . 
	\label{eq:lower bound system approximation}
	\end{equation}
	Henceforth let $ \phi = \mL\left( H' \right)^{\dag} \vec{\chi}_{s,t} $.
	Thus, \eqref{eq:lower bound system approximation} is equivalent to the statement
	\begin{equation}
	\norm{\vec{\phi'} - \phi}_{\mL\left( H' \right)}
	\leq
	\epsilon \cdot \sqrt{\res_{H'} (s, t)} \, . \label{lower bound bound with 
	effective resistance}
	\end{equation}
	
	It is well-known that for any Laplacian matrix $\mL$ with integer resistances from 
	$ 1 $ to $ R = \poly (n) $ this first eigenvalue is $ \lambda_1 (\mL) = 0 $ the 
	second eigenvalue is bounded by $ \lambda_2 (\mL) \geq \tfrac{1}{\poly(n)} $ and 
	thus for any vector $ \vec{x} $ the following bounds relating the matrix norm to 
	the infinity norm hold:
	\begin{equation}
	\norm{\vec{x}}_\infty
	\leq \norm{\vec{x}}_2
	\leq \frac{1}{\sqrt{\norm{\mL^{\dag}}_2}} \norm{\vec{x}}_{\mL}
	= \frac{1}{\lambda_2 (\mL)} \norm{\vec{x}}_{\mL}
	\leq \poly (n) \norm{\vec{x}}_{\mL} \label{eq:lower bound bounds on norms}
	\end{equation}
	The combination of \eqref{lower bound bound with effective resistance} and 
	\eqref{eq:lower bound bounds on norms} together with the estimate $ \res_{H'} (s, 
	t) \leq n R $ gives
	\begin{equation}
	\norm{\vec{\phi'} - \phi}_\infty \leq \epsilon \cdot \poly (n) \, .\label{eq:lower 
	bound infinity norm}
	\end{equation}
	By setting $ \epsilon $ to a small enough value inversely polynomial in $ n $, 
	\eqref{eq:lower bound infinity norm} implies 
	\begin{equation*}
	\res_{H'} (s, t) - 0.25 \leq \phi' (s) - \phi' (t) \leq \res_{H'} (s, t) + 0.25 \, 
	.
	\end{equation*}

	In the rest of the proof, we argue that knowing the value $ \phi' (s) - \phi' (t) 
	$ (which can be made known to both $ s $ and $ t $ in $ 
	\operatorname{dist}_{\overline{G}} $ rounds) suffices to decide whether $ s $ and 
	$ t $ are connected in $ H $.
	If $ s $ and $ t $ are connected in $ H $, then -- since effective resistances 
	obey the triangle inequality -- the effective $s$-$t$ resistance in $ H' $ is 
	upper-bounded by the length of the shortest path between $ s $ and~$ t $ in $ H $, 
	i.e., $ \res_{H'} (s, t) \leq n - 1 $.
	If $ s $ and $ t $ are not connected in $ H $, then let $ S $ be the connected 
	component containing $ S $ and let $ e_1, \ldots, e_k $ (for some $ k \leq m \leq 
	n^2 $) be the edges leaving $ S $ in $ H $.
	By the Nash-Williams inequality (see, e.g., ~\cite{LyonsP16}, chapter 2.5), the 
	effective $s$-$t$ resistance is at least
	\begin{equation*}
	\res_{H'} (s, t) \geq \frac{1}{\sum_{i=1}^k \frac{1}{r (e_i)}} = 
	\frac{1}{\sum_{i=1}^k \frac{1}{(1 + \epsilon)^2 n^3}} \geq (1 + \epsilon)^2 n \, .
	\end{equation*}
	Thus, $ s $ and $ t $ are connected in $ H $ if and only if $ \phi' (s) - \phi' 
	(t) \leq n - 0.5 $.
\end{proof}

\section{Building Blocks for Distributed Minors}
\label{sec:MinorProofs}

Here we give the deferred proofs from Section~\ref{subsec:Communication}.
All of our algorithms are based on placing about $\sqrt{\overline{n}}$
special vertices, picked randomly so that any vertex is within a distance
of about $\sqrt{\overline{n}}$ from these, and then aggregating information
at these special vertices globally via a DFS tree in about
$\sqrt{\overline{n}} + D$ rounds.

For this, it is useful to define a tree decomposition scheme for the set
of overlapping trees used to connect the supervertices.

\begin{lemma}
\label{lemma:SpecialVertices}
There is an algorithm \SpecialVertices that takes an input a forest $F$
specified with mappings of vertices and edges into a graph $\overline{G}$
such that each vertex of $\overline{G}$ appears in at most $\rho$ trees
of $F$, and each edge of $\overline{G}$ is used in at most $\rho$
trees of $F$, and returns after
$O( \rho \sqrt{\overline{n}})$ rounds of communication
a collection of $O( \rho \sqrt{\overline{n}} \log{n})$
special vertices of $F$, labeled at their mapped vertices in $\overline{G}$,
such that with high probability,
for each $T$ in $F$, we have:
\begin{enumerate}
\item either the diameter of $T$ is at most $\sqrt{\overline{n}}$,
\item or for any vertex of $T$,
\begin{enumerate}
\item it can reach at most $2$ special vertices, without going through
more special vertices.
\item its distance in $T$ to closest special vertex is at most
$O(\sqrt{\overline{n}})$.
\end{enumerate}
\end{enumerate}
\end{lemma}

\begin{algorithm}[ht]
\caption{Partition all trees of a forest into low diameter pieces
via special vertices \label{algorithm:TreePartition}}
\SetKwProg{Proc}{procedure}{}{}
\Proc{\SpecialVertices ($F$, $\overline{G}$)}{
  Sample $K$ by including each vertex in each tree of $F$
   with probability $\log \overline n / \overline n^{1/2}$.\\
  \For{$O( \sqrt{ \overline{n}})$ rounds}{
     Each vertex propagate to all its neighbors whether
     taking that edge towards it leads to a vertex in $K$.
  }
  Add all vertices that can reach special vertices in three or
  more directions to $K$.\\
}
\end{algorithm}

\begin{proof}
Consider \SpecialVertices in Algorithm~\ref{algorithm:TreePartition}.

The congestion bound gives that the total number of vertices
among the trees is at most $O(\rho \overline{n})$.
This means picking $O(\rho \sqrt{ \overline{n} }\log{\overline n})$ random vertices
from
\[
\bigcup_{v^{G} \in V\left( G \right)}
S^{G \rightarrow \overline{G}}(v^{G}),
\]
ensures that with high probability, the maximum distance to a special
vertex along any tree path is $O( \sqrt{ \overline{n}})$ with high probability.

This means that in $O( \sqrt{ \overline{n}})$ rounds of propagation,
we can find, for each edge in each tree,
whether there is a special vertex in either direction.
Formally, the local operation at each vertex is to check whether
there are 2 or more directions from it that lead to special vertices:
if there are, then all edges entering this vertex are part of paths that
reach special vertices.
Otherwise, all except that one direction that the path from a special vertex
came from can be continued, so we push `possible' along all except that direction.

By declaring all vertices with three or more edges leaving it that lead
to special vertices as special themselves.
As this only adds in the lowest common ancestors of the previous special
vertices, it only increases the number of special vertices by a constant factor.
This step is also completely local, so the total cost is dominated by the
propagation steps.
\end{proof}

This partition routing allows us to root all of the trees.

\begin{proof}(of Lemma~\ref{lemma:RootTrees})
  We first run the partition scheme \SpecialVertices on the forest that's
  the unions of the spanning trees of the supernodes of $G$.
  Lemma~\ref{lemma:SpecialVertices} gives that each resulting piece consisting of
  edges reachable to each other without going through special vertices have
  diameter at most $O( \sqrt{\overline{n}})$.
  So $O( \sqrt{\overline{n}})$ rounds of propagation lets the root vertex inform
  all nodes in its piece.
  In this number of rounds, we also propagate the ID of these special vertices,
  as well as distances to them,
  to all vertices in each piece using another $O( \sqrt{ \overline{n}})$ rounds
  of communication.

  Note this in particular allows the two special nodes on each piece to know
  their distance to each other.
  The amount of information aggregated at each special vertex may be large.
  Aggregating these information centrally along the BFS tree of $\overline{G}$ then
  allows us to find the distance from the root vertex to all the special vertices
  in their piece in a centralized manner.
  Once this information is propagated back, each edge can just be oriented towards
  the direction of the special vertex closer to the root, giving the desired
  orientations.
\end{proof}

This direction to the root is necessary for propagating information
that cannot be duplicated, such as the sum of values.
Using it, we can prove our main communication tool, Lemma~\ref{lemma:Communication}.

\begin{proof} (of Lemma~\ref{lemma:Communication})
  Once again we run the partition scheme \SpecialVertices.
  Lemma~\ref{lemma:SpecialVertices} ensures that all paths hit one such vertex after at
  most $O( \sqrt{ \overline{n}})$ steps.
  After that, we root the tree using Lemma~\ref{lemma:RootTrees}.

  For the push case, we repeatedly push information from vertices
  to their neighbors.
  Each such step costs $O(\rho)$ due to the vertex congestion bound.
  By the bound above, after $O( \sqrt{ \overline{n}})$  such steps
  (which costs a total of $O( \rho \sqrt{\overline{n}})$ rounds,
  this information either reached all nodes in the corresponding
  supernode, or some special vertex.

  Getting all info on special nodes to a center node
  (along with the ID of $v^{G}$ that they originated from)
  over a BFS tree takes
  $O(\rho \sqrt{ \overline{n} }\log{\overline n} + D)$ rounds,
  after which they can also be re-distributed to all special nodes.
  Then by the distance bound from Lemma~\ref{lemma:SpecialVertices}
 another $O( \rho \sqrt{\overline{n}})$ rounds of propagations
  to neighbors passes the information to everyone.

  The aggregation of sum or minimum follows similarly.
  We repeat $O( \rho \sqrt{\overline{n}})$ rounds of all non-special vertices
  propagating their sum, or min, up to their parents.
  In case of sum, once a value `floats' to its parent, it's set to $0$ at the
  current vertex so we do not over count.
  Finally, all the information at the special vertices are aggregated via the
  global BFS tree, and passed to the corresponding root vertices.
\end{proof}

\begin{proof} (of Lemma~\ref{lemma:MinorCompose})
Let
\begin{align*}
S^{G_1 \rightarrow G_2}
&
: V\left( G_1 \right) \rightarrow V\left( G_2 \right)^{*}\\
V_{map}^{G_1 \rightarrow G_2}
&
: V\left( G_1 \right) \rightarrow V\left( G_2 \right)\\
T^{G_1 \rightarrow G_2}
&
: V\left( G_1 \right) \rightarrow E\left( G_2 \right)^{*}\\
E_{map} ^{G_1 \rightarrow G_2}
&
: E\left( G_1 \right) \rightarrow E\left( G_2 \right)
\end{align*}
be the maps from $G_1$ to $G_2$, and similarly,
let the mapping from $G_2$ to $\overline{G}$ be:
\begin{align*}
S^{G_2 \rightarrow \overline{G}}
&
: V\left( G_2 \right) \rightarrow V\left( \overline{G} \right)^{*}\\
V_{map}^{G_2 \rightarrow \overline{G}}
&
: V\left( G_2 \right) \rightarrow V\left( \overline{G} \right)\\
T^{G_2 \rightarrow \overline{G}}
&
: V\left( G_2 \right) \rightarrow E\left( \overline{G} \right)^{*}\\
E_{map}^{G_2 \rightarrow \overline{G}}
&
: E\left( G_2 \right) \rightarrow E\left( \overline{G} \right)
\end{align*}

We will construct the mapping from $G$ to $\overline{G}$ from these.
The edge and vertex maps are directly by composition:
\begin{align*}
V_{map}^{G_1 \rightarrow \overline{G}}\left( v^{G_1} \right)
& =
V_{map}^{G_2 \rightarrow \overline{G}} \left(
V_{map}^{G_1 \rightarrow G_2} \left( v^{G_1} \right) \right)\\
E_{map}^{G_1 \rightarrow \overline{G}}\left( e^{G_1} \right)
& =
E_{map}^{G_2 \rightarrow \overline{G}} \left(
E_{map}^{G_1 \rightarrow G_2} \left( e^{G_1} \right) \right)
\end{align*}
The edge mapping is a direct transfer of the pre-images, locally per edge.
While the vertex label propagation is via one round of communication
along new supervertex, via Lemma~\ref{lemma:Communication}.

So we can focus on the construction of new supervertices and their spanning trees.
For $v^{G_1} \in V(G_1)$, we let
\[
S^{G_1 \rightarrow \overline{G}} \left( v^{G_1} \right)
=
\bigcup_{v^{G_2} \in S^{G_1 \rightarrow G_2}\left( v^{G_1} \right)}
S_{G_2 \rightarrow \overline{G}}\left( v^{G_2}\right)
\]
with corresponding spanning tree a subset of the edges
\[
\bigcup_{v^{G_2} \in S^{G_1 \rightarrow G_2}\left( v^{G_1} \right)}
T^{G_2 \rightarrow \overline{G}}\left( v^{G_2}\right).
\]

To compute this union we have each root vertex of each supervertex corresponding
to some $v^{G_2}$ inform the entire supervertex of the new ID in $G_1$.
As each vertex of $G_2$ corresponds to the image of at most $\rho_1$ $v^{G_1}$s,
this mapping takes
$O(\rho_1 \rho_2 ( \sqrt{\overline{n}} \log{n} + D ))$ iterations.

Then the edges of $\overline{G}$ with the new labels (of vertex ID from $G_1$)
gives the corresponding supervertices.
That is, $S^{G_1 \rightarrow \overline{G}}(v^{G_1})$ is simply the set of vertices
that received the label $v^{G_1}$ after we propagated from the root vertices
of $G_2$ in $\overline{G}$.

We then need to find spanning forests among these unions of trees.
We do so with a variant of Brouvka's algorithm, combined with parallel
tree contraction.
Specifically, we iterate the following $O(\log{n})$ times:
\begin{enumerate}
\item Each remaining vertex (in each of the super vertices) pick a random priority.
\item Each vertex identify highest priority neighbor, computed
on the minor with some edges already contracted using Lemma~\ref{lemma:Communication}.
\item Contract either:
\begin{enumerate}
\item all the leaf vertices (degree $1$ nodes),
\item a vertex disjoint subset of edges with both end points having degree $2$.
\end{enumerate}
\end{enumerate}
In the second case of contracting edges with both endpoints having degree $2$,
we can find an independent set of such edges whose total size is at least a
constant factor of all such edges by simply picking a random subset with
probability $1 / 5$, and dropping the ones where another end point is picked.

This procedure reduces the number of vertices in each component a constant factor
in expectation.
This is because the first step ensures the edges found form a tree
with edge count at least half the number of vertices.
Then parallel tree contraction~\cite{MR89} ensures that either the number
of leaves, or the number of edges with both end points degree $2$ is at
least a constant factor of the tree size.
Contracting the larger set of these then gives
the desired constant factor progress, so it terminates in $O(\log{n})$ rounds.

Note that in subsequent rounds, the edges already identified to be
part of $T^{G_1 \rightarrow \overline{G}}(v^{G})$ form a spanning forest,
and we're working on the minor with this forest contracted.
This means we need to invoke Lemma~\ref{lemma:Communication} repeatedly
to do the neighborhood aggregations on this contracted graph.
\end{proof}

\begin{proof} (of Corollary~\ref{corollary:ActualMinor})
We want to simplify $F$ into a sequence of subgraphs that
have simple $1$-minor distributions into $G$.
After that, we can invoke Lemma~\ref{lemma:MinorCompose} to make progress.

We repeat the same tree contraction procedure used for finding
spanning forests in Lemma~\ref{lemma:MinorCompose} above.
It gives that at each step, we're computing $G / F$ for a set
of vertex disjoint stars $F$.

Given such a $F$, we can then construct a $1$-minor distribution of
$G / F$ into $G$ by:
\begin{enumerate}
\item Having the center vertex generate the new vertex ID.
\item Propagate this ID to vertices in the corresponding supervertex
in $\overline{G}$.
\item Have edges that declared themselves part of $F$ pass this info
from one end point to the other.
\item All leaf supervertices then pull this new vertex ID into their
root vertex.
\end{enumerate}
after which invoking Lemma~\ref{lemma:MinorCompose} gives the minor.
This halves the number of
\end{proof}

\section{Max flow and Minimum Cost Flow Algorithm}\label{appendix:flow}
In this section, we give the missing details of the max flow algorithm and minimum cost flow algorithm as Section~\ref{sec:Implications}.
\subsection{Max Flow Algorithm}\label{sec:maxflow_algo}
In this section, we give the missing subroutines of Algorithm~\ref{alg:madry}. The subroutines \Augmentation, \Fixing and \Boosting are shown in Algorithm \ref{augmentation}, \ref{fixing} and \ref{boosting} respectively.

\begin{algorithm}[H]
\KwIn{directed graph $G_0=(V, E_0,\vec{u})$ with each $e\in E_0$ having two non-negative integer capacities $u_{e}^{-}$ and $u_{e}^{+}$; $|V|={n}$ and $|E_0|={m}$; source $s$ and sink $t$; the largest integer capacity $U$; target flow value $F\ge 0$;}
\caption{\MaxFlow($G_0$, $s$, $t$, $U$, $F$)}\label{alg:madry}
\tcc{Preconditioning Edges}
Add $m$ undirected edges $(t, s)$ with forward and backward capacities $2U$ to $G_0$\;
\tcc{Initialization}
\For{each $e=(u,v)\in E_0$}{
replace $e$ by three undirected edges $(u,v)$, $(s,v)$ and $(u,t)$ whose capacities are $u_{e}$\;}
Let the new graph be $G=(V,E)$\;
Initialize $\vec{f}\leftarrow\vec{0}$ and $\vec{y}\leftarrow\vec{0}$\;
\tcc{Progress Step}
$\vec{\widetilde{f}}, \vec{\widehat{f}}, \vec{\widehat{y}}\leftarrow\Augmentation(G, s, t, F)$\;
Compute $\vec \rho$ by letting $\rho_{e}\leftarrow\frac{\widetilde{f}_{e}}{\min\{u_{e}^{+}-f_{e}, u_{e}^{-}+f_{e}\}}$\;
$\vec{f}, \vec{y}\leftarrow\Fixing\left(G, \vec{\widehat{f}}, \vec{\widehat{y}}\right)$\;
$\eta\leftarrow\frac{1}{14}-\frac{1}{7}\log_{m}{U}-O(\log\log(mU))$, 
$\widehat{\delta}\leftarrow\frac{1}{m^{\frac{1}{2}-\eta}}$\;
\For{$t=1$ \KwTo $100\cdot\frac{1}{\widehat{\delta}}\cdot\log{U}$}{
\eIf{$\left\|\vec{\rho}\right\|_{3}\le\frac{m^{\frac{1}{2}-\eta}}{33(1-\alpha)}$}
{$\delta\leftarrow\frac{1}{33(1-\alpha)\left\|\vec{\rho}\right\|_{3}}$\;
$\vec{\widetilde{f}}, \vec{\widehat{f}}, \vec{\widehat{y}}\leftarrow\Augmentation(G, s, t, F)$\;
Compute $\vec \rho$ by letting
$\rho_{e}\leftarrow\frac{\widetilde{f}_{e}}{\min\{u_{e}^{+}-f_{e}, u_{e}^{-}+f_{e}\}}$\;
$\vec{f}, \vec{y}\leftarrow\Fixing\left(G, \vec{\widehat{f}}, \vec{\widehat{y}}\right)$\;
}
{let $S^{*}$ be the edge set that contains the $m^{4\eta}$ edges with the largest $|\rho_{e}|$\;
$G\leftarrow\Boosting\left(G, S^{*}, U, \vec{f}, \vec{y}\right)$\;
}
}
$\vec{f}\leftarrow\FlowRounding(G,\vec{f}, s, t)$\;
\While{there is an augmenting path from $s$ to $t$ with respect to $\vec f$ for $G$} {
augment an augmenting path for $\vec f$\;
}
\end{algorithm}
\newpage

\begin{algorithm}[H]
\caption{\Augmentation($G$, $s$, $t$, $F$)}
\label{augmentation}
For each $e\in E$, let $r_{e}\leftarrow\frac{1}{(u_{e}^{+}-f_{e})^{2}}+\frac{1}{(u_{e}^{-}+f_{e})^{2}}$ and $w_{e}\leftarrow\frac{1}{r_{e}}$\;
Solve Laplacian linear system $\mL(G)\vec{\widetilde{\phi}}=F\cdot\vec{\chi}_{s,t}$ where $\vec{\chi}_{s,t}$ is the vector whose entry is $-1$ (resp. $1$) at vertex $s$ (resp. $t$) and $0$ otherwise\;
For each $e=(u,v)\in E$, let $\widetilde{f}_{e}\leftarrow\frac{\widetilde{\phi}_{v}-\widetilde{\phi}_{u}}{r_{e}}$ and $\widehat{f}_{e}\leftarrow f_{e}+\delta\widetilde{f}_{e}$\;
For each $v\in V$, let $\widehat{y}_{v}\leftarrow y_{v}+\delta\widetilde{\phi}_{v}$\;
\Return $\vec{\widetilde{f}}$, $\vec{\widehat{f}}$, $\vec{\widehat{y}}$\;
\end{algorithm}

\begin{algorithm}[H]
\caption{\Fixing$\left(G, \vec{\widehat{f}}, \vec{\widehat{y}}\right)$}
\label{fixing}
For each $e=(u,v)\in E$, let $r_{e}\leftarrow\frac{1}{\left(u_{e}^{+}-\widehat{f}_{e}\right)^{2}}+\frac{1}{\left(u_{e}^{-}+\widehat{f}_{e}\right)^{2}}$, $w_{e}\leftarrow\frac{1}{r_{e}}$ and
$\theta_{e}\leftarrow w_{e}\left[\left(\widehat{y}_{v}-\widehat{y}_{u}\right)-\left(\frac{1}{u_{e}^{+}-\widehat{f}_{e}}-\frac{1}{u_{e}^{-}-\widehat{f}_{e}}\right)\right]$\;
$\vec{f'}\leftarrow\vec{\widehat{f}}+\vec{\theta}$\;
Let $\vec{\widehat{\delta}}$ be $\vec{\theta}$'s residue vector\;
For each $e\in E$, let $r_{e}\leftarrow\frac{1}{(u_{e}^{+}-f'_{e})^{2}}+\frac{1}{(u_{e}^{-}+f'_{e})^{2}}$ and $w_{e}\leftarrow\frac{1}{r_{e}}$\;
Solve Laplacian linear system $\mL(G)\vec{\phi'}=-\vec{\widehat{\delta}}$\;
For each $e=(u,v)\in E$, let $\theta_{e}'\leftarrow\frac{\phi_{v}'-\phi_{u}'}{r_{e}}$\;
For each $e\in E$, $f_{e}\leftarrow f'_{e}+\theta_{e}'$\;
For each vertex $v\in V$, $y_{v}\leftarrow \widehat{y}_{v}+\phi_{v}'$\;
\Return $\vec{f}$, $\vec{y}$\;
\end{algorithm}

\begin{algorithm}[H]
\caption{\Boosting$\left(G, S^{*}, U, \vec{f}, \vec{y}\right)$}
\label{boosting}
\For{each edge $e=(u,v)\in S^{*}$}{
$\beta(e)\leftarrow 2+\lceil{\frac{2U}{\min\{u_{e}^{+}-f_{e}, u_{e}^{-}+f_{e}\}}}\rceil$\;
replace $e$ with path $u\leadsto v$ that consists of $\beta(e)$ edges $e_{1}, \cdots, e_{\beta(e)}$ oriented towards $v$ and $\beta(e)+1$ vertices $v_{0}=u, v_{1},\cdots, v_{\beta(e)-1}, v_{\beta(e)}=v$\;
$e_{1}, e_{2}\leftarrow e$\;
for $3\le i\le\beta(e)$, let $u_{e_{i}}^{+}\leftarrow +\infty$ and $u_{e_{i}}^{-}\leftarrow\left(\frac{1}{u_{e}^{+}-f_{e}}-\frac{1}{u_{e}^{-}+f_{e}}\right)^{-1}(\beta(e)-2)-f_{e}$\;
for each $1\le i\le\beta(e)$, let $f_{e_{i}}\leftarrow f_{e}$\;
$y_{v_{0}}\leftarrow y_{u}$\;
$y_{v_{\beta(e)}}\leftarrow y_{v}$\;
$y_{v_{1}}\leftarrow y_{v}$\;
$y_{v_{2}}\leftarrow y_{v}+\frac{1}{u_{e}^{+}-f_{e}}-\frac{1}{u_{e}^{-}+f_{e}}$\;
    for $3\le i\le\beta(e)$, set $y_{v_{3}},\cdots,y_{v_{\beta(e)-1}}$ such that $y_{v_{i}}-y_{v_{i-1}}=-\frac{1}{\beta(e)-2}\left(\frac{1}{u_{e}^{+}-f_{e}}-\frac{1}{u_{e}^{-}+f_{e}}\right)$\;
}
Update $G$\;
\end{algorithm}
\subsection{Unit Capacity Minimum Cost Flow Algorithm}
\label{sec:min-cost-flow-algo}
We give the detailed unit capacity minimum cost flow algorithm (see Algorithm \ref{algo:min-cost-flow-final}) proposed by Cohen et al.~\cite{CMSV17} in this subsection. The subroutines \Initialization, \Perturbation, \Progress and \Repairing are shown in Algorithm \ref{algo:initializationbipartitegraph_app}, \ref{perturbation}, \ref{progress} and \ref{repairing} respectively.

\begin{algorithm}[H]
\KwIn{directed graph ${G_0}=({V_0},{E_0},\vec{c}_0)$ with each edge having unit capacity and cost $\vec{c}_0$; $|{V_0}|={n}$ and $|{E_0}|={m}$; integral demand vector $\vec{\sigma}$; the absolute maximum cost $W$;}
\caption{\MinCostFlow($G$, $\vec{\sigma}$, $W$)}
\label{algo:min-cost-flow-final}
$G = (P\cup Q, E), \vec b, \vec{f}, \vec{y}, \vec{s}, \vec{\nu}, \widehat{\mu}$, $c_{\rho}, c_{T}, \eta\assign \Initialization(G_0, \vec \sigma)$\;
Add a new vertex $v_0$ and undirected edges $(v_0, v)$ for every $v\in P$ to $G$\;
\For{$i=1$ \KwTo $c_{T}\cdot m^{1/2-3\eta}$
}{
\For{each $v\in P$}
{\label{line:mincost17_app}
set resistance of edge $(v_0, v)$ for each $v\in P$ to be $r_{v_{0}v}\assign\frac{m^{1+2\eta}}{a(v)}$, where $a(v)\assign\sum_{u\in Q, e=(v,u)\in E}\nu_e+\nu_{\overline{e}}$; \Comment{$\overline{e}=(\overline{v},u)$ is $e$'s partner edge that is the unique edge sharing one common vertex from $Q$.}
\label{line:mincost19_app}
}
\For{$j=1$ \KwTo $m^{2\eta}$}{
\While{$\left\|\vec{\rho}\right\|_{\vec{\nu},3}>c_{\rho}\cdot m^{1/2-\eta}$}{
	$\vec{\rho}, \vec{y}, \vec{s}, \vec{\nu}\leftarrow \Perturbation(G, \vec{\rho}, \vec{f}, \vec{y}, \vec{s}, \vec{\nu})$\;}
	$\vec{f}, \vec{s}, \vec{\rho}, \widehat{\mu} \leftarrow \Progress(G, \vec{\sigma}, \vec{f}, \vec{\nu})$\;
}}
\Repairing($G$, $\vec{f}$, $\vec{y}$)\;
\end{algorithm}

\newpage
\begin{algorithm}[H]
\caption{\Initialization($G$, $\vec{\sigma}$) \label{algo:initializationbipartitegraph_app}}
Create a new vertex $v_{aux}$ with $\sigma(v_{aux})=0$\;
\For{each $v\in{V_0}$}{
$t(v)\leftarrow\sigma(v)+\frac{1}{2}\mathrm{deg}_{in}^{{G_0}}(v)-\frac{1}{2}\mathrm{deg}_{out}^{{G_0}}(v)$\;
\lIf{$t(v)>0$}{
construct $2t(v)$ parallel edges $(v, v_{aux})$ with costs $\left\|\vec{{c}}_0\right\|_1$
}
\uElseIf{$t(v)<0$}{
construct $|2t(v)|$ parallel edges $(v_{aux}, v)$ with costs $\left\|\vec{{c}}_0\right\|_1$\;
}}
Let the new graph be $G_1=(V_1, E_1, \vec{c}_1)$\;
Initialize the bipartite graph $G=(P\cup Q, E, \vec{c})$ with $E\leftarrow\emptyset
$, $P\leftarrow V_1$ and $Q\leftarrow\{e_{uv}\mid(u,v)\in E_1\}$ where $e_{uv}$ is a vertex corresponding to edge $(u,v)\in E_1$\;

\For {each $(u,v)\in E_1$} {
let $E\leftarrow E\cup\{(u, e_{uv}), (v, e_{uv})\}$ with $c(u, e_{uv})=c_1(u, v)$ and $c(v, e_{uv})=0$, and set $b(u)\leftarrow\sigma(u)+\mathrm{deg}_{in}^{{G_1}}(u)$, $b(v)\leftarrow\sigma(v)+\mathrm{deg}_{in}^{{G_1}}(v)$ and $b(e_{uv})\leftarrow 1$\;
}
For each $v\in P$, set $y_v\leftarrow\left\|\vec{c}\right\|_{\infty}$, and for each $v\notin P$, set $y_v\leftarrow 0$\;
For each $e=(u,v)\in E$, set $f_e\leftarrow\frac{1}{2}$, $s_e\leftarrow c_e+y_u-y_v$ and $\nu_e\leftarrow\frac{s_e}{2\left\|\vec{c}\right\|_{\infty}}$\;
Set $\widehat{\mu}\leftarrow\left\|\vec{c}\right\|_{\infty}$,
$c_{\rho}\leftarrow 400\sqrt{3}\cdot\log^{1/3}{W}$,
$c_{T}\assign 3c_{\rho}\log{W}$ and 
$\eta\leftarrow\frac{1}{14}$\;
\Return $G$, $\vec b$, $\vec{f}$, $\vec{y}$, $\vec{s}$, $\vec{\nu}$, $\widehat{\mu}$, $c_{\rho}$, $c_{T}$ and $\eta$\;
\end{algorithm}

\begin{algorithm}[H]
\caption{\Perturbation($G$, $\vec{\rho}$, $\vec{f}$, $\vec{y}$, $\vec{s}$, $\vec{\nu}$)}
\label{perturbation}
\For{each $v\in Q$}{
let $e=(u,v)$ and $\overline{e}=(\overline{u}, v)$\;
$y_v\assign y_v-s_e$\;
$\nu_e\assign 2\nu_e$\;
$\nu_{\overline{e}}\assign\nu_{\overline{e}}+\frac{\nu_{e}f_{\overline{e}}}{f_e}$\;
}
\end{algorithm}

\begin{algorithm}[H]
\caption{\Progress($G$, $\vec{\sigma}$, $\vec{f}$, $\vec{\nu}$)}
\label{progress}
For each $e\in E$, let $r_{e}\leftarrow \frac{\nu_{e}}{f_{e}^{2}}$\;
Solve Laplacian linear system $\mL(G)\vec{\widehat{\phi}}=\vec{\sigma}$\;
For each $e=(u,v)\in E$, let $\widehat{f}_{e}\leftarrow\frac{\widehat{\phi}_{v}-\widehat{\phi}_{u}}{r_{e}}$
and $\rho_{e}\leftarrow\frac{|\widehat{f}_{e}|}{f_{e}}$\;
$\delta\leftarrow\min\left\{\frac{1}{8\left\|\vec{\rho}\right\|_{\vec{\nu},4}},\frac{1}{8}\right\}$\;
Update $f_{e}'\leftarrow(1-\delta)f_{e}+\delta\widehat{f}_{e}$
and
$s_{e}'\leftarrow s_{e}-\frac{\delta}{1-\delta}(\widehat{\phi}_{v}-\widehat{\phi}_{u})$\;
For each $e\in E$, let $f_{e}^{\#}\leftarrow\frac{(1-\delta )f_{e}s_{e}}{s_{e}'}$\;
Obtain the flow vector $\vec{\sigma'}$ corresponding to the residue $\vec{f'}-\vec{f}^{\#}$\;
For each $e\in E$, let $r_{e}\leftarrow\frac{s_{e}'^{2}}{(1-\delta)f_{e}s_{e}}$\;
Solve Laplacian linear system $\mL(G)\vec{\widetilde{\phi}}=\vec{\sigma'}$\;
For each $e=(u,v)\in E$, let $\widetilde{f}_{e}\leftarrow\frac{\widetilde{\phi}_{v}-\widetilde{\phi}_{u}}{r_{e}}$\;
Update $f_{e}\leftarrow f_{e}^{\#}+\widetilde{f}_{e}$ and $s_{e}\leftarrow s_{e}'-\frac{s_{e}'\widetilde{f}_{e}}{f_{e}^{\#}}$\;
\end{algorithm}

\newpage
\begin{algorithm}
\caption{\Repairing($G$, $\vec{f}$, $\vec{y}$)}
\label{repairing}
Let $\vec{b}^+$ be the demand vector corresponding to the current flow $\vec{f}$\;
For each $v\in P\cup Q$, set $b_{v}^{\le}\leftarrow\min(b_{v},b_{v}^{+})$\;
For each $v\in P\cup Q$, if $f(E(v))>b_{v}^{\le}$, set $\vec f$ on $E(v)$ such that $f(E(v))= b_{v}^{\le}$, and let the resulting vector be $\vec{f}^{\le}$\;
Add source $s$ and sink $t$ to $G$, and 
connect $s$ to each $v\in P$ with $f^{\le}_{sv}\leftarrow f^{\le}(E(v))$, and connect each $v\in Q$ to $t$ with $f^\le_{vt}\leftarrow f^{\le}(E(v))$ in $G$\;
$\vec M\leftarrow\FlowRounding(G, \vec{f}^{\le}, s, t)$\;
Remove $s, t$ and related coordinates on $\vec f^\le$ and $\vec M$ from $G$\;
\For{$i=1$ to $\widetilde{O}(m^{3/7})$}{
construct graph $\overrightarrow{G}_{M}=(P\cup Q, E_M, \widetilde{c}_M)$ using $G$, $\vec{M}$ and $\vec{\widetilde{c}}$ such that for each $e=(u,v)\in E$, $\widetilde{c}_e=c_e-y_u-y_v$ and 
$E_M=\{(u, v)\in E\mid u\in P, v\in Q\}\cup\{(u,v)\mid u\in Q, v\in P, M_{uv}\neq 0\}$, $\widetilde{c}_M(u,v)=\left\{\begin{array}{ll}
     \widetilde{c}_{uv},&u\in P, v\in Q  \\
     -\widetilde{c}_{uv},&u\in Q, v\in P 
\end{array}\right.$\;
set $F_{M}\leftarrow\{v\in P\cup Q\mid M(v)<b_{v}\}$\;
compute a shortest path $\pi$ in $\overrightarrow{G}_{M}$ from $P\cap F_{M}$ to $Q\cap F_{M}$\;
\tcc{$\mathcal{D}_{\overrightarrow{G}_{M}}(P,u)$ is the distance from $P$ to $u$ in $\overrightarrow{G}_{M}$}
\tcc{Edges that are reachable in $\overrightarrow{G}_M$ from $P \cap F_M$ have non-negative weights $\widetilde{c}_e$}
\For{$u\in P\cup Q$}{
\If{$u$ can be reached from $P$ in $\overrightarrow{G}_{M}$}{
\eIf{$u\in P$}{$y_{u}\leftarrow y_{u}-\mathcal{D}_{\overrightarrow{G}_{M}}(P,u)$\;
}{$y_{u}\leftarrow y_{u}+\mathcal{D}_{\overrightarrow{G}_{M}}(P,u)$\;}
}}
augment $\vec{M}$ using the augmenting path $\pi$\;
}
\Return $\vec{M}$\;
\end{algorithm}

\end{document}